\documentclass[a4papper,oneside]{book}

\usepackage{graphicx}
\graphicspath{{./gfx/jpeg/}{./gfx/pdf/}{./gfx/pdf/QueueingSimRes/}{./gfx/pdf/QueueingSimRes/timeDepenendetModel/}}
\DeclareGraphicsExtensions{.pdf,.jpeg,.png,.jpg}

\usepackage{float,subfloat,relsize}
\usepackage{amsmath, amsthm, amssymb, amsfonts, dsfont}
\usepackage{caption,subcaption,cite}
\usepackage[final]{pdfpages}
\usepackage{color,upgreek,mathtools,verbatim,stfloats,fixltx2e}
\usepackage{url,hyperref}
\usepackage{inputenc,fancyhdr}
\usepackage[toc,page]{appendix}
\numberwithin{equation}{section}

\usepackage [english]{babel}
\usepackage [autostyle, english = american]{csquotes}
\MakeOuterQuote{"}

\newtheorem{proposition}{Proposition}
\newtheorem{theorem}{Theorem}
\newtheorem{Definition}{Definition}
\newtheorem{lemma}{Lemma}
\newtheorem{remark}{Remark}
\makeatletter
\newcommand{\rmnum}[1]{\romannumeral #1}
\newcommand{\Rmnum}[1]{\expandafter\@slowromancap\romannumeral #1@}

\makeatother

\usepackage{lipsum}
\newcommand\abstractname{Abstract}
\makeatletter
\if@titlepage
  \newenvironment{abstract}{%
      \@beginparpenalty\@lowpenalty
      \begin{center}%
        \bfseries \abstractname
        \@endparpenalty\@M
      \end{center}}%
     {\par\vfil\null\endtitlepage}
\else
  \newenvironment{abstract}{%
      \if@twocolumn
        \section*{\abstractname}%
      \else
        \small
        \begin{center}%
          {\bfseries \abstractname\vspace{-.5em}\vspace{\z@}}%
        \end{center}%
        \quotation
      \fi}
      {\if@twocolumn\else\endquotation\fi}
\fi
\makeatother

\makeatletter
\def\cleardoublepage{\clearpage\if@twoside \ifodd\c@page\else
  \hbox{}
  \thispagestyle{empty}
  \newpage
  \if@twocolumn\hbox{}\newpage\fi\fi\fi}
\makeatother

\begin{document}
\pagestyle{plain}
\pagenumbering{Roman}
\begin{center}
\addcontentsline{toc}{chapter}{Abstract}
\Large Capacity and performance analysis for multi-user system under distributed opportunistic scheduling in a time dependent channel\\[0.5cm]

\large O. Shmuel, A. Cohen, and O. Gurewitz\\[1cm]
\end{center}

\begin{abstract}
    \boldmath
    Consider the problem of a multiple access channel with a large number of users. While several multi-user coding techniques exist, in practical scenarios, not all users can be scheduled simultaneously. Thus, a key problem is which users to schedule in a given time slot. Although the problem has been studied in various time-independent scenarios, capacity scaling laws and algorithms for time-dependent channels (e.g., Markov channels) remains relatively unexplored. The basic assumption that users' capacities are \textit{i.i.d.} random variables is no longer valid, and a more realistic approach should be taken. The channel observed by a user is time varying, hence its distribution changes in time, i.e., a channel with memory. Since intelligent user selection has significant advantages, analyzing the distribution of the maximal capacity seen by a user in the system can give a good approximations for the capacity scaling law.

    In this work, we consider a channel with two channel states, Good and Bad, where in each point in time a user is associated with one of them. The process of moving between states is modeled by a Markovian process (Gilbert-Elliott Channel). First, we derive the expected capacity under scheduling for this time dependent environment and show that its scaling law is $O(\sigma_g\sqrt{2\log K}+\mu_g)$, were $\sigma_g, \mu_g$ are the good channel parameters (assuming Gaussian capacity approximation, e.g., under MIMO) and $K$ is the number of users. Analysis uses Extreme Value Theory (EVT) along with finding the suitable normalizing constants.

    In addition, a distributed algorithm for this scenario is suggested along with it's channel capacity analysis. The algorithm is threshold-based and the rate for exceeding it is analyzed using Point Process Approximation (PPA). The analysis is done using the properties of the extremes of chain dependent sequences, which is proved in this study to converge to one of the extreme value distributions. The expected capacity, while imposing the discussed algorithm, scales as $O\left(e^{-1}(\sigma_g\sqrt{2\log K}+\mu_g)\right)$, hence, there is no loss in optimality due to the distributed algorithm. The foundation for extending this study to more channel states is set and can be easily done.

    Finally, we turn to performance analysis of such systems while assuming the users are not necessarily fully backlogged, focusing on the queueing problem and, especially, on the strong dependence between the queues. We adopt the celebrated model of Ephremides and Zhu to give new results on the convergence of the probability of collision to its average value (as the number of users grows), and hence for the ensuing system performance metrics, such as throughput and delay. We further utilize this finding to suggest a much simpler approximate model, which accurately describes the system behavior when the number of queues is large. The system performance as predicted by the approximate models shows excellent agreement with simulation results.

\end{abstract}

\tableofcontents
\newpage

\pagestyle{fancy}
\pagenumbering{arabic}
\thispagestyle{empty}
\chapter{Introduction}\label{intro}

    Numerous innovations, such as sophisticated coding techniques and advances in the endpoint equipment, have been suggested over recent years in order to cope with the high traffic demands over the last-hop wireless medium. Yet, the rapid growth in the number of users and devices which are connected via this last mile wireless access architecture requires efficient scheduling techniques as well, since in practical scenarios, not all of them can be scheduled together. Indeed, scheduling schemes have been widely studied is such multiuser environments, trying to exploit the multiuser diversity which is inherent in wireless medium, by scheduling users when their instantaneous channel quality is at its best. Utilizing multiple antennas both at the transmitter and the receiver has also been adopted by state of the art wireless standards such as 802.11n, 802.11ac and LTE advanced. Multiuser diversity strategies under such Multiple-Input Multiple-Output (MIMO) technologies have been exploited in many scheduling studies, e.g., in \cite{kim2005scheduling}, \cite{yoo2006optimality}, where Zero-Forcing Beamforming (ZFBF) was investigated and user selection was performed in order to avoid interference among users streams.\\

    Scheduling schemes can be categorized as \emph{centralized}, typically requiring high overhead in collecting channel reports from many users by the central entity (e.g., access point or base station),  in order to perform informed user selection, or \emph{distributed}, which are designed such that the users can schedule themselves. One such opportunistic distributed scheme which has shown excellent exploitation of multiuser diversity, lets the users transmit based on their instantaneous channel condition. Specifically, users may transmit only if their channel gain is above a given threshold. \\

    While the problem in its many forms, as described above, has been studied in various time-independent scenarios, capacity and algorithms for time-dependent channels (e.g., Markov channels) remained relatively unexplored. We consider the situation for the upstream traffic in a distributed manner with the assumption that a user experiences a time varying channel, therefore a time dependent model is suggested, which reflects the time varying channel distribution. Using this model and its properties we were able to achieve closed analytic expressions in different ways for the channel capacity while using two statistical tools, Extreme Value Theory (EVT) and Point Process Approximation (PPA). Those enabled us to examine the limit distribution of the maximal capacity value for the transmitting user, i.e. multi-user diversity. We also looked on a model that describes the segmentation of users according to the different channel states they exist in, and tried to obtain also the expression for the channel capacity. \\

    Obviously, a distributive paradigm in which users schedule themselves for transmission can result in collisions between users, which, in turn, affects channel utilization and quality of service (QoS). In this work we are interested in characterizing these important metrics. We study users' performance under realistic traffic patterns, in which users are not fully backlogged and the arrival process of each user is a random process. Specifically, we assumes a realistic model in which the arrival process of each user is modeled as a Poisson process. Each user transmits packets in a first-in-first-out (FIFO) manner, such that packets arriving to a user which is already busy with a pending packet transmission, wait for their turn to be transmitted. Accordingly, each user maintains a queue in which it stores packets waiting for transmission. We study performance metrics other than capacity, such as delay and buffer occupancy. We suggest a general model, consisting with a system of $K$ users in a random, multiple-access system who interact and affect each other. This theoretical problem was studied in the past and it has been shown that the interaction between the queues in such systems makes the analysis inherently difficult. Thus there is no clear understanding of the performance of such systems with large number of users (more than 3), due to the complexity in describing the strong interdependence between the user's queues. Only approximate models were developed in order to describe the system behaviour.\\

    This work is organized as follows. In Chapter \ref{related work}, literature survey is given which will review related work in the subject. In Chapter \ref{preliminaries} we will review some of the relevant materials for better understanding the problem and it's properties. When part of which will cover some results concerning EVT and PPA. Those are necessary to carry out the asymptotic analysis of the channel capacity. We will elaborate on results that are relative for time dependent processes (stationary sequences), which we shall use in our time dependent situation. In Chapter \ref{model description} we will describe our model of the system and basic assumptions for this work. As mentioned before, this work is divided to two principal areas of research, the first is given in chapter \ref{Capacity under time dependent channel}, there we will focus on the design and analysis of an algorithm which schedules users in the time-dependent regime in a distributed manner. The algorithm is threshold based, and takes the advantages of Multi-User diversity in its operation. The resulting channel capacity is given analytically, simulations and bounds are also given. In chapter \ref{Delay and QoS under the distributed algorithm}, the second area will be given, which will examine the performance analysis of the queues in our distributed multiple-access system. Models for the independent and dependent time scenarios is presented along with numerical results.

    \section{Main Contributions}\label{main contributions}
    \begin{itemize}
      \item In chapter \ref{Capacity under time dependent channel}, we start with giving a closed expression for the channel capacity scaling laws, in a centralized time-dependent environment, while the best user, with the highest channel quality, is chosen to utilize the channel. Then we present a distributed algorithm for the users scheduling mechanism, and the time dependent channel capacity scaling laws under this algorithm. During the analysis of the distributed algorithm, we show a proof for the extreme values convergence of our chain dependent process to one of the EVT distributions, while using the dependency condition $D$ and $D'$ for stationary sequences. And as part of this convergence, we derive the normalizing constants for the convergence to extreme value distribution, where the marginal distribution is a mixture of normal distributions. All of these analytical developments, set foundations for the extension of the channel model to more than two states, and finding the maximal value distribution for this case to derive the expected channel capacity.
      \item In chapter \ref{Delay and QoS under the distributed algorithm}, we present approximate models for the systems' queues behaviour. We first concentrate and review the model presented in \cite{ephremides1987delay}, while fitting it to our setting. Numerical results and interesting observations are presented for this elaborated model. Then, we present a simpler approximate model which in contrary to the former, is able to describe the systems' behaviour when the number of queues is large. Lastly, we suggest a time dependent queue model which relates to our time dependent channel and shows very good agrement with simulations results.

    \end{itemize}

\chapter{Related Work}\label{related work}

    Two of the most important results in network information theory are the capacity of the Gaussian multiple-access channel (MAC) and the capacity of the Gaussian broadcast channel (BC). The desire to maximally utilize both channel capacities is in the heart of many of the studies dealing with shared medium communication. As mentioned in the pervious chapter, scheduling and user selection is a crucial part in achieving these capacities when it comes to large population. In this study we are dealing with the Gaussian MAC model, nevertheless many ideas from works concerning the Gaussian BC are extremely relevant such as the utilization of multi-user diversity and the use of order statistic (EVT) for statistical deductions. In the following sections we review some of state of the art works concerning both types of channels and the attempts to obtain the optimal capacity by using different methods for scheduling and user selection along with variety of statistical analysis.\\

    The communication system considered in this work is single user MIMO - Multiple Input Multiple Output. Basically it means to use multiple antennas at both the transmitter and the receiver in order to increase data rates through multiplexing or to improve performance through diversity. Multiplexing is obtained by exploiting the structure of the channel gain matrix to obtain independent signalling paths that can be used to send independent data. \\
    By using this system we can achieve high data rate without increasing the total transmission power or bandwidth. Additionally the use of multiple antennas at both the transmitter and receiver provides significant increase in the capacity \cite{telatar1999capacity,goldsmith2003capacity,goldsmith2005wireless,foschini1998limits}. MIMO has been supported by several novel wireless standards such as 802.11ac and LTE Advanced.\\


    \section{Scheduling Algorithms}
    While considering a large multi-user networks, each user will experience different channel state which caused by various conditions. This diversity in users' channel state may be used to enhanced the ergodic capacity of the channel. In \cite{knopp1995information}, it was shown that in order to improve spectral efficiency, at any given time, only the user (or users) with the strongest signal should transmit over the entire bandwidth. Hence, we must choose intelligently the user with minimum complexity and overhead.
    The work of \cite{bettesh1998low} improved the work of \cite{knopp1995information} when considering delay aspects, when they treat the situation of starvation. Another work which strengthened the fact that appropriate user selection is essential, and in several cases can even achieve optimality with sub-optimal detectors is \cite{choi2010user}, which discussed various user selection methods in several MIMO detection schemes.

    \subsection{Broadcast Channel}
    Dirty Paper Coding (DPC) for Gaussian broadcast channel, have proven to be the optimal solution for achieving the channel capacity \cite{weingarten2006capacity}. In practice this technique is not used as it involve complex computations and it is not feasible for large population of users. Hence, with the use of multi-user diversity, simpler methods for achieving channel capacity ware derived. In \cite{yoo2006optimality}, the researchers investigate a technique to achieve the MU-MIMO capacity called Zero-Forcing Beamforming (ZFBF). The paper investigates if, for a large number of users, simpler schemes can achieve the same performance as dirty paper coding (DPC) technique.\\
    Beamforming (BF) is a strategy that can serve multiple users at a time, where each user stream is coded independently and multiplied by a beamforming weight vector for transmission through multiple antennas. By choosing the weight vectors mutual interference between different streams could be reduced or disappear, allowing spatial separation between users and as a result support multiple users simultaneously. In the paper the researchers have managed to present a low complexity algorithm for semi orthogonal user selection, which achieves a sum rate close to the optimal rate given by DPC. In addition they presented a fairly scheduling scheme that is asymptotically optimal and has good performance.\\
    As BF and DPC show similar performances of the sum rate capacity, a comparison was made with time sharing capacity in \cite{sharif2007comparison}, and the researchers' conclusion was that the sum rate of DPC and BF outperform of the time sharing if the transmitter is equipped with multiple antennas.
    In \cite{zakhour2011min}, asymptotically optimal scheme for downstream traffic for coordinated beamforming where multiple BSs jointly optimize their respective beamforming vectors was given.
    In \cite{kim2005scheduling}, another scheduling algorithm over a MIMO broadcast channel is introduced. Given a set of users, the scheduler selects more than one user and transmits independent data to them simultaneously by using zero-forcing beamforming. Taking computational complexity into account, a greedy method finding the best and most orthogonal channel vectors is proposed. Additionally, considering fairness among asymmetric users, the researchers also propose an asymptotically fair scheduling algorithm.
    The technique, opportunistic beamforming and nulling, was introduced in \cite{viswanath2002opportunistic}.It is used to induce fast and large fluctuations in the SINR of the users so that multiuser diversity can still be exploited . This technique amplifies the possible multiuser diversity gain while satisfying reasonable latency requirements.
    In \cite{caire2006mimo}, low complexity scheme are considered and compered in means of complexity and performance. It is shown that under quiet simple settings linear beamforming achieves the best tradeoff between performance and complexity.
    In \cite{pun2007opportunistic}, the authors suggested opportunistic beamforming and scheduling schemes with different linear combining techniques. The authors used EVT to derive throughput and scaling laws of the effective signal-to-interference ratios.\\

    In \cite{chen2006enhancing}, a scheduling scheme in TDMA system is suggested which uses the Wishart channel matrix \textbf{$HH^\dag$} eigenvalues analysis.
    Another work regarding user selection is \cite{airy2003spatially}. The paper considers the benefit of opportunistic channel state dependent scheduling which can provide significant performance gains for wireless networks by exploiting the independence of fading statistics across the user population.
    A combination of spatial multiplexing with multiuser diversity is given in \cite{choi2004downlink}. The researchers develop an optimal cross-layer scheduling mechanism that executes fair scheduling at the upper layer and optimal antenna assignment at the physical layer. For fair scheduling, they propose a framework that achieves the objective of maximum capacity and proportional fairness. For optimal antenna assignment, they consider the Hungarian algorithm that maximally utilizes the characteristics of MIMO systems by adopting the graph theoretical approach.\\
    More relevant works are \cite{primolevo2005channel,yoo2006finite}, and a survey on limits of MIMO BC can be found  in \cite{hassibi2007fundamental}.\\

    The characteristics of Quality of Service (QoS) parameters can also be used as appropriate scheduling decisions rules, in \cite{song2004joint}, a multi-user downlink data scheduler with quality-of-service (QoS) provisioning was introduced. In addition to algorithms that are aware of both channel and queue state information to achieve the maximum aggregate utility in the network. Another scheduler using queue-length-based scheduling is introduced in \cite{eryilmaz2005fair}, which implements a congestion control mechanism either at the base station or at the end users and can lead to fair resource allocation and queue-length stability.
    In \cite{swannack2004low} a practical low-complexity scheduling strategy was suggested by users' queue state, along with sub-optimal multiplexing strategy.\\

    Another relevant work which use order statistics is \cite{wang2007coverage}, where the authors suggested a subcarrier assignment algorithm for Orthogonal Frequency Division Multiplexing (OFDM) based systems, and used order statistics to derive an expression for the resulting link outage probability. Order statistics (or EVT) is also required when one wishes to determine the distribution of the selected users, rather than the a priori distribution.\\
    As we can see appropriate user selection is essential for achieving best utilization of the channel capacity.

    \subsection{Multiple-Access Channel}

    Centralized scheduling systems can use information given to them from the users in order to preform better scheduling decision and pass them along to the users.
    In \cite{liu2003opportunistic}, a framework for opportunistic scheduling over multiple wireless channels is introduced. Users are selected for data transmission at any time, though with different throughput and system resource requirements. The article introduced and analyzed Multi-channel Fair Scheduler (MFS), the first wireless scheduling algorithm that provide long term deterministic (MFS-D) and probabilistic (MFS-P) fairness guarantees respectively over multiple wireless channels.
    In \cite{zhang2008low}, the writers present a new low complexity user selection algorithm based on the adaptive Markov chain Monte Carlo (AMCMC) optimization technique to maximize the sum capacity in an uplink MIMO multiuser system. With the proposed algorithm, 99\% of the optimal capacity obtained by exhaustive search method can be achieve.\\
    In \cite{bandemer2007capacity}, a novel scheduling metric was derived based on ergodic sum capacity. A per-user receiver-side correlated channel model was combined with a second-order Taylor series of the ergodic sum capacity as a function of SNR. Thus, an ergodic capacity approximation can be written as an explicit function of the channel model parameters, i.e., the receiver-side channel correlation matrices.\\
    This approximation can directly serve as a measure of spatial compatibility among a given set of users, i.e., as a scheduling metric. The scheduling metric presented in this work differs from previously proposed metrics mainly in the fact that no instantaneous channel state information is required. Long-term channel statistics are sufficient, implying scheduling robustness and reduced effort in channel sounding.
    In \cite{lee2009novel}, an effective MIMO precoding scheme to handle the intercell interference for uplink communication in a multicell environment is introduced.
    The precoding scheme allows to select the number of transmit streams adaptively to the surrounding environments.\\

    \subsubsection{Distributed Scheduling}

    In contrast to the centralized approach, distributed scheduling system decisions is preformed autonomously relying solely on information available at the user. The basic assumption is the the users has all information they need in order to preserve the integrity of the algorithm and updates from the BS are negligible. \\
    In \cite{qin2003exploiting,qin2006distributed}, a decentralized MAC protocol for Orthogonal Frequency Division Multiple Access (OFDMA) channels is given, it is a threshold based algorithm and each user estimates his channel gain and compares it to the threshold, in case of exceedance the user can transmit. The optimal threshold indeed maximizes the number of slots in which only one user exceeds it. This distributed scheme achieves $1/e$ of the capacity which can be achieved by scheduling the strongest user. \cite{bai2006opportunistic} extended the distributed threshold scheme to multi-channel setup, where each user competes on $m$ channels and examined proportional fairness in heterogeneous systems.
    In \cite{qin2008distributed}, the authors used a similar approach for power allocation in the multi-channel setup, and suggested an algorithm that asymptotically achieves the optimal water filling solution. To reduce channel contention, \cite{qin2003exploiting}\cite{qin2004opportunistic} introduced a splitting algorithm which resolves collision by allocating several mini-slots devoted to finding the best user. Assuming all users are equipped with a Collision Detection (CD) mechanism, the authors also analyzed the suggested protocol for users not fully backlogged, where the packets randomly arrive with a total arrival rate $\lambda$ and for channels with memory. \cite{to2010exploiting} used a similar splitting approach to exploit idle channels in a multichannel setup, and showed improvement of 63\% compared to the original scheme in \cite{qin2003exploiting}. In \cite{garcia2012distributed}, the authors perform a joint optimization of both the thresholds and the access probabilities, which provides a fair allocation that achieves a good tradeoff between total throughput and fairness. \cite{kampeas2012capacity} extended the suggested scheme in \cite{qin2003exploiting} and gave a novel way to analyse such threshold-based algorithms, the Point Process Approximation (PPA), and apply it to non identically distributed scenarios as well.

    \textbf{Extreme Value Theorem }- one of the ways to analyse such threshold exceedances is to obtain the distribution of the maximum channel gain. A convenient tool for analyzing this kind of extreme distributions is the Extreme Value Theorem (EVT). As the number of users rises ($K\rightarrow\infty$) the distribution function of the maximal capacity degenerates, therefore EVT helps to obtain the limit distribution for large number of \textit{i.i.d.} users.
    Prior work which is using this tool and refer to these maximal distributions can be found in the article of \textit{Choi and Andrews (2008)}\cite{choi2008capacity}, where we can find the scaling laws of maximal base station scheduling via Extreme Value Theory (EVT). By scheduling the user with the strongest channel among $K$ users one can gain a factor of $O(\sqrt{2\log K})$ in the expected capacity compared to arbitrary (random, Round-Robin) scheduling.
    Works mentioned in this chapter such as \cite{kampeas2012capacity,pun2007opportunistic,wang2007coverage}, also used EVT for statistical deductions. EVT has been proved useful for analyzing reliable broadcasting as well \cite{xiao2010extreme,xiao2012reliable}. In this work EVT is used in the context of stationary process since this work, in contrary to the works presented above, assume the realistic assumption of time dependency.

    \section{Time Dependent Channel}
    Since this study extends the problem to the case which the users experience a time varying channel, the system model should be described as a system with memory which has different channel states. In the article of \textit{Gillbert (1960)} \cite{gilbert1960}, we can find a model for burst-noise binary channel, where the channel uses a Markov chain in order to illustrate a case of error free with a good state and a case of probability for error in transmission with a bad state. The channel capacity for this model is derived. Proceeding to the article of \textit{Gillbert}, in the article of \textit{elliott (1963)} \cite{elliott1963}, a means for estimating error rates for binary block codes was introduced. In \cite{zhang1999finite}, the researchers form a finite-state Markov channel model to represent Rayleigh fading channels using the expended model of \textit{Gillbert-elliott}. They develop and analyze a methodology to partition the received signal-to-noise ratio (SNR) into a finite number of states according to the time duration of each state. Each state corresponds to different channel quality indicated by bit-error rate (BER). The number of states and SNR partitions are determined by the fading speed of the channel. In \cite{biglieri1998fading} a review on the development of fading channel (i.e. channels with states) concerning statistical models, capacity in terms of single user and multiuser transmission and more is presented. More works which are highly relevant are \cite{tse1998multiaccess,hanly1998multiaccess} dealing with the capacity of a time dependent channel and the optimal power allocation under average power constraints.

    \section{Limit laws in time dependent channels}

    The study on the limit laws for \textit{i.i.d.} random variables was fully characterized in many works starting with \cite{gnedenko1943distribution}, therefore studies regarding stationary and non stationary sequences and their extreme laws became the natural extension for this problems. One interesting case is the chain dependent process which aligns with our desire to model the users' channel as described above. Chain dependent process was first introduced by Janssen (1969, \cite{janssen1969processus}) and it is known as the $J-X$ process, where we have the bivariate random variables sequence $(J_n,X_n), n\geqq0$. The sequence $X_n$ depends on the state which the process $J_n$ exist in. Exact formulation for the problem will be given later. We will review some relevant work in this subject. In \cite{resnick1970limit} we can find characterization for the limiting behavior of the maximal value distribution of the sequence $X_n$ while conditioning on $J_n$. With suitable normalization constants, the exact solution to the distribution of the maximum value is given. This solution is computationally difficult and was formed by using the theory of semi-Markov processes, which employs similar formulation for matrices as formed by the $J-X$ process definition, in order to approximate the limiting distribution laws. In addition they establish basic properties of the normalizing constants, but did not give a specific way to derive them. In \cite{resnick1973almost} and in its preceding work \cite{resnick1972stability} necessary and sufficient conditions are found for the distribution of the maximal value to be relatively and almost surely stable. One important work is \cite{denzel1975limit} which stated that the limiting behaviour of the maximal value is not dependent on the initial distribution of the underlying Markov chain, thus the sequence $J_n$ may start with the stationary distribution of the chain which make $X_n$ also stationary process. In addition the strong mixing property of chain dependent sequences was also given leading eventually to the fact that under fairly general conditions the limiting behavior of the maxima for chain dependent processes is identical to that for the \textit{i.i.d.} process with the same marginal distribution. In \cite{turkman1992limit} and in its preceding paper \cite{turkman1983limit} dependency conditions of the $X_n$ sequence known as $D$ and $D'$ which were given in \cite{leadbetter1974extreme} were used while achieving the results for the limiting distributions for chain dependent process. Properties for a suitable normalizing constants were also given under general conditions, in order for the limiting behaviour will converge to one of the limiting distributions introduced in \cite{gnedenko1943distribution}. Rate of convergence for these approximations for the limiting convergence can be found in \cite{mccormick2001rates}. It is important to mention that all the studies above showed convergence to one of the three limiting distribution given in \cite{gnedenko1943distribution} and in some we have specific characterizations for achieving the normalizing constants. In this work the channel model is based on chain dependent process and results regarding the maximal value behavior is relevant for understanding the channel capacity in our scenario, which make the studies above extremely relevant. Similarly to \cite{turkman1983limit} we used the long range and local dependence condition $D$ and $D'$ respectively, for finding the extreme laws for our sequence. However the assumptions made in \cite{turkman1983limit} regarding the marginal distributions are more strict comparing to the ones we did resulting with different normalizing constant for the convergence.

    \section{Queueing system}
    The analysis of the performance evaluation of our random multiple-access system is an extremely interesting and important theoretical problem. The interaction between the queues in such systems makes the analysis to be inherently difficult \cite{fayolle1979two}, and the importance of such analysis was recognized in \cite{kleinrock1981interfering}. As will be described in the sequel, we are dealing with a slotted time system and therefore the most notable system is the slotted aloha. This system was widely studied in the past with two major aspects, the stability region of the system and its QoS properties and performance. Where the first aspect relate to the problem of finding the values of the arrival rates so the system will be stable, and the second aspect relate to the analysis of such system in terms of delay, mean queue size and more. Here we will not engage the stability problem, since existing results are known for the cases of $K=2$ \cite{rao1988stability,tsybakov1979ergodicity } and $K=3$ \cite{szpankowski1994stability}, and also sufficient bounds on the arrival rates can be found in \cite{fayolle1977stability,rao1988stability,luo1999stability} for $K>3$ in order to maintain stability. In fact, we will investigate the performance of such system which have appeared in many works.

    Even toady there is no clear understanding of the performance of such systems with large number of users (more than 3), due to the complexity in describing the strong interdependence between the user's queues.
    Only approximate models were developed in order to describe the system behaviour. In \cite{saadawi1981analysis}, an iterative approximation model was suggested while using decoupled Markov chains. A refined model was presented in \cite{ephremides1987delay}. In \cite{sidi1983two}, the mean delay was given for the case of two identical users, as well as an approximation for a larger population. Extension for slotted CSMA/CD model can be found in \cite{takagi1985mean}.

\chapter{Preliminaries}\label{preliminaries}

    This chapter offers the required technical background which will be used throughout this work. Starting with the Gaussian channel through the definition of a MIMO system and the its capacity under Gaussian channels, statistical tools for analysing independent and dependent sequences of random variables and the review of results regarding sequences which depends on an underlying Markov chain.\\

    \section{The Gaussian Channel}\label{preliminaries-Gaussian channel}

    In information theory the communication between two points means a physical process which the information transfers through it, therefore it is subject to uncontrollable ambient noise and imperfections of the physical signalling process itself. In \cite{Cover:1991} it is proven that the channel capacity is the tightest upper bound on the amount of information that can be reliably transmitted over a communication channel, while taking into account the interferences.\\
    Our communication system is MIMO, and between each pair of antennas the channel model is a Gaussian channel which is defined as follows:\\
    \begin{equation}
     Y=X+Z
    \end{equation}
    Where Y is the output of the channel with input X and with an independent noise Z with variance N.  Without any conditions the capacity of this channel may be infinite Since X can take any value. If the variance of the noise is zero, then the receiver receives the transmitted symbol perfectly, whereas if the variance of the noise is non-zero we can choose an infinite subset of inputs so that they are distinguishable at the output with arbitrarily small probability of error. Such a scheme has an infinity capacity as well.\\
    The most common limitation on the input is an energy or power constraint:
    \begin{equation}
      E[X^2]\leq P
    \end{equation}
    And so the capacity of such channel is:
    \begin{equation}
      C=\max_{E[X^2]\leq P}{\mathcal{I}(X;Y)}=\frac{1}{2}\log{(1+\frac{P}{N})}
    \end{equation}
    The maximum is attained when $ X\sim \mathcal{N}(0,P).$\\

    \section{Capacity of MIMO}\label{preliminaries-capacity of MIMO}
    Consider a MIMO system with $t$ transmit antennas and $r$ receives antennas  as illustrated in Figure \ref{fig-MIMOchannel}. The linear link model between the transmit and receive antennas can be represented in the vector notation as
    \begin{equation}\label{equ-linaer model of MIMO}
      \textit{\textbf{y}=H\textbf{x}+\textbf{n}}
    \end{equation}
    where $\textit{\textbf{y}}$ is the $r\times1$ received vector from the $t\times1$ input vector $\textit{\textbf{x}}$. $\textit{\textbf{n}}$ is $r\times1$ zero mean complex Gaussian noise vector with independent, equal variance real and imaginary parts. And $\textit{H}$ is a $r\times t$ normalized channel matrix,
    \begin{equation}\label{equ-Channel matrix}
      H =
                 \begin{pmatrix}
                  h_{1,1} & h_{1,2} & \cdots & h_{1,t} \\
                  h_{2,1} & h_{2,2} & \cdots & h_{2,t} \\
                  \vdots  & \vdots  & \ddots & \vdots  \\
                  h_{e,1} & h_{r,2} & \cdots & h_{r,t}
                 \end{pmatrix}
    \end{equation}
    Each element $h_{ij}$ n represents the complex gains between the $n^{th}$ transmit and $m^{th}$ receive antennas.
    \begin{figure}[h]
    \centering
    \includegraphics[width=4in]{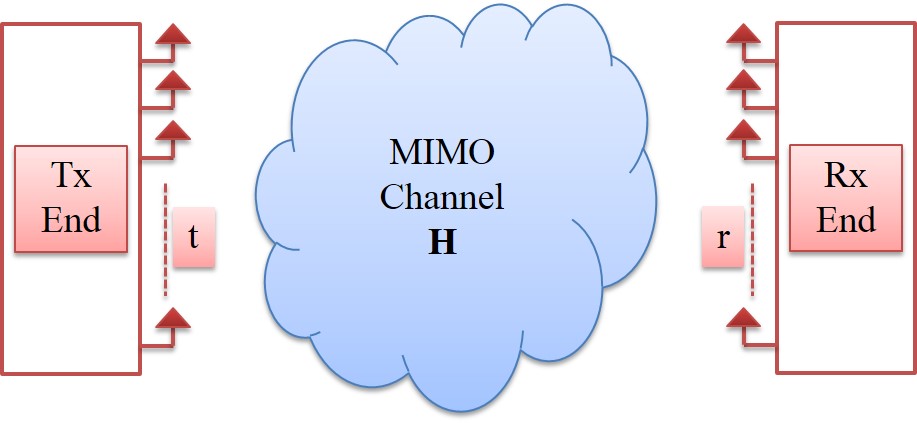}
    \caption[MIMO channel]{MIMO channel}
    \label{fig-MIMOchannel}
    \end{figure}

    \subsection{multi-antenna Gaussian channels}\label{preliminaries-capacity of MIMO-Gaussian channels}
    After describing the system and the capacity of a single Gaussian channel we would like to explore the capacity of the system. In \cite{telatar1999capacity} the researcher investigates the use of multiple antennas at the receiver and at the transmitter for single user communication over the additive Gaussian channel, with or without fading, and calculates the capacity of such a system.
    The linear model of the system is as given in \eqref{equ-linaer model of MIMO}, where now $\textit{\textbf{y}}\in\mathbb{C}^r$, $\textit{\textbf{x}}\in\mathbb{C}^r$. The transmitter is constrained in its total power to P,
    \begin{equation}
      E[\textit{\textbf{x}}^\dagger\textit{\textbf{x}}]\leq P.
    \end{equation}
    The matrix \textit{H} is a random matrix when all its entries are \textit{i.i.d.} Gaussian variables with zero-mean independent real and imaginary parts, each with variance 1/2. This matrix models the Rayleigh fading environment of the channel. In \cite{telatar1999capacity}, the researcher focus in the scenario which the matrix \textit{H}, is a random matrix, but in order to expose the techniques used to calculate the channel capacity he consider the case which \textit{H} is deterministic.\\
    Using the singular value decomposition and consider \textit{H} is deterministic
    we will derive the capacity of this cannel.\\
    This is a vector Gaussian channel, We can compute the capacity by decomposing the vector channel into a set of parallel, independent scalar Gaussian sub-channels using SVD.\\
    The matrix \textit{H} can be written as $\textit{H}=\textit{UDV}^\dag$ where
    $\textit{U}\in\mathbb{C}^{r\times r}$ and $\textit{V}\in\mathbb{C}^{t\times t}$ are unitary, and $\textit{D}\in\mathbb{R}^{r\times t}$ is non-negative and diagonal (the diagonal elements are the ordered singular values of the matrix \textit{H}).\\
    Therefore (4) can be written like this:
    \begin{equation}
      \textit{\textbf{y}}=\textit{UDV}^\dag\textit{\textbf{x}+\textbf{n}}
    \end{equation}
    because \textit{U} and \textit{v} are invertible we can write:
    \begin{equation}
      \textit{U}^\dag \textit{\textbf{y}}=\textit{DV}^\dag \textit{\textbf{x}}+\textit{U}^\dag \textit{\textbf{n}}
    \end{equation}
    and we mark it as: $\textit{\textbf{\~{y}}=D\textbf{\~{x}}+\textbf{\~{n}}}$, where $\textit{\textbf{\~{n}}}$ has the same distribution as $\textit{\textbf{{n}}}$.\\
    Since \textit{H} is of rank at most $min\{r,t\}$, at most $min\{r,t\}$ of the singular values of it are non zero. In addition $\|\textit{\textbf{\~x}}\|^2=\|\textit{\textbf{x}}\|^2$, thus the energy is preserved and we have equivalent representation as a parallel Gaussian channel:
    \begin{equation}\label{}
       \textit{\~y}_{i}=\lambda_i\textit{\~x}_{i}+\textit{\~n}_{i},\hspace*{4 mm} 1\leq i\leq\min{\{r,t\}}
    \end{equation}
    So the capacity can be calculated as a system of parallel independent channels:
    \begin{equation}
       C=\sum_i\log{(1+\frac{{P_{i}}^*{\lambda_{i}}^2}{N_0})}
    \end{equation}
    where ${P_{1}}^*,...,P_{min\{r,t\}}^* $ are the waterfilling power allocations: ${P_{i}}^*=(\mu-\frac{N_0}{\lambda_{i}}^2)^+ $,
    with $\mu$ chosen to satisfy the total power constraint $\Sigma_{i} {P_{i}}^*=P$.\\

    In the case which \textit{H} is random, it is shown that the capacity of the channel is achieved when \textit{\textbf{x}} is a circularly symmetric complex Gaussian with zero-mean and covariance $\frac{P}{t}\textit{I}_t$.
    The capacity is given by
    \begin{equation}
      E[\log{(det(\textit{I}_r+\frac{P}{t}HH^\dag))}],
    \end{equation}
    as t grows the expression $\frac{1}{t}HH^\dag \rightarrow \textit{I}_r $ almost surely (law of large numbers) and therefore the capacity is
    \begin{equation}
      r\log{(1+P)}.
    \end{equation}
    The way to derive the capacity is by maximizing the average mutual information $\mathcal{I}(\textit{\textbf{x}};\textit{\textbf{y}})$.\\

    In \cite{goldsmith2003capacity} we can find an overview of the extensive recent results on the capacity of MIMO channels for both single-user and multiuser systems. The great capacity gains predicted for such systems can be realized in some cases, but realistic assumptions about channel knowledge and the underlying channel model can significantly mitigate these gains.

    \subsection{A Gaussian Approximation to the MIMO Capacity}\label{preliminaries-capacity of MIMO-Gaussian Approximation}
    Asymptotic capacity analysis of MIMO channels with arbitrary large number of antennas is widely used to overcome mathematical complexity and often allows for significant insight. Several works \cite{girko1998refinement,smith2002gaussian,chiani2003capacity} investigate the accuracy of a Gaussian approximation by deriving the variance and mean of the capacity of a MIMO system resulting with the fact that, distribution of the single-user-MIMO capacity follows the Gaussian distribution with mean that grows linearly with $\min(r, t)$, and variance which is mainly in influenced by the power constraint $P$.\\
    According to this, in this work each user capacity is approximated by a Gaussian distribution, i.e., for each user $i$ , $C_i\sim N(\mu_i,\sigma_i)$.

    \section{Extreme Value Theory}\label{preliminaries-EVT}

    As the number of users raises ($K\rightarrow\infty$) the distribution function of the maximal capacity is degenerating, therefore a need to obtain the limit distribution for asymptotic analysis is required.\\
    Extreme Value Theory (EVT) is a common tool used to approximate such limit distributions, for large number of random variables.
    In this sub-section, we briefly review some EVT important results that are used in the asymptotic analysis of the capacity gain(\cite{EVT:Springer1983},\cite{EVT:Springer2001}).\\
    We first review results for \textit{i.i.d.} sequences. We then elaborate on extremes results for stationary sequences which will be used when analysing the time dependent capacity sequence each user experience.\\

    \begin{theorem} (\cite{EVT:Springer1983},\cite{EVT:Springer2001}). Let $(X_1,...,X_n)$ be a sequence of \textit{i.i.d.} random variables with distribution $F(x)$, and let $M_n=\max{(X_1,...,X_n)}$. If there exists a sequence of normalizing constants $a_n>0$ and $b_n$ such that as $n\rightarrow\infty$,
        \begin{equation}
            P\{a_n(M_n-b_n)\leq x\}\rightarrow G(x)
        \end{equation}
        for some non-degenerate distribution G, then G belongs to the type of one of the following three distributions:
                \begin{equation*}
                  \begin{aligned}
                  \text{Type \Rmnum{1}:  }&G(x)=e^{-e^{-x}} \quad \quad -\infty<x<\infty    \\
                  \text{Type \Rmnum{2}:  }&G(x) = \left\{
                                              \begin{array}{l l l}
                                                    0 & & x\leq0 \\
                                                    e^{-x^{-\alpha}} & \quad \text{for some $\alpha$}>0 &x>0
                                              \end{array} \right.\\
                  \text{Type \Rmnum{3}:  }&G(x)=   \left\{
                                              \begin{array}{l l l}
                                                    e^{-(-x)^{\alpha}} & \quad \text{for some $\alpha$}>0 & x\leq0\\
                                                    1 & & x>0
                                              \end{array} \right.\\
                  \end{aligned}
                \end{equation*}

        The three types of distribution limits can be formulated to a generalized extreme value (GEV) distribution type
        \begin{equation}
            G(z)=\exp\{-(1+\xi(\frac{z-\mu}{\sigma}))^{-1/\xi}\}
        \end{equation}
        and we say that $F(x)$ is in the domain of attraction of G, where $\xi,\mu,\sigma$ are the shape, location and scale parameters respectively.\\

        \vspace*{2 mm}In order to determine the shape parameter $\xi$ by the ancestor distribution $F(x)$, let us look on $h(x)$ the following reciprocal hazard function
        \begin{equation*}
          h(x)=\frac{1-F(x)}{f(x)}\ \text{for }x_F\leq x\leq x^F,
        \end{equation*}
        where $x_F:=inf\{x:F(x)>0\}$ and $x^F:=sup\{x:F(x)<1\}$ are the lower and upper
        endpoints of the ancestor distribution, respectively. Then the shape parameter $\xi$ is obtained as the following limit,
        \begin{equation*}
          \frac{d}{dx}h(x)\overset{(x\rightarrow x^F)}{\rightarrow}\xi.
        \end{equation*}
    \end{theorem}

    We have considering convergence of probabilities of the form $P\{a_n(M_n-b_n)\leq x\}$, which may be rewritten as $P\{M_n\leq u_n\}$ where $u_n=u_n(x)=x/a_n+b_n$. So the theorem from above can be formulated in a different way:
    \begin{theorem}\label{thm-EVT convergence different formulation } (\cite[Theorem 1.5.1]{EVT:Springer1983})
        Let $\{X_n\}$ be an \textit{i.i.d.} sequence. Let $0\leq \uptau \leq \infty$ and suppose that ${u_n}$ is a sequence of real numbers such that
        \begin{equation*}
          n(1-F(u_n))\rightarrow \uptau \ \ \text{as } \ n\rightarrow \infty.
        \end{equation*}
        Then
        \begin{equation*}
          P\{M_n\leq\ u_n\}\rightarrow e^{-\uptau} \ \ \text{as } \ n\rightarrow \infty.
        \end{equation*}
    \end{theorem}

    For the more interesting and relevant case of normal distribution, we have the following\\
    \begin{theorem}\label{thm-EVT normal sequnce converge to Gumbel} (\cite[Theorem 1.5.3]{EVT:Springer1983})
        If $\{X_n\}$ are an \textit{i.i.d.} standard normal sequence of random variables, then the asymptotic distribution of $M_n=\max{(X_1,...,X_n)}$ is a Gumbel distribution.\\
        Specifically
        \begin{equation}\label{equ-EVTMaxDistribution}
          P\{a_n(M_n-b_n)\leq x\}\rightarrow \exp\{-e^{-x}\}
        \end{equation}
        where
        \begin{equation}\label{equ-parameter a_n(standard)}
          a_n=(2\log n)^{1/2}
        \end{equation}
        and
        \begin{equation}\label{equ-parameter b_n(standard)}
          b_n=(2\log n)^{1/2}-\frac{1}{2}(2\log n)^{-1/2}[\log{\log n}+\log{4\pi}].
        \end{equation}
        If $\{X_n\}$ follows the Gaussian distribution with mean $\mu$ and variance $\sigma^2$ , then the above normalizing constants are
        \begin{equation}\label{equ-parameter a_n}
          a_n=\sigma(2\log n)^{1/2}
        \end{equation}
        and
        \begin{equation}\label{equ-parameter b_n}
          b_n=\sigma\Big[(2\log n)^{1/2}-\frac{1}{2}(2\log n)^{-1/2}[\log{\log n}+\log{4\pi}]\Big]+\mu.
        \end{equation}
    \end{theorem}
    Please note that if $F(x)$ follows a Gaussian distribution we can see that $\uptau=e^{-x}$.\\

    We turn now for results regarding maxima of stationary sequences. Dependence in stationary series can take many different forms, and it is impossible to develop a general characterization of the behavior of extremes unless some constraints are imposed(\cite{EVT:Springer2001}). The following conditions makes precise the notion of extreme events being near-independent if they are sufficiently distant in time.

    \begin{Definition}\label{def-StrongMixing}
    We say $\{X_n\}$ is strongly mixing if there is a function $g$ on the positive integers with $g(k) \rightarrow 0$ as $k \rightarrow \infty$ such that, if $A \in \mathfrak{F}(X_1,...,X_m)$ and $B \in \mathfrak{F}(X_{m+k},X_{m+K+1},...)$ for some $k, m\geq 1$, then $\mid P(A\cap B)-P(A)P(B) \mid \leq g(k)$. Where $\mathfrak{F}(\cdot)$ denotes the $\sigma$-field generated by the indicated random variables. Strong mixing was first introduced by Rosenblatt (1956).
    \end{Definition}

    In weakening the strong mixing condition, one notes that the events of interest in extreme value theory are typically those of the form $\{X_n\leq u\}$. Hence one may be led to propose a condition like mixing but only required to hold for events of this type.

    \begin{Definition}\label{def-D condition}
    (\cite{EVT:Springer2001},\cite{leadbetter1974extreme}) A stationary series $X_1,X_2,...$ is said to satisfy the $D(u_n)$ condition if,\\ for all $i_1<...<i_p<j_1<...<j_q$ with $j_1-i_p>l$,
    \begin{equation*}
        \begin{aligned}
            &|P_r\{X_{i_1}\leq u_n,...,X_{i_p}\leq u_n, X_{j_1}\leq u_n,...,X_{j_q}\leq u_n\} - \\
            &P_r\{X_{i_1}\leq u_n,...X_{i_p}\leq u_n\}P_r\{ X_{j_1}\leq u_n,...,X_{j_q}\leq u_n\}|\leq\alpha(n,l),
        \end{aligned}
    \end{equation*}
    where $\alpha(n,l)\rightarrow 0$ for some sequence $l_n$ such that $l_n/n \rightarrow 0$ as $n\rightarrow \infty$.
    \end{Definition}
    If $u_n$ is a sequence that increases with $n$ the long range dependence condition $D(u_n)$ ensures that, for sets of variables that are far enough apart, the difference of probabilities expressed above, while not zero, is sufficiently close to zero to have no effect on the limit laws for extremes.

    Another condition that is very relevant to us is the local dependence condition $D'(u_n)$ \cite{leadbetter1983extremes}:
    \begin{Definition}\label{def-D' condition}
    \begin{equation*}
      D'(u_n): \limsup_{n \rightarrow \infty} n\sum_{j=2}^{ \lfloor n/k\rfloor } P\{X_1>u_n,X_j>u_n\}\rightarrow 0 \quad \text{as}\quad k\rightarrow \infty
    \end{equation*}
    \end{Definition}
    This condition bounds the probability of more than one exceedance among $X_1,...,X_{ \lfloor n/k\rfloor }$. This will eventually ensure that there are no multiple points in the points process of exceedances which is necessary in obtaining a simple poisson limit for this point process as will be discussed later.  A result regarding to the two conditions $D(u_n),D'(u_n)$ is:
    \begin{theorem} (\cite[Theorem 3.4.1]{EVT:Springer1983} )
        Let ${u_n}$ be constants such that $D(u_n),D'(u_n)$ hold for the stationary sequence $\{X_n\}$. Let $0\leq \uptau \leq \infty$. Then
        \begin{equation*}
          P\{M_n\leq u_n\}\rightarrow e^{-\uptau} \quad \text{if and only if}\quad n(1-F(u_n))\rightarrow \uptau.
        \end{equation*}
    \end{theorem}
    Thus, from the above we can conclude that the particular GEV type which applies is the same as if the sequence $\{X_n\}$ were \textit{i.i.d.} with the same marginal distribution function $F$, and the same normalizing constants may be used.\\

    \section{Point Process Approximation}\label{preliminaries-PPA}
    As mentioned earlier the expected channel capacity was derived in a number of method, one of them was done while using the point process approximation. PPA is a tool for analysing threshold arrival rates and tail distributions. This analysis imposes threshold-based algorithm which discussed earlier as an option to achieve the maximum channel capacity. Again we will start with the basic theory which concern an \textit{i.i.d.} sequences and elaborate on stationary process.\\

    A point process on a set $\mathcal{A}$ is a stochastic rule for the occurrence and position of point events (\cite{EVT:Springer2001}). Our characterization for the statistical properties of the point process is to define a set of non-negative integer-valued random variables, $N(A)$, for each $A\subset \mathcal{A}$, such that $N(A)$ is the number of points in the set $A$ which can represent a period of time $[t_1,t_2]$. The expected number of points in any subset $A \subset \mathcal{A}$, is defined to be the intensity measure $\Lambda(A)$ of the process. The number of points in a given interval follows a Poisson distribution with intensity measure $\Lambda(A)$ and intensity density $\lambda(t,x)$ where:
    \begin{equation*}
      \Lambda(A)=\int_{A}\lambda(t,x)dA
    \end{equation*}
    We assume ${X_1,X_2,...,X_n}$ is a sequence of \textit{i.i.d.} random variables, with common distribution function $F(x)$. We suppose that $F(x)$ is in the domain of attraction of some GEV distribution $G(\cdot)$, with normalizing constants $a_n$ and $b_n$. We then define a sequence of point processes $N_1,N_2,...$  on $[0,1]\times \mathds{R}^2$ by:
    \begin{equation}\label{equ-sequence of point processes N_n}
      N_n=\Big\{\frac{i}{n},\frac{x_i-b_n}{a_n}:i=1,...,n \Big \}
    \end{equation}
    We can notice that the numbers of occurrences counted in disjoint intervals are independent from each other, law value points are normalised to the same value $b_l$ with:
    \begin{equation*}
      b_l=\lim_{n\rightarrow \infty} \frac{x_F-b_n}{a_n}
    \end{equation*}
    and high value points are retained in the limit process. Under all the above we get:
    \begin{theorem}\label{thm-Point process convergance of iid}(\cite{Kampeas2012scheduling},\cite{EVT:Springer2001})
        The sequence $N_n$ on the set $[0,1]\times (b_l+\epsilon, \infty)$, where $\epsilon >0$, converge to a non-homogeneous Poisson process $N$, that is:
        \begin{equation*}
          N_n\rightarrow N \ \ \text{as } \ n\rightarrow \infty,
        \end{equation*}
        with intensity:
        \begin{equation*}
          \lambda(t,x)=(1+\xi x)_{+}^{-1-1/\xi}
        \end{equation*}
    \end{theorem}
    We are interesting in the sets of the form: $B_u=[0,1]\times (u,\infty)$ where $u>b_l$ which reflect the users with a capacity that exceeded a threshold $u$. In this case the intensity measure is:
    \begin{equation*}
        \begin{aligned}
          \Lambda(B_u)&= \Lambda([0,1]\times (u,\infty))\\
                      &=\int_{t=0}^{1}\int_{x=u}^{\infty}\lambda(t,x)dxdt\\
                      &=\int_{t=0}^{1}\Big[-(1+\xi x)_{+}^{-1/\xi}\Big]_{x=u}^{\infty}dt\\
                      &=\int_{t=0}^{1}(1+\xi u)_{+}^{-1/\xi}dt\\
                      &=(1+\xi u)_{+}^{-1/\xi}
        \end{aligned}
    \end{equation*}
    where $a_+$ denotes $\max\{0,a\}$. The occurrences of capacity exceedances of the threshold $u$ can be modeled by a Poisson process, with parameter $\Lambda(B_u)$. In case the sequence $\{X_n\}$ are standard normal random variables:
    \begin{equation*}
      \lim_{\xi \rightarrow 0} (1+\xi u)_{+}^{-1/\xi}=e^{-u}=\uptau \Big|_{x=u}
    \end{equation*}

    Since we are interesting in time dependent environment we would like to explore the point process analysis for stationary sequences. In \cite{leadbetter1976weak} we can find the point process of exceedances of a high level $u_n$ by a stationary sequence $\{X_i\}$ (i.e., points where $X_i>u_n$), obtaining Poisson limits under quite weak dependence restrictions. These involve a long range dependence condition $D(u_n)$ and a local dependence condition $D'(u_n)$ as defined in definitions \ref{def-D condition} and \ref{def-D' condition}.
    \begin{theorem}\label{thm-Point process convergance of stationary process}(\cite{leadbetter1976weak})
        Let $D(u_n),D'(u_n)$ hold ($u_n$ satisfying $u_n=u_n(\uptau)$, such that $1-F_1(u_n)=P\{X_1>u_n\}\sim \uptau/n$ as $n\rightarrow\infty$ for all $\uptau > 0$ ) for the stationary sequence $\{X_i\}$. Let $N_n$ be the point process, consisting of the exceedances of $u_n(\uptau)$. Then $N_n \overset{d}{\rightarrow} N$ as $n\rightarrow\infty$, where $N$ is a Poisson process with parameter $\uptau$.
    \end{theorem}
    As a result of the above theorem, the point process under the dependence consideration behaves just like one obtained from an \textit{ i.i.d.} sequence with the same marginal d.f. $F$. If the local condition is weakened or omitted, then clustering of exceedances may occur.

    \section{Chain Dependent Process}\label{preliminaries-Chain Dependent Process}
    In this work we use a specific model for dependent data referred to as chain dependent model. We will review the formulation for this model and some relevant results.\\
    Let $\{J_n\}$ be a finite state space Markov chain with ergodic probability transition matrix $P=(p_{ij}), \ 1\leq i,j\leq m$. Further let $\{X_n\}$ be a chain dependent sequence such that
    \begin{equation}\label{equ-J-X process defenition - perliminaries}
        \begin{aligned}
            P(J_n &= j , X_n \leq x \mid X_0,J_0,...,X_{n-1},J_{n-1}=i) \\
                  &= P(J_n=j,X_n\leq x \mid J_{n-1}=i)=p_{ij}H_i(x)
        \end{aligned}
    \end{equation}
    where $H_i, \ i=1,...,m$ are the underlying non-degenerate and honest ($H_i(+\infty)=1$) distributions, associated with each state of the chain.\\
    Immediate consequences\cite{resnick1970limit}:
    \begin{enumerate}
      \item[i] The stationary probabilities associated with $P$ are $(\pi_1,...,\pi_m)$.
      \item[ii] $P(X_n\leq x \mid J_{n-1}=i)=H_i(x)$.
      \item[iii] $P(X_1\leq x_1,...,X_n\leq x_n \mid J_0,J_1,...,J_{n-1}=i)=\prod_{i=1}^n P(X_i\leq x_i \mid J_{i-1})$\\
      The random variables $\{X_n\}$ are conditionally independent given the chain in precisely the sense given by (iii).
    \end{enumerate}

    Define a nonnegative $m\times m$ matrix by $Q(x)=(p_{ij}H_i(x))$ and the random variable $M_n=\max{(X_1,...,X_n)}$, then In \cite{resnick1970limit} we can find the distribution of the maximum given by
    \begin{equation}\label{equ-Chain dependent Q matrix defenition}
      P(J_n=j,M_n\leq x \mid J_0=i)=Q_{ij}^n(x)
    \end{equation}
    where $Q^n$ denotes $n$th matrix power of $Q$. The matrix $Q$ with suitable normalizing constants $a_{ijn}>0$ and $b_{ijn}, \ i,j=1,..,m, \ n\geqq 1$, convergence to non-degenerate mass matrix as $n \rightarrow \infty$
    \begin{equation*}
      \{Q_{ij}^n(a_{ijn}x+b_{ijn})\} \rightarrow \{U_{ij}(x)\} \quad \text{as } n \rightarrow \infty .
    \end{equation*}
    The Perron-Frobenius eigenvalue $\rho_U(x)$, of $\{U_{ij}(x)\}$, is an extreme value distribution. To attain this result the Semi-Markov matrices theory was used.\\

    A very important results given in \cite{denzel1975limit} states that that the limiting behaviour of $M_n$ is not dependent on the initial distribution of the underlying Markov chain. Thus it may be started with the stationary distribution ($J_0=\pi_i, \ i=1,...,m$), in which case $\{X_n\}$ is stationary. In addition it is stated that every stationary chain-dependent process $\{X_n\}$ is strongly mixing (Lemma 1.1, \cite{denzel1975limit}). Then under any initial distribution for $J_0$ (theorem 3.1, \cite{denzel1975limit}) we have
    \begin{equation}\label{equ-Convergence to l by stationary distribution}
      P(M_n\leq d_n)\rightarrow l \quad \Leftrightarrow \quad H^n(d_n)\rightarrow l \qquad \text{as } n \rightarrow \infty
    \end{equation}
    where $d_n$ is a sequence of reals, $0<l<1$ and $H(x)$ is a distribution function such that $H(x)=\sum_{i=1}^m \pi_i H_i(x)$.\\

    The dependency of the sequence $\{X_n\}$ can be characterized such that it would be easier to attain the limiting distribution of $M_n$. Such dependency conditions are the
    $D(u_n)$ and $D'(u_n)$ as given in definitions \ref{def-D condition} and \ref{def-D' condition}. In \cite{turkman1983limit} a similar analysis was performed for this purpose. They defined a class of sequences, were the chain dependent process is a sub case of it, that fulfill some properties which under them the extreme laws could be achieved. Among those properties, the ones that are relevant to us are that the $D(u_n)$ and $D^*(u_n)$ (an extension for $D'(u_n)$) conditions hold and that there exists a sequence of real numbers $\{u_n\}$ and real numbers $z_i, \ i=1,...,m$, such that for every $H_i, \ i=1,...,m$, as in (ii) we have
    \begin{equation}\label{equ-Condition on u_n for convergence-turkman}
      1-H_i(u_n)=\frac{z_i}{n}+o\left(\frac{1}{n}\right) \quad \text{as } n \rightarrow \infty.
    \end{equation}
    Under the conditions (1a) to (1e) in \cite{turkman1983limit} their basic result is
    \begin{theorem}\label{thm-Convergence of M_n ny chain dependent process}(\cite{turkman1983limit})
        Let $\{X_n\}$, $\{u_n\}$ be sequences of random variables and real numbers respectively, satisfying the conditions (1a) to (1e) given. Then
        \begin{equation}\label{equ-Convergence of M_n ny chain dependent process}
            \lim_{n \rightarrow \infty} P(M_n\leq u_n)= e\{-\sum_{i=1}^m p_i z_i \}
        \end{equation}
        where $p_i, \ i=1,...,m$, are the stationary distribution on the chain and $z_i$ are as given in \eqref{equ-Condition on u_n for convergence-turkman}.
    \end{theorem}

    In addition for the above theorem they give a corollary regarding the properties for convergence to the specific extreme value distributions. However assumption (2.2) in their work does not hold for this study marginal distributions, which associated with the states of the chain. Therefore their specific rules for convergence does not apply here exactly.

\chapter{Model Description}\label{model description}

    We consider a slotted MIMO uplink system with $K$ independent users and one base station. The time axis is divided into fixed-length intervals referred to as time slots or, simply, slots, in the sequel. The distribution of the channel capacity a single user experiences, is time varying according to a Gilbert Elliott model \cite{gilbert1960}. We first assume that a single user, with the strongest received signal at the base station, is selected by a centralized scheduler to utilized the channel and transmit in each slot. We then consider the distributed setting as well. Hence, we first wish to analyze the expected channel capacity where the base station chooses the strongest user for transmission using the CSI sent to it from the users. We then turn to analyzing the capacity under a distributed opportunistic user selection algorithm, which select the strongest user to transmit, without coordinating between the users and without all users sending CSI to the base station. We assume the transmitters approximate their channel parameters from a pilot signal sent from the receiver.




    The linear model of the system is:
    \begin{equation}\label{equ-SystemLinearModel}
      \textit{\textbf{y}(n)=H(n)\textbf{x}(n)+\textbf{N}(n)}
    \end{equation}
    $ \textit{\textbf{y}} \in \mathbb{C}^{r} $ is the received vector from the input vector $\textit{\textbf{x}}\in\mathbb{C}^r$ where \textit{H} is a $r\times t$ complex matrix (r - number of receiving antennas, t - number of transmitting antennas) and \textit{\textbf{N}} is zero mean \textit{i.i.d.} complex Gaussian noise with independent, equal variance real and imaginary parts. The transmitter is constrained in its total power to $P$,
    \begin{equation*}
      E[\textit{\textbf{x}}^\dagger\textit{\textbf{x}}]\leq P.
    \end{equation*}

    The matrix $\textit{H}(n)$ is a random matrix which models the Rayleigh fading environment of the time dependent channel.\\

    \section{Time dependent channel}

    the channel distribution each user sees is governed by a Markov chain with two states, G (for Good) and B (for bad), where each state determines the channel distribution (Figure \ref{fig-GoodBadchannel}).
    \begin{figure}[!t]
        \centering
        \includegraphics[width=4cm]{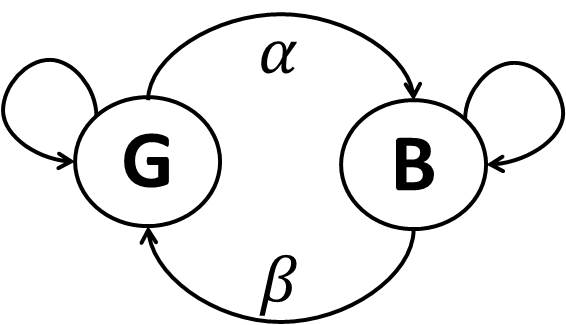}
        \caption[Markovian channel model]{A Good Bad channel model according to \cite{gilbert1960}. The capacity distribution is either good or bad Gaussian random variable according to the state of the system.}
        \label{fig-GoodBadchannel}
    \end{figure}
    Thus, the channel is memoryless conditioned on the state process, and the capacity of the $i$-th user in each slot is determined by the Good Bad Markov process $\{J(n)\}$:
    \begin{equation}\label{equ-UserCapacityProcess}
    C_i(n)=
            \begin{cases}
                N_g    & \text{when  } J_i(n)=Good \\
                N_b    & \text{when  } J_i(n)=Bad,
            \end{cases}
    \end{equation}
    where $N_g$ and $N_b$ are normal random variables with parameters $(\mu_g,\sigma_g)$ and $(\mu_b,\sigma_b)$, respectively. This is due to the Gaussian approximation for the MIMO channel capacity \cite{smith2002gaussian,chiani2003capacity}. The parameters reflect the differences in the channel qualities, that is, good channel parameters maintain $(1) \ \sigma_g>\sigma_b, \ \mu_g,\mu_b\in\mathds{R}$ or $(2) \ \sigma_g=\sigma_b$ and $\mu_g>\mu_b$.\\

    As mentioned, we assume a non trivial arrival process and that the users are not always backlogged. The arrival process of new packages to the system is characterized by a Poisson process. In addition, we assume that all user are homogenous, thus, all users has the same arrival rate $\lambda$. We present several approximate models. Thus, new packages may enter the system at any given (continuous) point on the time axis or at the beginning of a slot, depending on the approximate model. The size of the packages is set such that the transmission time is exactly one slot. We define the service time as the time from the moment a package becomes first in queue, until it is successfully transmitted. Hence, the service time of the packages is composed of the waiting time for transmission and the transmission time. In addition, we assume that services are synchronized to (i.e., can only start and end at) slot boundaries. This of course may cause collisions. 
\chapter{Capacity under time dependent channel}\label{Capacity under time dependent channel}

    In order to exploit user diversity, the user with the best channel must utilizes the channel. Therefore, the problem is in finding the distribution of the random variable $\widetilde{M_K}$, defined as:
    \begin{equation}\label{equ-CapacityDefinition}
            \widetilde{M_K}=\max\{C_1(n),C_2(n),...,C_K(n)\},
    \end{equation}
    where $C_i(n)$ is defined in \eqref{equ-UserCapacityProcess}. We are interested in the expected channel capacity $E[\widetilde{M_K}]$.

    In the beginning of each slot the channel each user experience depends on the state he exists in, resulting with two independent groups of users, a good group and a bad group. The two groups are independent due to the independency between the users, and in each separate group the users identically distributed. The size of these groups is constantly changes according to each user state process which happen in the beginning of each slot. Given the state of all users the groups size are known and we may derive the expected channel capacity in the following way
    \begin{equation}\label{equ-CapacityDefinitionTwoGroups}
        \begin{aligned}
        E[\widetilde{M_K}]=&E\big[\max\{C_1(n),C_2(n),...,C_K(n)\}\big]\\
                          =&\sum_{j=0}^K E\Big[ \max\big\{ \max\{C_1(n),...,C_j(n)\},\max\{C_{j+1}(n),...,C_K(n)\} \big\} \Big\rvert \ j \  \Big] P(j)
        \end{aligned}
    \end{equation}
    The distribution of the maximum for each group can be obtain by using EVT, although the complicated part is to do so for every group size. Unfortunately due to complicated analytical expressions the result remained open, and more details can be found in the appendices regarding this way of analysis.\\

    Nevertheless we were able to achieve a closed form for the channel capacity in more informative ways. The first one is by using the stationarity of the users' capacities sequences, and the second is to analyse the rate of exceedance above some threshold while imposing our distributed algorithm.

    \section{Capacity - Centralized Scheduling}

    In this section we give the expected channel capacity with the assumption that the strongest user utilizes the channel.
    Let us consider the distribution of each $C_i(n)$, as defined in \eqref{equ-UserCapacityProcess}, which is determined by the Good-Bad Markov process $\{J(n)\}$. 
    The stationary distribution of the chain is $(\frac{\beta}{\alpha+\beta},\frac{\alpha}{\alpha+\beta})$, which we will denote as $(p,q)$. Therefore, the distribution of $C_i(n)$ assuming a stationary distribution is as follows:
    \begin{equation}\label{equ-stationary distribition of C_i(n)}
        \begin{aligned}
            F(t)=&P(C_i(n)\leq t)= \\
            & P(N_g(n)\leq t \mid J(n)=G )P(J(n)=G)+\\
            &P(N_b(n)\leq t \mid J(n)=B )P(J(n)=B)\\
            &= pF_g(t)+qF_b(t),
        \end{aligned}
    \end{equation}
    where $F_g(t)$ and $F_b(t)$ are Gaussian distributions with parameters $(\mu_g,\sigma_g)$ and $(\mu_b,\sigma_b)$ respectively which represent the users' different channel states. The distributions parameters maintain $\sigma_g>\sigma_b, \ \mu_g,\mu_b\in\mathds{R}$ or $\sigma_g=\sigma_b$ and $\mu_g>\mu_b$. The distribution of the maximal value is in the form of
    \begin{equation*}
        P(\widetilde{M_K}\leq x)= P(C_1(n)\leq x, C_2(n)\leq x,...,C_n(n)\leq x)=F^n(x),
    \end{equation*}
    due to the independence between the users. We wish to test the behavior of $F^n(x)$as $n \rightarrow \infty$.\\

    The legitimacy to use the stationary distribution also comes from \cite{denzel1975limit} as given in \eqref{equ-Convergence to l by stationary distribution}, the convergence of $F^n(d_n)$ to a real number $l$, where $0<l<1$ and $d_n$ is a sequence of reals, exist if and only if, $P(M_n\leq d_n)\rightarrow l$ as $n \rightarrow \infty$. Thus we can analyze the expected channel capacity using EVT with the distribution above, as the marginal distribution function for each user. In order to do so, first we need to prove that convergence to one of the extreme distributions types $G(x)$ exists, and then find the constants $a_n$ and $b_n$ such that as $n\rightarrow \infty$ , $P\{\widetilde{M_K}\leq\ u_n\}\leq x\}\rightarrow G(x)$. \cite[1.6.1, 1.6.2]{EVT:Springer1983} gives necessary and sufficient conditions on the marginal distribution $F$ to belong to each of the three possible domains of attraction of the extreme value distributions. The first sufficient condition states that if $f$ has a negative derivative $f'$ for all $x$ in some interval $(x_0,x_F),\ (x_F\leq \infty),\ f(x)=0$ for $x\geq x_F$, and
    \begin{equation}\label{equ-Sufficient type 1 condition}
        \lim_{t \uparrow x_F} \frac{f'(t)(1-F(t))}{f^2(t)}=-1
    \end{equation}
    then $F$ is in the domain of attraction of Type \Rmnum{1} extreme value distribution (Gumbel distribution). The second necessary and sufficient condition for $F$ to be in the domain of attraction of Type \Rmnum{1} extreme value distribution, states that there exists some strictly positive function $g(t)$ such that
    \begin{equation}\label{equ-Necessary and sufficient type 1 condition}
     \lim_{t \uparrow x_F} \frac{1-F(t+xg(t))}{1-F(t)}=e^{-x}\\
    \end{equation}
    for all real $x$.\\

    Our main result in this context is the following.
    \begin{theorem}\label{thm-Capacity distribution in a time dependent channel convergence to a Gumbel}
        Let $(C_1,...,C_K)$ be a sequence of random variables with distribution $F(t)$ as given in \eqref{equ-stationary distribition of C_i(n)}, then the asymptotic distribution of $\widetilde{M_K}=\max{(C_1,...,C_K)}$ is of Type \Rmnum{1}. Specifically,
        \begin{equation}
                P\{a_K(\widetilde{M_K}-b_K)\leq x\}\rightarrow e^{-e^{-x}}
        \end{equation}
        where
        \begin{equation}
                a_K= \frac{\sqrt{2\log{K}}}{\sigma_g}
        \end{equation}
        and
        \begin{equation}
                b_K= \sigma_g\left(\sqrt{2\log K}-\frac{\log{\log K}+\log{\frac{4\pi}{p^2}}}{2\sqrt{2\log K}}\right)+\mu_g.
        \end{equation}
    \end{theorem}
    The proof outline is as follows. First we start by solving directly the limit in condition \eqref{equ-Sufficient type 1 condition}, while using the sandwich rule, replacing the Gaussian CDF with its corresponding complimentary error function and using the upper and lower bounds of it. The second condition \eqref{equ-Necessary and sufficient type 1 condition} has been proven by finding that $g(t)=\sigma^2/(t-\mu_g)$, along with the asymptotic analysis that the ''Bad" Gaussian distribution has little affect in it's tail comparing to the ''Good" distribution. And finally using similar asymptotic analysis as preformed in condition \eqref{equ-Necessary and sufficient type 1 condition} the normalizing constants was derived. For more details, see appendix \ref{Appendix A}, which gives the fulfillment of the two conditions above and the derivation of the normalizing constants $a_K$ and $b_K$.\\

    Using Theorem \ref{thm-Capacity distribution in a time dependent channel convergence to a Gumbel} we can now give the expected channel capacity
    \begin{equation}\label{equ-Expected channel capacity (stationary distribution)}
        \begin{aligned}
            E[\widetilde{M_K}]&= b_K+\frac{\gamma}{a_K}\\
                  &=\sigma_g\left(\left(\sqrt{2\log K}-\frac{\log{\log K}+\log{\frac{4\pi}{p^2}}}{2\sqrt{2\log K}}\right)\right.\left.\quad +\frac{\gamma}{\sqrt{2\log{K}}}\right) +\mu_g
        \end{aligned}
    \end{equation}
    where $\gamma=0.57721$ is Euler-Mascheroni constant. Thus, for $K$ sufficiently large,
    \begin{equation}\label{equ-CapacityExpression}
        E[\widetilde{M_K}] = \sigma_g\sqrt{(2\log K)}+\mu_g+\text{o}\left(\frac{1}{\sqrt{\log{K}}}\right).
    \end{equation}

    That is, for large number of users, the expectation capacity grows like $\sigma_g\sqrt{(2\log K)}$. A word is in place here concerning the values of $p$ and $q$. We assume that $(p,q)$ are bounded away from zero, which implies that the uninteresting case which all users has bad channel state, or the opposite case where all have good channel does not apply. Also, in the proof of Theorem \ref{thm-Capacity distribution in a time dependent channel convergence to a Gumbel} we used this assumption while we made a division by $p$ and $q$. Simulation results for the capacity distribution is given in Figure \ref{fig_MaxialCapacityDistributionForTimeDependentChannel}.
    \begin{figure}[h]
        \centering
        \includegraphics[width=0.8\textwidth]{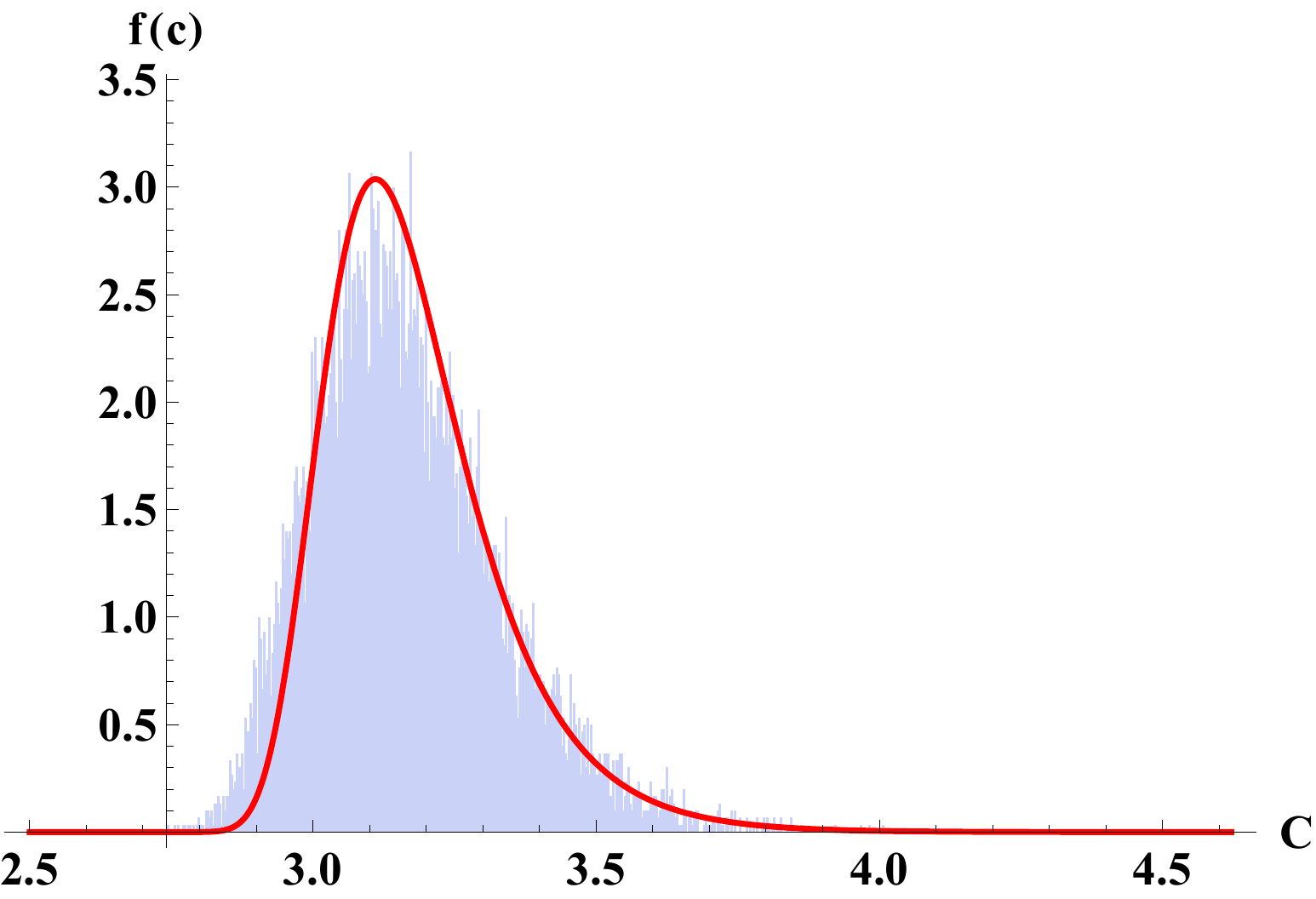}
        \caption[Simulation for capacity distribution]{Simulation for the maximal capacity distribution, when choosing the maximal capacity among 5000 capacities that following the stationary distribution where the red line is the corresponding Gumbel density with the constants $a_K$ and $b_K$. }
        \label{fig_MaxialCapacityDistributionForTimeDependentChannel}
    \end{figure}

    \begin{remark}[The probability $p$]
        The probability $p$, which is the stationary probability to exist in a good state, rules the average number of users which are in good state while considering the whole system. In each time slot we can distinguish between two groups of users, the users that are in good state and the users that are in bad state.  It is easy to show that the number of good users in the good group on average is $pK$. Hence, as $p$ grows the expected capacity grows since there are more users in good state, we can see it in the analytical result \eqref{equ-Expected channel capacity (stationary distribution)} for the channel capacity and in Figure \ref{fig-Capacity comparison as Function of p}.
    \end{remark}

    \begin{figure}[h]
        \centering
        \includegraphics[width=0.6\textwidth]{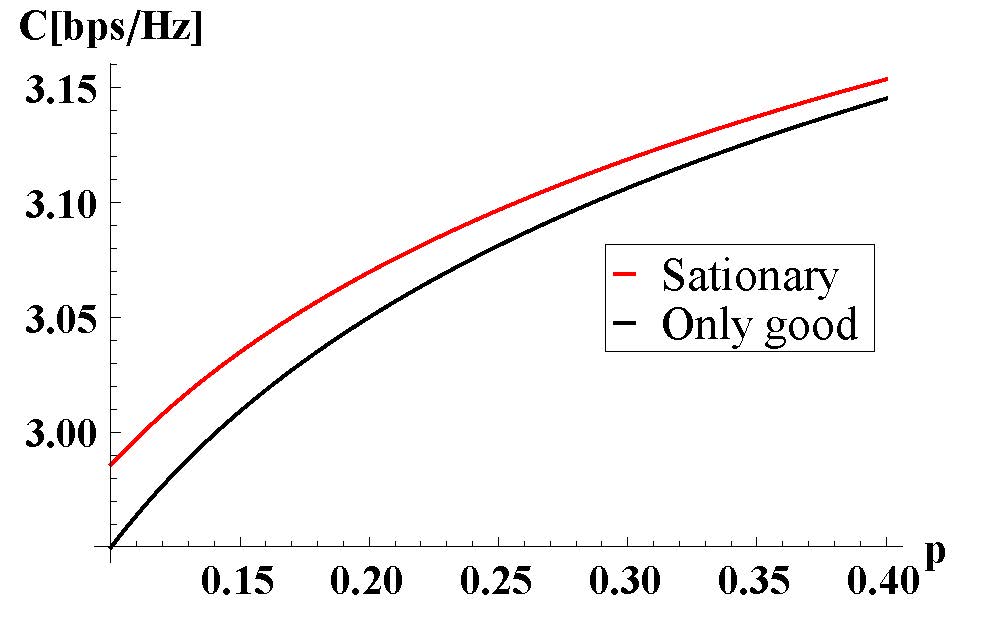}
        \caption[Capacity comparison as Function of p]{The channel capacity gain as given in \eqref{equ-Expected channel capacity (stationary distribution)} in comparison to the gain of choosing only from the good group with size $Kp$, for 5000 users as a function of $p$ and $\mu_g = \sqrt{2}, \sigma_g = 0.5$. }
        \label{fig-Capacity comparison as Function of p}
    \end{figure}

    One may wonder on why considering the bad group at all, meaning, why should a user which belongs to the bad group will be taken into account in the scheduling decision process in the beginning of a slot. Thus leaving only the users with the good channel to compete for the channel. For this case, the capacity with $K$ sufficiently large is
    \begin{equation}\label{equ-CapacityExpressionOnlyGood}
        E[\widetilde{M_{pK}}] = \sigma_g\sqrt{(2\log pK)}+\mu_g+\text{o}\left(\frac{1}{\sqrt{\log{pK}}}\right)
    \end{equation}

    In Figure \ref{fig-Capacity comparison as Function of p} we can see the influence of the bad group on the capacity gain. For rather small values of $p$ it is beneficial to schedule the strongest user from both groups. As $p$ grows the size of the good group grows as well hence the two capacity gains converge to the case were all the users are in good state. In Figure \ref{fig-gain for time dependent channel} we can see the capacity gain as a function of the number of users, specifically, in Figure \ref{fig-capacityGainFforTimeDependentChannel_p=0.2} the difference between the gains is noticeable.

   \begin{figure}
        \centering
        \begin{subfigure}[b]{0.5\textwidth}
                \centering
                \includegraphics[width=\textwidth]{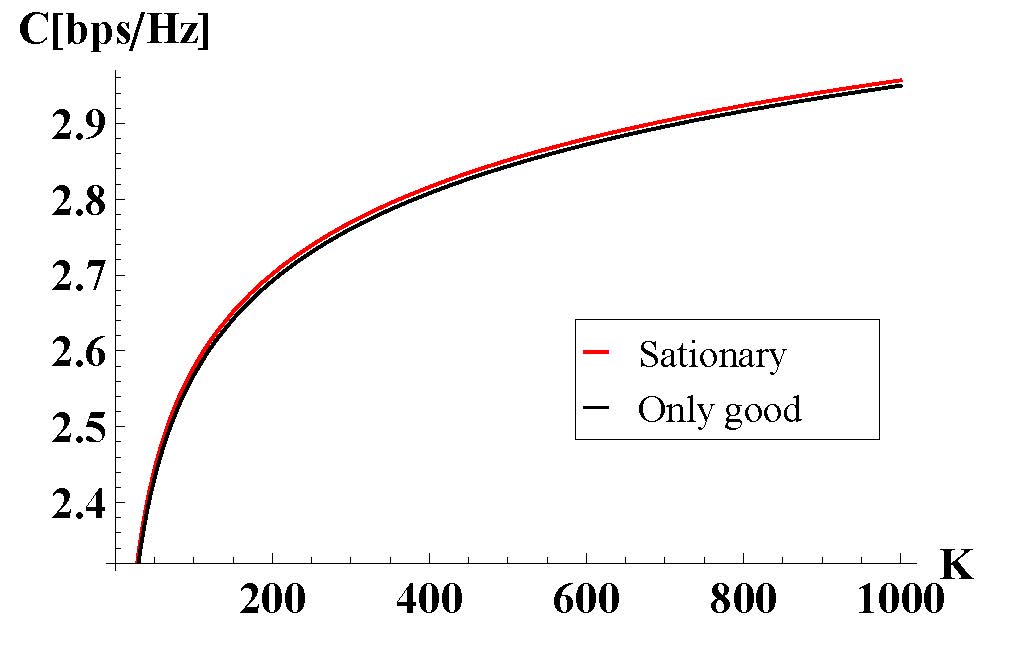}
                \caption{}
                \label{fig-capacityGainFforTimeDependentChannel_p=0.5}
        \end{subfigure}%
        ~ 
        \begin{subfigure}[b]{0.5\textwidth}
                \includegraphics[width=\textwidth]{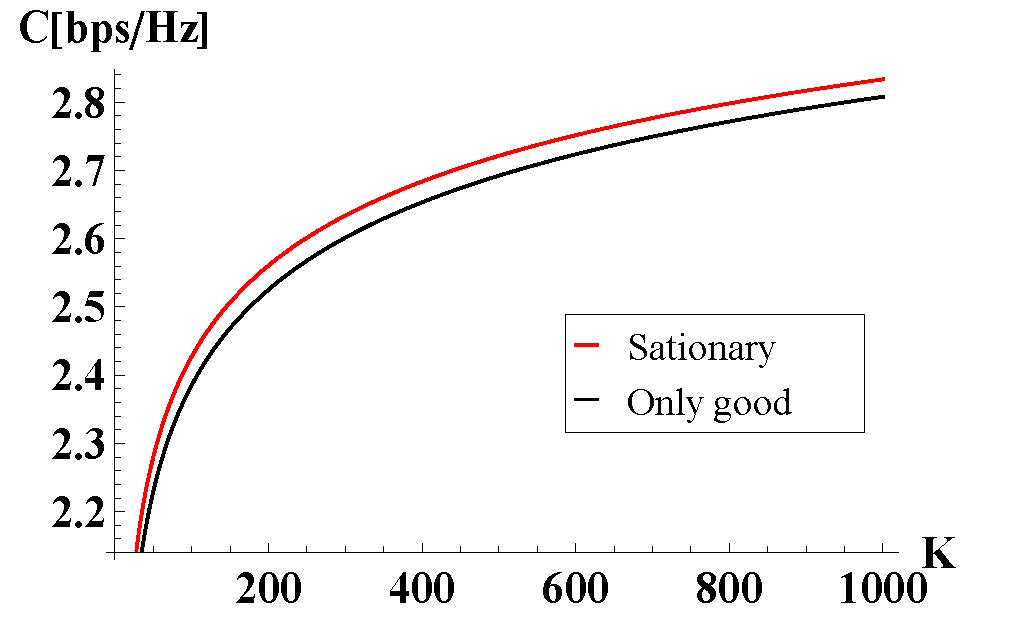}
                \caption{}
                \label{fig-capacityGainFforTimeDependentChannel_p=0.2}
        \end{subfigure}
        \caption[Capacity gain comparison]{The capacity gain for choosing the best user among all the population and only the good where $\sigma_g = 0.5$ and (a) p=0.5 (b) p=0.2 }
        \label{fig-gain for time dependent channel}
\end{figure}

    \section{Capacity - Distributed Scheduling}

    When considering a distributed scheduling algorithm, we derive the expected channel capacity using PPA. PPA is a tool for analysing threshold arrival rates and tail distributions, while imposing a threshold-based algorithm for the channel access, which was discussed earlier as an option for a distributed scheduling mechanism. This method can also be found in \cite{kampeas2014accepted} for the case of \textit{i.i.d.} and heterogeneous users. First we analyze each user separately and examining its sequence of channel capacities over time. This sequence can be modeled as a point process on some interval, using the definition in \eqref{equ-sequence of point processes N_n}, yet with the different that our process of points is formed from a dependent sequence of random variables.  As it turns out, the exceeding points above some threshold $u$ converges to a non-homogeneous Poisson process, which we discuss later, and the question that can be asked is, what is the rate of exceedance above the threshold $u$. We show that our dependent sequence converge to the same Poisson process if the sequence was of an \textit{i.i.d.} random variables, which will eventually allow us to examine the sequence of all users' capacities in a specific time.\\

    The sequence $\{C_i(n)\}$ depends on what state user $i$ exists in, which is, in turn, ruled by the underlying Markov chain process $\{J_i(n)\}$ as defined in \eqref{equ-UserCapacityProcess}. This kind of process have been studied before and is known as a $\tilde{J}-X$ process by Janssen (1969, \cite{janssen1969processus}), or "chain-dependent" process, emphasizing the fact that this is a natural extension of an \textit{i.i.d.} process. The extreme laws for chain dependent processes have showed that, convergence to one of the extreme distributions exist under some conditions. This is important for the formation of the point process and the analysis of the rate of exceedance we do in the sequel.\\

    Before we give the formulation for the chain dependent process in our work, let us remind the formulation of the original $\tilde{J}-X$ process as reviewed also in Chapter \ref{preliminaries-Chain Dependent Process}. The $\tilde{J}-X$ process as in \cite{o1974limit}, \cite{denzel1975limit} satisfies:
    \begin{equation}\label{equ-J-X process defenition}
        \begin{aligned}
          &P(\tilde{J}_n = j , X_n \leq \alpha \mid \tilde{J}_0,X_1,\tilde{J}_1,...,X_{n-1},\tilde{J}_{n-1}=i) \\
                   &= P(\tilde{J}_n=j,X_n\leq \alpha \mid \tilde{J}_{n-1}=i)=P_{ij}\tilde{H}_i(\alpha),
        \end{aligned}
    \end{equation}
    where $P$ is the transition matrix of the chain and $\tilde{H}_i(\alpha)$, where $i$ belongs to the state space, are the distribution functions associated with the chain states, respectively. Note that each state determines the distribution of $X$ for the preceding time transition. This means that given the chain process $\{\tilde{J}_n\}$ The random variables of the $\{X_n\}$ process are conditionally independent. More precisely,
    \begin{equation*}
        \begin{aligned}
            &P(X_1 \leq \alpha_1,X_2 \leq \alpha_2,...,X_n \leq \alpha_n \mid\\
            & \quad \quad \tilde{J}_0=j_0,\tilde{J}_1=j_1,...,\tilde{J}_{n-1}=j_{n-1})\\
            &=P(X_1 \leq \alpha_1 \mid \tilde{J}_0=j_0)P(X_2 \leq \alpha_2 \mid \tilde{J}_1=j_1) \cdots \\
            & \quad \quad P(X_n \leq \alpha_n \mid \tilde{J}_{n-1}=j_{n-1})\\
            &=\tilde{H}_{j_0}(\alpha_1)\tilde{H}_{j_1}(\alpha_2) \cdots \tilde{H}_{j_{n-1}}(\alpha_n).
        \end{aligned}
    \end{equation*}
    As discussed in Chapter \ref{preliminaries-Chain Dependent Process}, every stationary chain-dependent process $\{X_n\}$ is strongly mixing. This is true if the initial distribution of the chain is the stationary distribution, i.e. $P(J_0=i)=\pi_i$ for all $i$ in the finite state space, where $\pi$ is the stationary distribution. Then, the distribution of $X_n$ is $\tilde{H}(x)=\Sigma \pi_i\tilde{H}_i(x)$. This result is very important to us since we are interested in a time dependent environment and we would like to explore the point process analysis for the stationary sequences associated with it. \\

    Returning to our context, we shall consider the sequence $\{X_n\}$ as the sequence of capacities $\{C_i(n)\}$ of the i-th' user over time and the chain process $\{\tilde{J}_n\}$ as the sequence of the irreducible, aperiodic, 2-state Good-Bad Markov chain $\{J_i(n)\}$. The distributions $\tilde{H}_i(\alpha)$ will be the distributions $H_i(\alpha)$, where $i=g,b$, according to the Good-Bad sates, respectively. Since we are analysing the capacity process of one user, which is identically to the other users, we will note $\{C_n\}$ as it's capacity sequence and $\{J_n\}$ as the sequence of the chain.

    Unlike the definition of the $\tilde{J}-X$ process, in our model each state determines the capacity distribution in the current slot. So equation \eqref{equ-J-X process defenition} according to our case can be written:
    \begin{equation}\label{equ-J-C process defenition our model}
        \begin{aligned}
            P(J_n &= j , C_n \leq \alpha \mid J_0,C_1,J_1,...,C_{n-1},J_{n-1}=i) \\
               &= P(J_n=j,C_n\leq \alpha \mid J_{n-1}=i)=P_{ij}H_j(\alpha)
        \end{aligned}
    \end{equation}
    The fact that the capacity in each slot depends on the chain state in the current slot doesn't impact the strong mixing property of the sequence $\{X_n\}$.
    We will show that the $J-C$ process as defined to meet our model, is still strongly mixing by definition \ref{def-StrongMixing}, while using the proof guidelines in \cite{o1974limit}.
    \begin{lemma}\label{lem-C_n is strongly mixing}
        $\{C_n\}$ is strongly mixing with $g(k)=\sum_{j} \pi_i \mid P_{ij}^k -\pi_j \mid$, where $\pi$ is the stationary distribution of the chain and $P_{ij}^k=(P^k)_{ij}$.
    \end{lemma}
    \begin{proof}
        Let $A$ and $B$ be as definition \ref{def-StrongMixing}. Then
        \begin{equation*}
            \begin{aligned}
                 & \mid P(A \cap B)-P(A)P(B) \mid  \\
                 & \leq \sum_{i,j\in\{0,1\}} \mid P(A \cap B,J_m=i,J_{m+k}=j)-\\
                 & \quad \quad P(A,J_m=i)P(B,J_{m+k}=j) \mid \\
                 &=\sum_{i,j\in\{0,1\}} P(A|J_m=i)P(B|J_{m+k}=j)P(J_m=i) \\
                 & \quad \quad \mid P(J_{m+k}=j|J_m=i)-P(J_{m+k}=j) \mid \\
                 &\leq  \sum_{i,j\in\{0,1\}} \pi_i \mid P_{ij}^k -\pi_j \mid = g(k).\\
            \end{aligned}
        \end{equation*}
        For each $i \in {0,1}$, $\sum_{j} \pi_i \mid P_{ij}^k -\pi_j \mid \rightarrow 0$ as $k \rightarrow \infty$, this was also used in \cite{o1974limit} for showing the strongly mixing property, nevertheless, it is easy to notice that in fact as $k \rightarrow \infty$, $(P^k)_{ij} \rightarrow \pi_j$ regardless to $i$.
     \end{proof}

     Considering Theorem \ref{thm-Point process convergance of stationary process}, we will show that the point process under the chain dependent model convergence to a Poisson process just like one obtained from an \textit{i.i.d.} sequence with the same marginal d.f. $F$, which in our case is $H(x)$. The normalizing constants will be as given in Theorem \ref{thm-Capacity distribution in a time dependent channel convergence to a Gumbel} which maintains the condition $1-H(u_n)=P\{X>u_n\}\sim \uptau/n$ as $n\rightarrow\infty$.
     As mentioned before, the strong mixing condition hold for our sequence of channel capacity, and as a result so does the condition $D(u_n)$ as defined in definition \ref{def-D condition} which is a weakened case of strong mixing \cite{leadbetter1974extreme}. We would like to show that condition $D'(u_n)$ as defined in definition \ref{def-D' condition} also holds so we would be able to characterize the rate of exceedance over the threshold $u_n$. Note that we are only interested in a sequence of reals $\{u_n\}$ which satisfies $1-H(u_n)=\uptau/n+o(1/n)$, in the fulfillment of the $D'(u_n)$ condition.

     \begin{lemma}\label{lem-condition D'(u_n) holds on C_n}
           The local dependence condition $D'(u_n)$ holds for the sequence $\{C_n\}$ as defined in \eqref{equ-UserCapacityProcess}.
    \end{lemma}
    \begin{proof}
        \begin{equation*}
            \begin{aligned}
                &\limsup_{n \rightarrow \infty} n\sum_{r=2}^{ \lfloor n/k\rfloor } P(C_1>u_n,C_r>u_n) \\
                &\overset{(a)}{\leq} \limsup_{n \rightarrow \infty} n\sum_{r=2}^{ \lfloor n/k\rfloor } \sum_{i,j\in\{0,1\}} P(C_1>u_n | J_1=i)\\
                &\quad \quad \quad \cdot P(C_r>u_n | J_r=j)  P(J_1=i)P(J_r=j | J_1=i) \\
                &= \limsup_{n \rightarrow \infty} n\sum_{r=2}^{ \lfloor n/k\rfloor } \sum_{i,j\in\{0,1\}} (1-H_i(u_n))(1-H_j(u_n))\pi_i P_{ij}^r\\
                &\overset{(b)}{\leq} \limsup_{n \rightarrow \infty} n\sum_{r=2}^{ \lfloor n/k\rfloor } \sum_{i,j\in\{0,1\}} \left(\frac{\uptau}{n\pi_i}+o\left(\frac{1}{n}\right) \right)\\
                &\quad \quad \quad \cdot \left(\frac{\uptau}{n\pi_j}+o\left(\frac{1}{n}\right) \right) \pi_i P_{ij}^r\\
                &\leq \limsup_{n \rightarrow \infty} n\sum_{r=2}^{ \lfloor n/k\rfloor } \sum_{i,j\in\{0,1\}}\frac{1}{\pi_j\pi_i} \left(\frac{\uptau}{n}+o\left(\frac{1}{n}\right) \right)^2  \pi_i P_{ij}^r\\
                &\leq \limsup_{n \rightarrow \infty} n \left\lfloor \frac{n}{k}\right\rfloor \left(\frac{\uptau}{n}+o\left(\frac{1}{n}\right) \right)^2 \sum_{i,j\in\{0,1\}}\frac{1}{\pi_j\pi_i} \pi_i \max_r\{P_{ij}^r\}\\
                &=(\uptau^2+o(1))^2 \frac{1}{k} \sum_{i,j\in\{0,1\}}\frac{1}{\pi_j\pi_i} \pi_i \max_r\{P_{ij}^r\}  \rightarrow 0 \text{  as  } k \rightarrow \infty
            \end{aligned}
        \end{equation*}

    In the above chain, (a) is since once $J_1$ is known, $C_1$ and $C_r,J_r$ are independent, then we conditioned on $J_r$. (b) is true since $u_n$ maintains $n(1-F_1(u_n))=\uptau +o(1)$ so we have
        \begin{equation*}
         \begin{aligned}
              &1-F_1(u_n)=1-H(u_n)=1-\sum_{i} \pi_i H_i(u_n)\\
              &=\sum_{i} \pi_i(1- H_i(u_n))= \frac{\uptau}{n}+o\left(\frac{1}{n}\right)
         \end{aligned}
        \end{equation*}
        and therefore
        \begin{equation*}
              \pi_l(1-H_l(u_n))=\frac{\uptau}{n}-\sum_{i,i\neq l} \pi_i(1- H_i(u_n))\leq \frac{\uptau}{n} +o\left(\frac{1}{n}\right).
        \end{equation*}
        Note that $\pi_l$ and $\uptau$ are constants.
     \end{proof}

     Thus, in our paradigm a single users' channel capacity process has the same laws of convergence as if the sequence $\{C_n\}$ was \textit{i.i.d.} with marginal distribution $H(x)$. Therefore, since the users are independent and each user sees the same marginal distribution $H(x)$, we can analyze the point process of the sequence of all users' capacities at a specific time (e.g. time slot), resulting in the basic case of \textit{i.i.d.} random variables, as given in Theorem  \ref{thm-Point process convergance of iid}. Considering the above, we now turn to evaluate the expected channel capacity.\\

     \subsection{Distributed Algorithm}
     The distributed algorithm for the channel access is based on the algorithm suggested in \cite{qin2006distributed}. Given the number of users, we set a capacity threshold $u$ such that only a small fraction of the users will exceed it. At the beginning of each slot, each user estimates its capacity for that slot. If the capacity anticipated by the user is greater than the capacity threshold, it transmits in that slot. Otherwise, it keeps silent. Choosing the appropriate threshold value is of great significance. A lower value means more users exceedances, therefore more collisions, while a high value means more idle slots. The optimal threshold value is set such that one user exceeded the threshold on average, and as a result, its transmission is successful. We treat this slot as a utilized slot. Hence the expected channel capacity has the form:
     \begin{equation*}
        C_{av}(u)=P_r(\text{utilized slot})E[C|C>u]
     \end{equation*}

     In order to calculate $E[C|C>u]$, the expected capacity experienced by a user who passed $u$, it is important to understand that the average capacity is ruled by the values of points that exceeds the threshold, so one needs to evaluate the distance of the exceeding points from the threshold. \cite{Kampeas2012scheduling} give the analytical tools to compute the tail distribution of the exceeding points. These points follow the generalized Pareto distribution, hence by using the PPA exceedance rate results, and since we showed that we have the same exceedance as an \textit{i.i.d.} case with $H(x)$ as the marginal distribution, the result is the same and we have
     \begin{equation*}
        E[C|C>u]=u+\frac{1}{a_K}+o(\frac{1}{a_K}),
     \end{equation*}
     where $a_K$ is the normalizing constant as in Theorem \ref{thm-Capacity distribution in a time dependent channel convergence to a Gumbel}.

     We say that a slot is utilized if only one point out of all $K$ points exceeds the threshold. Hence, as $K \to \infty$
     \begin{equation*}
        K\left(\frac{1}{K}\right)\left(1-\frac{1}{K}\right)^{K-1}\to e^{-1},
     \end{equation*}
     where the threshold $u$ was chosen such that $1-H(u)=1/K$. We will elaborate on this value in the next subsection.
     The expected channel capacity thus
     \begin{equation}\label{equ-Expected channel capacity (PPA)}
        C_{av}(u)=e^{-1}\left(u+\frac{1}{a_K}+o\left(\frac{1}{a_K}\right)\right)
     \end{equation}
     In the next subsection we show that $u=b_K$ so when comparing the capacity result in \eqref{equ-Expected channel capacity (PPA)} to the expected capacity $E[\widetilde{M_K}]$ as given in \eqref{equ-Expected channel capacity (stationary distribution)}, for the centralized approach, we see that we have the same scaling laws, only that the distributed approach is smaller only by a factor of $e^{-1}$. The normalizing constant $a_K$ along with the threshold $u=b_K$  is as given in Theorem  \ref{thm-Capacity distribution in a time dependent channel convergence to a Gumbel}, reflect the time dependency of the channel as explained before.

     \subsection{Threshold Estimation}

     The optimal threshold was set such that only one user on average exceeds it in each time slot. This is due to the analysis of \cite{qin2003exploiting}, where it has been shown that the optimal transmission probability for each user to transmit in a specific time slot is $\alpha(n)/n$, where asymptotically $\alpha(n) \rightarrow 1$ hence, the probability $1/n$ is a good approximation. By using this rule we can estimate the optimal threshold value for the time dependent channel as well:
     \begin{equation*}
        1-H(u)=\frac{1}{K}\\
     \end{equation*}
     Hence,
     \begin{equation*}
        1-pF_g(u)-qF_b(u)=\frac{1}{K},\\
     \end{equation*}
     where the above represent the probability that a user capacity will exceed the threshold $u$. With similar derivation as in appendix \ref{Appendix A} we get that
     \begin{equation}\label{equ-Estimated threshold}
        \begin{aligned}
            u=&\sigma_g\sqrt{2\log{K}}\left( 1-\frac{\frac{1}{2}\log{\frac{4\pi}{p^2}}+\frac{1}{2}\log{\log{K}}}{2\log{K}}+o\left( \frac{1}{\log{K}} \right) \right)\\
            & \quad +\mu_g\\
            =&b_K
        \end{aligned}
     \end{equation}
     Figure \ref{fig_Point_Process_Sim} gives simulation results for the users' capacities point process, and the exceedance above the threshold $u$ as given in \eqref{equ-Estimated threshold}. One can see that as the population grows, the threshold value also grows in a way which maintains the average probability of exceedance $1/K$. We can see that in Sub-Figure \ref{fig_Point_Process_Sim100} no exceedance occur among all users where in the other cases one user did manage to exceed alone.

    \begin{figure}
        \centering
        \begin{subfigure}[b]{0.4\textwidth}
                \centering
                \includegraphics[width=\textwidth]{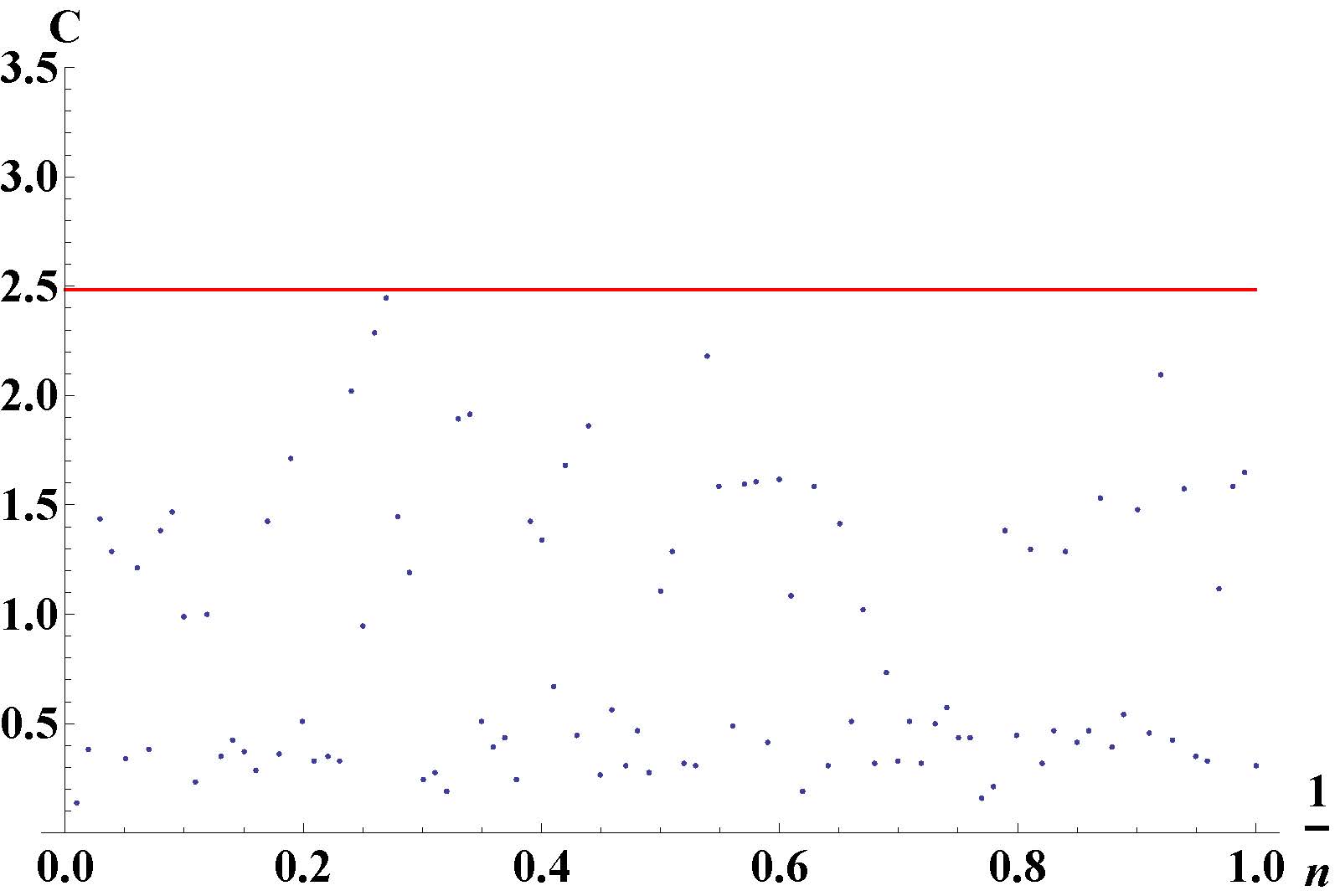}
                \caption{}
                \label{fig_Point_Process_Sim100}
        \end{subfigure}%
        ~ 
        \begin{subfigure}[b]{0.4\textwidth}
                \includegraphics[width=\textwidth]{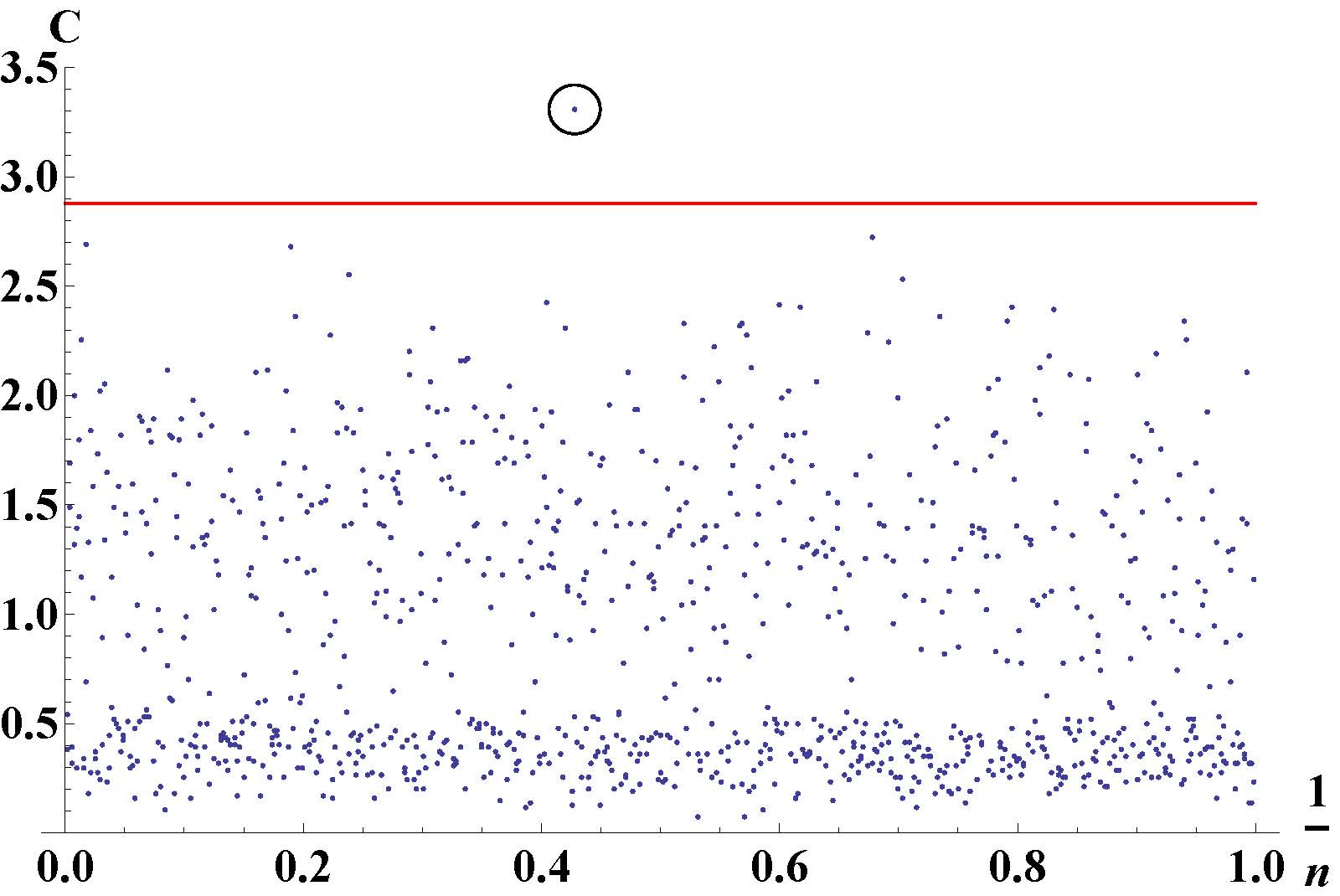}
                \caption{}
                \label{fig_Point_Process_Sim1000}
        \end{subfigure}
        ~ 

        \begin{subfigure}[b]{0.4\textwidth}
                \includegraphics[width=\textwidth]{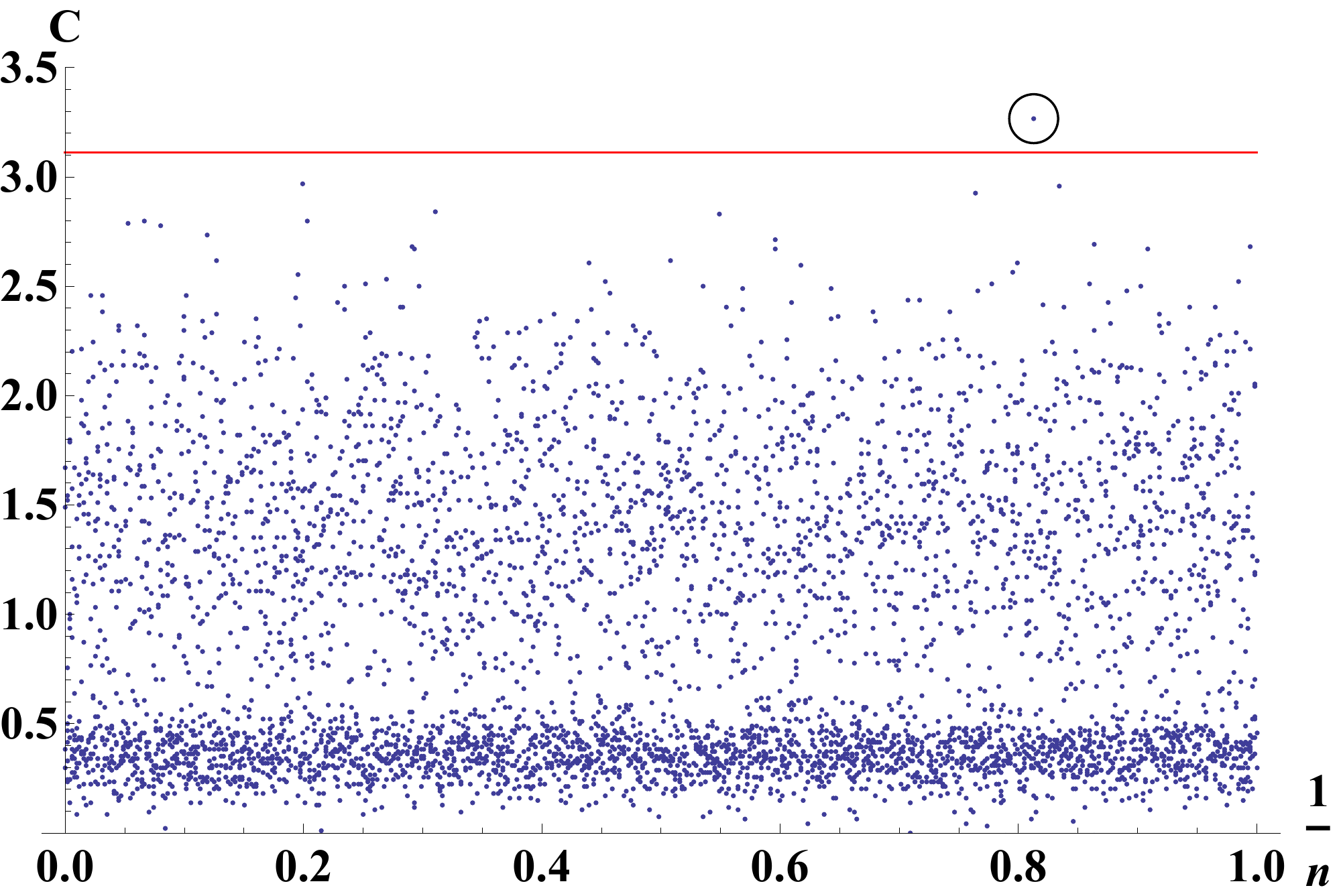}
                \caption{}
                \label{fig_Point_Process_Sim5000}
        \end{subfigure}%
        ~ 
        \begin{subfigure}[b]{0.4\textwidth}
                \includegraphics[width=\textwidth]{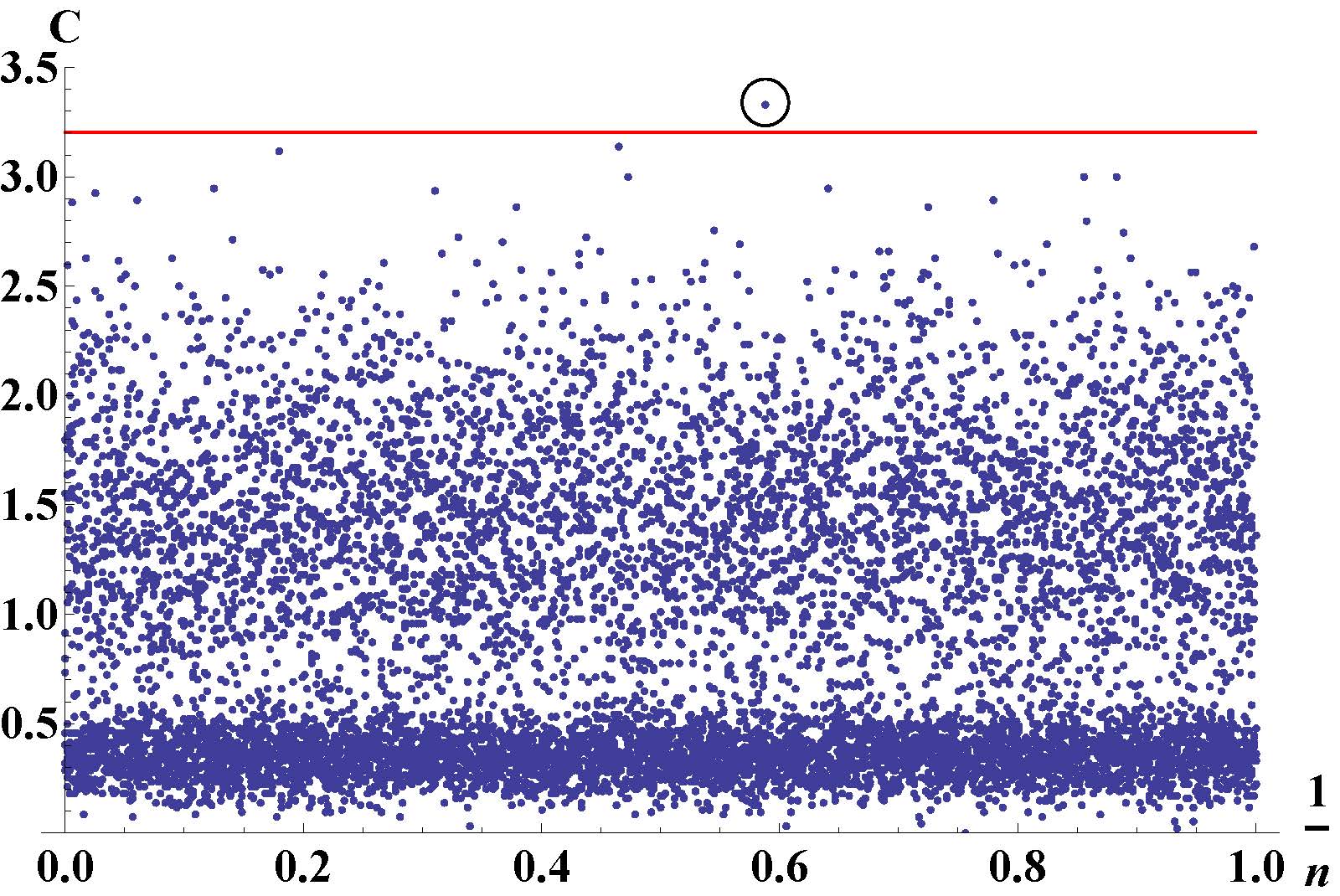}
                \caption{}
                \label{fig_Point_Process_Sim10000}
        \end{subfigure}
        \caption[Point process simulation]{Point process simulation for the exceedance of the threshold $u=b_K$ in a time dependent environment where(a) K=100 (b) K=1000 (c) K=5000 (d) K=10000}
        \label{fig_Point_Process_Sim}
    \end{figure}

    \subsection{Refinement On The Threshold Value}
    While considering \textit{i.i.d.} random variables as the sequence of channel capacity, we can give a more accurate value for the threshold such that only one user exceeds it on average. The threshold given in the previous chapter \eqref{equ-Estimated threshold}, such that $u_n=b_K$, was formed due to the calculations of the mixed Gaussian distribution \eqref{equ-stationary distribition of C_i(n)} which was the marginal distribution for a specific channel sample. As a reminder we are searching $u_n$ such that
    \begin{equation*}
        1-F(u_n)=\frac{1}{K}\\
    \end{equation*}
    where in this chapter $F$ is CDF of normal distribution with parameters $(\mu,\sigma)$, which has a known complementary function in contrast to the mixed Gaussian distribution.\\
    In \cite{kampeas2012capacity} a use in the complimentary error function was done in order find threshold $u_k$ such that $k$ users will exceed it among the $K$ users, in our case $k=1$, hence we have
    \begin{equation}\label{equ-accurate value for the threshold}
      \begin{aligned}
        u_1&=\mu+\sqrt{2}\sigma erfc^{-1}\left(\frac{2}{K}\right)\\
           &=\mu+\sigma\sqrt{2\log(K)-\log\left[-2\pi\left(2\log\left(\frac{1}{K}\right)+\log2\pi\right)\right]}+o\left(K^2\right).
      \end{aligned}
    \end{equation}

\chapter{Performance analysis - distributed algorithm}\label{Delay and QoS under the distributed algorithm}

    In the previous chapter, we concerned ourselves with the question of what is the expected channel capacity in a distributed scheduling system based on the exceedance of some threshold $u_K$. This analysis touched just one part of the information communication chain, as we can ask relevant questions regarding delay and quality of service of the information, before it leaves the transmitter. \\

    $K$ users wish to transmit data to a shared access point. Users are not necessarily fully backlogged, and each of them has a Poisson arrival process with rate $\lambda$. Clearly this means a user may not have any information to send in a certain slot. We focus our attention on the users' queues, as illustrated in Figure \ref{fig-QueuingSystem}, and, specifically, characterize the service process for each user.
    \begin{figure}[h]
            \centering
            \includegraphics[width=2.5in]{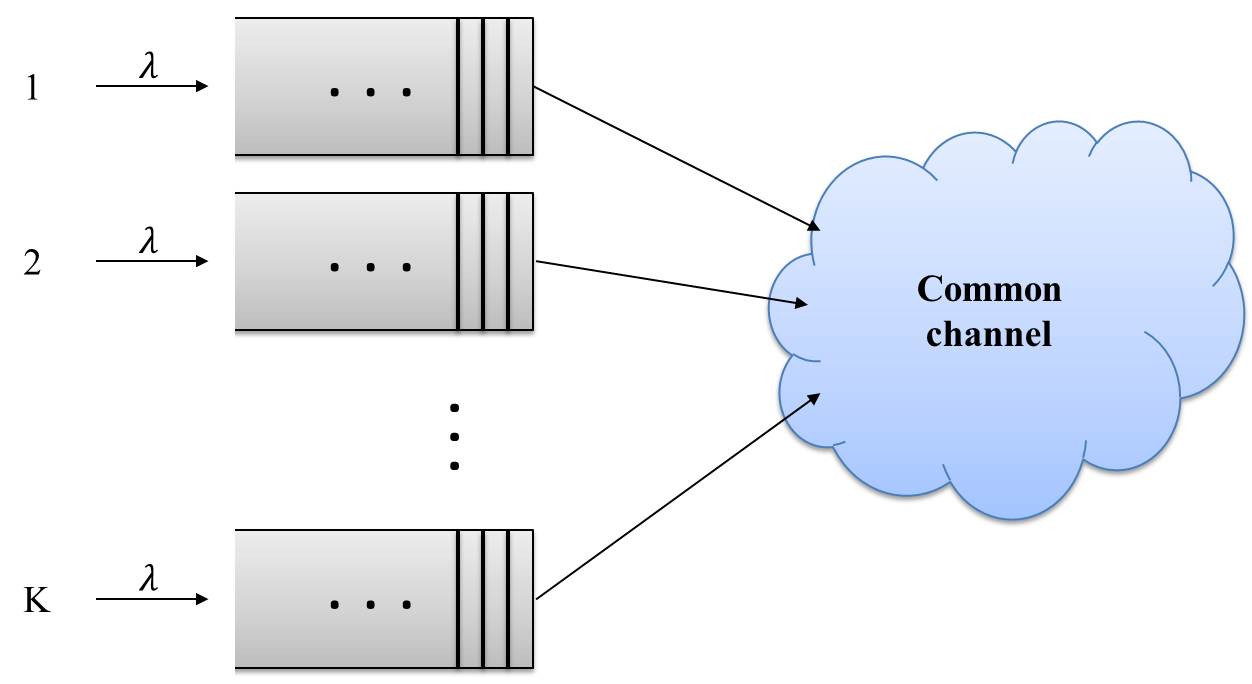}
            \caption[Model of queueing system]{System model. $K$ users access a common channel. Each user has a packet arrival process with rate $\lambda$}
            \label{fig-QueuingSystem}
    \end{figure}
    The users scheduling is based on the distributed algorithm, where a user checks if its channel capacity is above some threshold $u_K$, and transmits for one slot only if it exceeds it. The threshold is chosen such that \emph{the probability of exceedance equals $1/K$}. Assume for now that this probability is a constant variable.

    In the following, we investigate the behaviour of the system, starting with referring our queueing model to the known model of slotted aloha. We then review a known approximate model, which captures the system performance, and extend it to fit our model. Finally, we present two additional approaches, under different assumptions, to approximate this queuing system by the analogy to known types of queues. We do so by considering the process of threshold exceedance, which has Poisson properties and is asymptotically considered as a Poisson process.\\
    In the first two subsections, we assume the users experience a time \emph{independent channel}. Therefore, the rate for exceedance is the same for all users. In the last subsection we consider time dependency, and the possibility that users are in different channel states, hence, we have two different exceedance rates (as there is only value of the threshold).

    \section[Queueing Approximate model \Rmnum{1} ]{Approximation by System and Users' State}\label{Approximate model 1}
    It is not hard to see the resemblance of our model to the well explored slotted aloha system. Yet, the main difference is that in slotted aloha, the transmission scheme used is "immediate first transmission" i.e., if user $i$ has an empty queue when a package arrives, it may transmit the package instantaneously, while in the case of collision, it may transmit again with a retransmission probability $p_i$. Our model follows a "delayed first transmission" scheme, where channel transmission can only happen with probability $p_i$ (i.e., the threshold exceedance probability). When considering such a system consisting, of $K$ queues, each with a memoryless arrival process yet a common server (i.e., the communication channel), the interdependence between the queues is very strong and cannot be ignored. It is reflected, among others, in the users' queue length, probabilities for successful transmission or collision and the delay of packages. Several works in the literature explored the behavior of such systems, with the same terms of QoS. Due to the interdependence between the queues, each of these works considered a different simplified mathematical model, which resulted with an approximation for the required metrics. We elaborate on one specific work \cite{ephremides1987delay}, which introduced a relevant approximate model. Although this work introduced a general model for both "immediate" and "delayed" first transmission, the analytical derivations and numerical results refers only to the "immediate first transmission" scheme, and not the "delayed first transmission" scheme we are interested in. We will briefly review this model and emphasize the differences in the mathematical derivations due to the different schemes.

    \subsection{Review on the approximate model}
    The model presented in \cite{ephremides1987delay} assumed that the arrival rate at user $i$ is $\lambda_i$, and the arrival processes are statistically independent between the users. Time is slotted and it takes exactly one slot to transmit one packet. Thus, $\lambda_i$ is the probability of arrival for user $i$ in any given slot. User $i$ attempts, with probability $p_i$, to transmit the head-of-the-line packet in the queue (if the latter is nonempty) through a common channel to the receiver.

    At the beginning of a given slot the users' status may be one of the three categories: \emph{idle}, \emph{active} or \emph{blocked}. A user is idle if there are no packets in its queue at the end of the preceding slot. It is blocked if its queue is not empty and the latest attempted transmission was unsuccessful. It is active if its queue is not empty but its most recent attempted transmission was successful. We then let
    \begin{equation}\label{equ-probability for transmission queueing }
    p_i= \left\{
          \begin{array}{l l}
            r_i\lambda_i & \ \text{if $i$ is idle},\\
            s_i & \ \text{if $i$ is active},\\
            q_i & \ \text{if $i$ is blocked}.
    \end{array} \right.
    \end{equation}

    As can be seen, the setting of the problem is similar to the one considered in this study when $r_i=s_i=q_i=\frac{1}{K}$. As mentioned, in \cite{ephremides1987delay}, the mathematical analysis and numerical results considered the case where $r_i=s_i=1, q_i=p_i$, which is the original slotted aloha paradigm. Specifically, the model consist of a coupled Markov chains, the system-status chain and the queue-length chain. The transition probabilities of each chain, as will be presented later, depend on the steady-state probabilities of the other chain and thus, all state equations for both chains must be solved simultaneously.
    The system-status chain tries to capture the state of each user in any given time, hence, the status variable $\overline{S}$ consists of $K$ ternary variables, $S_1,S_2,...,S_K$, each of which indicates the status of the corresponding terminal. Namely,
    \begin{equation}\label{equ-system chain ternary variable S }
    S_i= \left\{
          \begin{array}{l l}
            0 & \ \text{if $i$ is idle},\\
            1 & \ \text{if $i$ is active},\\
            2 & \ \text{if $i$ is blocked}.
    \end{array} \right.
    \end{equation}
    The total number of states achievable by this vector is given by
    \begin{equation}\label{equ-total number of states in system chain}
      \sum_{i=0}^K \binom{K}{i}+ K\sum_{i=0}^{K-1} \binom{K-1}{i}=\sum_{i=0}^K \binom{K}{i}i=2^{K-1}(K+2)
    \end{equation}
    where $i$ indicates the number of blocked users, along with the observation which no more than one active user may be present in the system in any given time. The new calculations of the transition probabilities is according to the "delayed first transmission" scheme, which are different from the one presented in \cite{ephremides1987delay} displayed in appendix \ref{Appendix D}. Two quantities which are needed for the calculation of the transition probabilities are:
    \begin{equation}\label{equ-p(1|1) and p(0|2) definition}
      \begin{array}{l}
            P_i(1\mid 1) \triangleq P_r(\text{queue size $>1\mid$ user $i$ is active})\\
            P_i(0\mid 2) \triangleq P_r(\text{queue size $=1\mid$ user $i$ is blocked}).
          \end{array}
    \end{equation}
    Note that these are not the transition probabilities of a specific chain, rather, they reflect the coupling between the two chains. Along with these quantities, one can calculate $P(\overline{S})$, the joint steady-state probability distribution of the random vector $\overline{S}$.

    The queue-length Markov chain tracks both the status and queue length of each user, independently of the status and queue length of the other users. Specifically, the pair $(T_i,N_i)$ represents the state of the user and $\pi(T_i,N_i)$ denotes its steady-state probability. $N_i$ is the total number of packets at queue $i$ and
    \[ T_i= \left\{
    \begin{array}{l l}
            0 & \ \text{if $i$ is blocked},\\
            1 & \ \text{if $i$ is unblocked}.
    \end{array} \right.\]
    The transition probabilities of the chain depend on the following average success probabilities:
    \begin{equation}\label{equ-Avevrage success probabilities model 1}
    \begin{array}{l}
            P_B(i) =  P_r(\text{success $\mid$ user $i$ is blocked})\\
            P_A(i) =  P_r(\text{success $\mid$ user $i$ is active})\\
            P_I(i) =  P_r(\text{success $\mid$ user $i$ is idle}),
    \end{array}
    \end{equation}
    where the averaging is performed over the status of the other users (the status of the system). Namely, in order to calculate the probabilities \eqref{equ-Avevrage success probabilities model 1} one needs the stationary distribution of the system-status chain $P(\overline{S)}$. This leads to the coupling of the two sets of equations. The calculations of this average success probabilities are also different from the one presented in \cite{ephremides1987delay}, and can be found in appendix \ref{Appendix E}.
    The values which are in special interest to us, in order to help understand the system performance, are
    \begin{equation}\label{equ-probability for blocked and empty}
      \pi(0,0)=\frac{\lambda_i\overline{P}_I(i)}{\lambda_iP_A(i)+\overline{\lambda}_iP_B(i)}\pi(1,0)
    \end{equation}
    \begin{equation}\label{equ-probability for idle and empty}
      \pi(1,0)=\frac{\overline{\lambda}_iP_B(i)-\lambda_i\overline{P}_A(i)}{\overline{\lambda}_iP_B(i)-\lambda_i(P_I(i)-P_A(i))}
    \end{equation}
    \begin{equation}\label{equ-probability for active and not empty}
      \pi(1,1)=\frac{\lambda_i}{\overline{\lambda}_i}\pi(0,0)
    \end{equation}
    \begin{equation}\label{equ-probability to be blocked}
      G_{0}^i(1)=\frac{\lambda_i\overline{\lambda}_i\overline{P}_I(i)}{\overline{\lambda}_iP_B(i)-\lambda_i(P_I(i)-P_A(i))}
    \end{equation}
    \begin{equation}\label{equ-probability to be unblocked}
      G_{1}^i(1)=\lambda_i+\overline{\lambda}_i\frac{\overline{\lambda}_iP_B(i)-\lambda_i\overline{P}_A(i)}{\overline{\lambda}_iP_B(i)-\lambda_i(P_I(i)-P_A(i))}
    \end{equation}
    where \ref{equ-probability to be blocked} and \ref{equ-probability to be unblocked} are defined as the conditional moment generating function of the chain as follows
    \begin{equation}\label{equ-conditional moment generating function}
      G_{T_i}^i(z)\triangleq \sum_{N_i=0}^{\infty} \pi(T_i,N_i)z^{N_i}, \ \ \ \ T_i=0,1
    \end{equation}
    and $\overline{\lambda}=1-\lambda, \overline{P}=1-P$. The expressions presented here are the same ones as presented in \cite{ephremides1987delay}, and also the mathematical development to attain them is the same.
    In order to find the steady state of the chains, as explained, the state equations must be solved simultaneously. The way to do so is by an iterative process, each time given the auxiliary quantities $P_i(1\mid 1),P_i(0\mid 2)$ and $P_I(i),P_A(i),P_B(i)$, as shown in Figure \ref{fig-IterativeScheme}.
    The iterative method used is Wegstein's iteration scheme, and after achieving satisfying convergence, one may calculate the performance metrics using the steady state of the chains. More details can be found in \cite{saadawi1981analysis,ephremides1987delay}.

    \begin{figure}[h]
        \centering
        \includegraphics[width=3in]{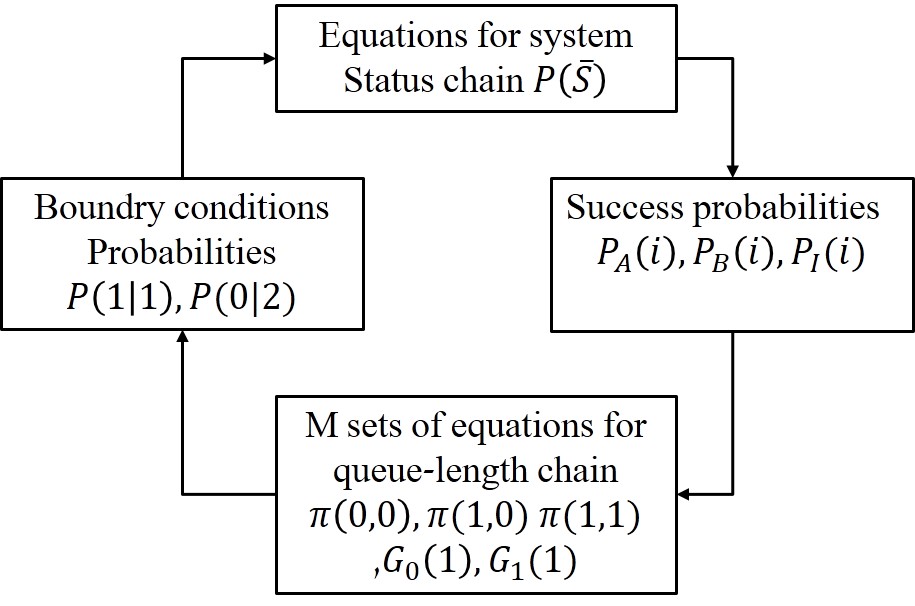}
        \caption[Iterative algorithm]{Iterative algorithm for the solution of the coupled sets of equations for the state probabilities of the system chain and the queue length chains.}
        \label{fig-IterativeScheme}
    \end{figure}

    \subsubsection{The performance metrics}\label{Approximate model 1 - Queueing performance analysis}
    The delay a package endures from the moment it is generated and arrives to a user's queue until it is successfully transmitted, consists of three components, the time in line $W_q(i)$, the service time $W_s(i)$ and transmission time. Each of them is approximated separately and the sum of the three gives the total average delay. The following expression is the delay at a user $i$ \cite{ephremides1987delay}.
    \begin{equation}\label{equ-Delay of user i model 1}
      D_i=W_q(i)+W_s(i)+1
    \end{equation}
    The $1$ added represents the transmission time, which takes one slot. The service time is approximated by
    \begin{equation}\label{equ-service time approximation model 1}
      W_s(i)=\frac{G_0(1)}{1-\overline{\lambda}_i\pi(1,0)} \cdot \frac{1}{P_B(i)},
    \end{equation}
     where the service time, as defined, is the time a package exists in the head of the queue. In case a successful transmission happens in the immediate arrival of a package, the service time is zero. In any other case, meaning the user is in blocked state and has packages in his queue, we counts the number of slots until a successful transmission, i.e., a geometric random variable. Note that in \cite{ephremides1987delay}, the service time was approximated differently, considering idle slots as zero service time. Clearly, there is no meaning to service time with an empty queue, therefore, service time should be \emph{normalized with the probability of not being in idle state with no arrival of a package}. The time in line is calculated by using Little's result, hence
    \begin{equation}\label{equ-time in line approximation model 1}
      W_q(i)=\frac{L_i}{\lambda_i}
    \end{equation}
    where $L_i$ is the average queue length of a user (without considering the blocked head-of-line packet as part of the queue), which is given by
    \begin{equation}\label{equ-mean queue size model 1}
      L_i=\frac{\lambda_i^2\overline{\lambda}_i\overline{P}_I(i)}{(\overline{\lambda}_iP_B(i)-\lambda_i\overline{P}_A(i))(\overline{\lambda}_iP_B(i)-\lambda_i(P_I(i)-P_A(i)))}.
    \end{equation}
    The total weighted average system delay $D$ is therefore
    \begin{equation*}
      D=\frac{\sum_{i=1}^{K}D_i\lambda_i}{\sum_{i=1}^{K}\lambda_i}
    \end{equation*}
    where $\overline{\lambda}=1-\lambda, \overline{P}=1-P$.
    The probability for success given that the user is not idle, i.e., has package to send, and exceeds the threshold is given by
    \begin{equation}\label{equ-success probability model 1}
      p_{succ}(i)= \frac{P_A(i)(G_1^i(1)-\pi(1,0))+P_I(i)\pi(1,0)+P_B(i)G_0^i(1)}{p_i(1-\pi(1,0))}
    \end{equation}
    where the numerator is the success probability regardless the users' state, divided by the probability that exceedance occur with a non empty queue.

    \begin{figure}[H]
        \centering
        \begin{subfigure}[b]{0.5\textwidth}
                \centering
                \includegraphics[width=\textwidth]{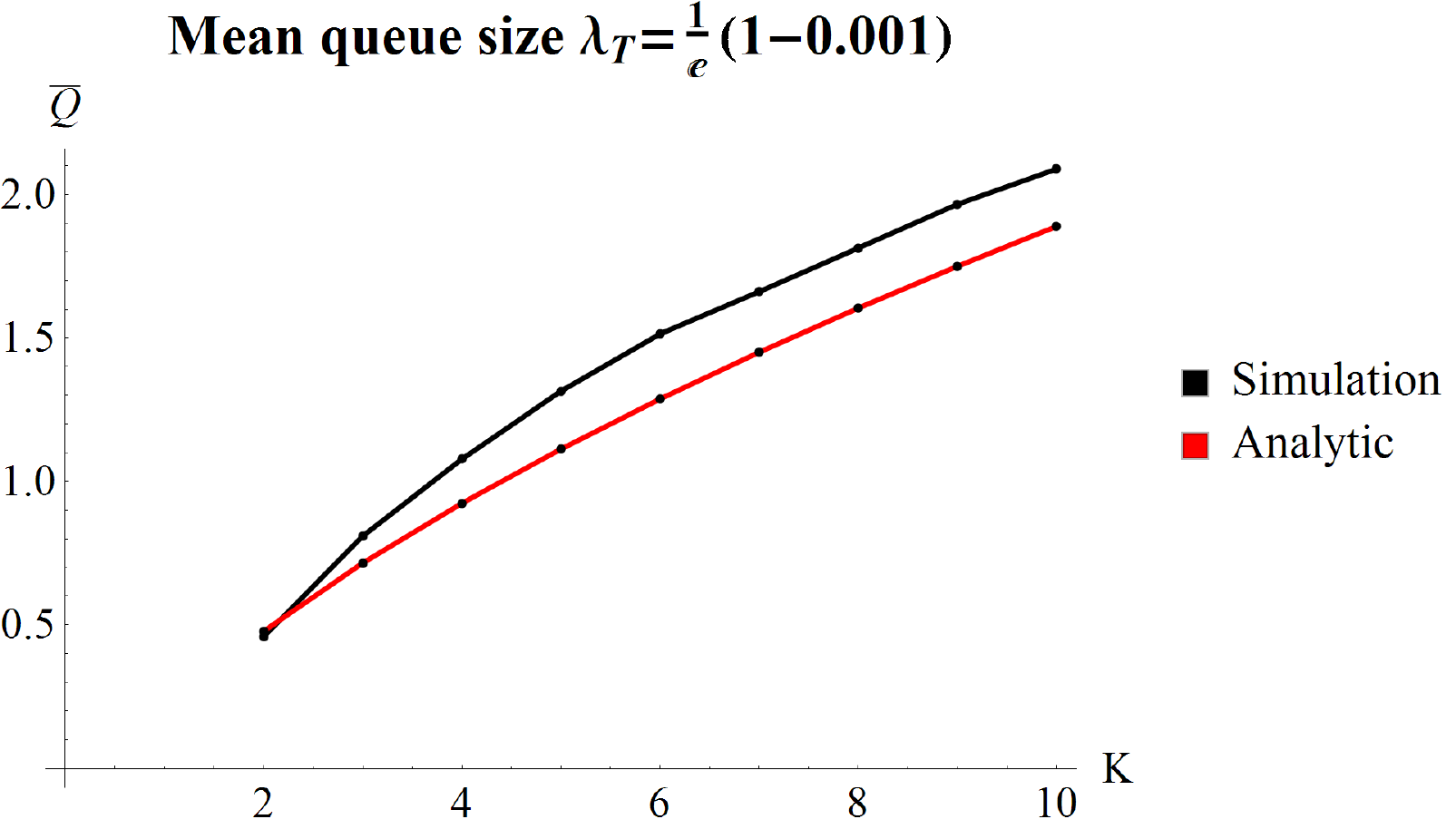}
                \caption{}
                \label{fig-MeanQueueSize_2-10_0366}
        \end{subfigure}%
        ~
        \begin{subfigure}[b]{0.5\textwidth}
                \includegraphics[width=\textwidth]{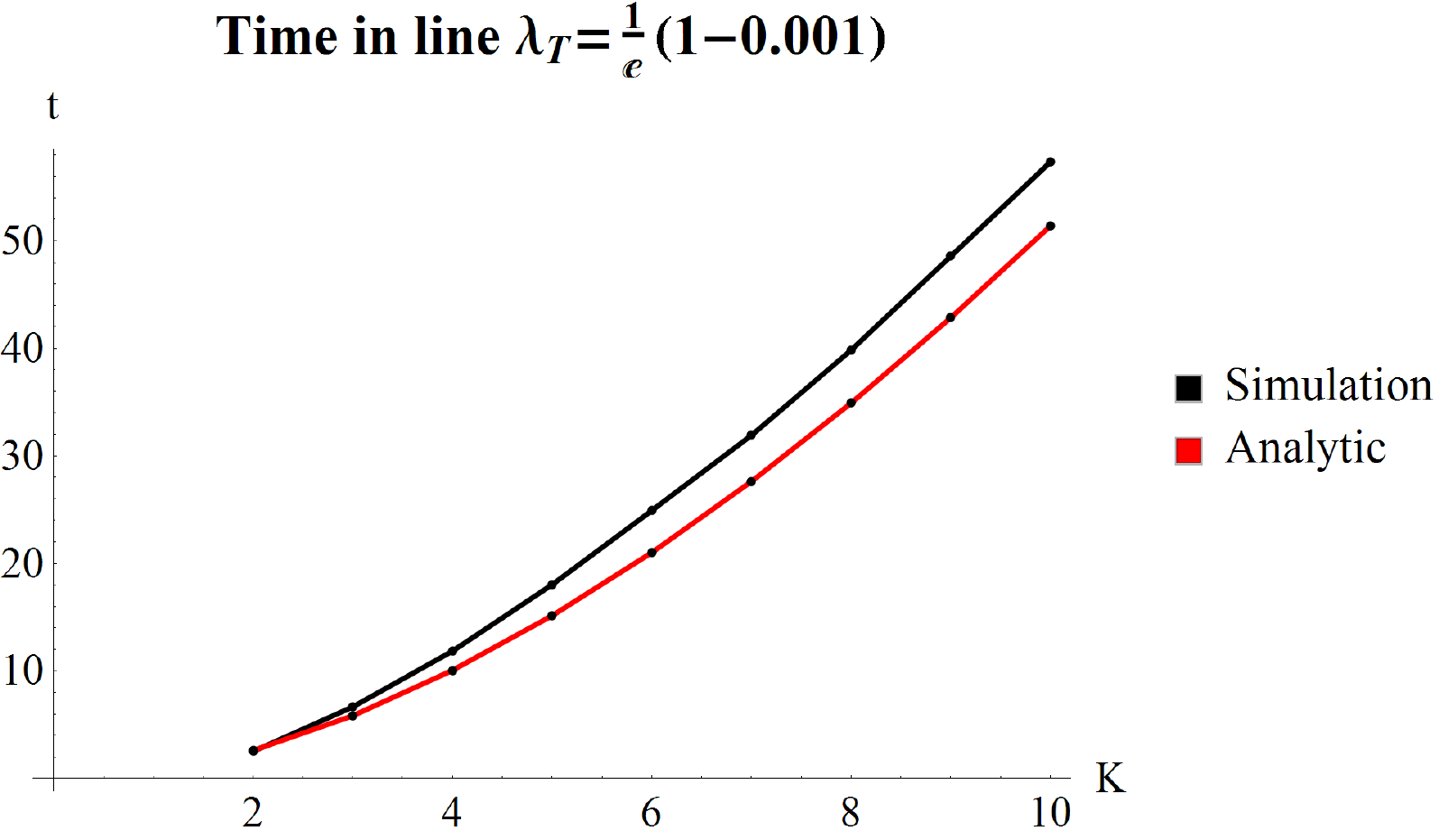}
                \caption{}
                \label{fig-TimeInLine_2-10_0366}
        \end{subfigure}
        \hfil

        \begin{subfigure}[b]{0.5\textwidth}
                \includegraphics[width=\textwidth]{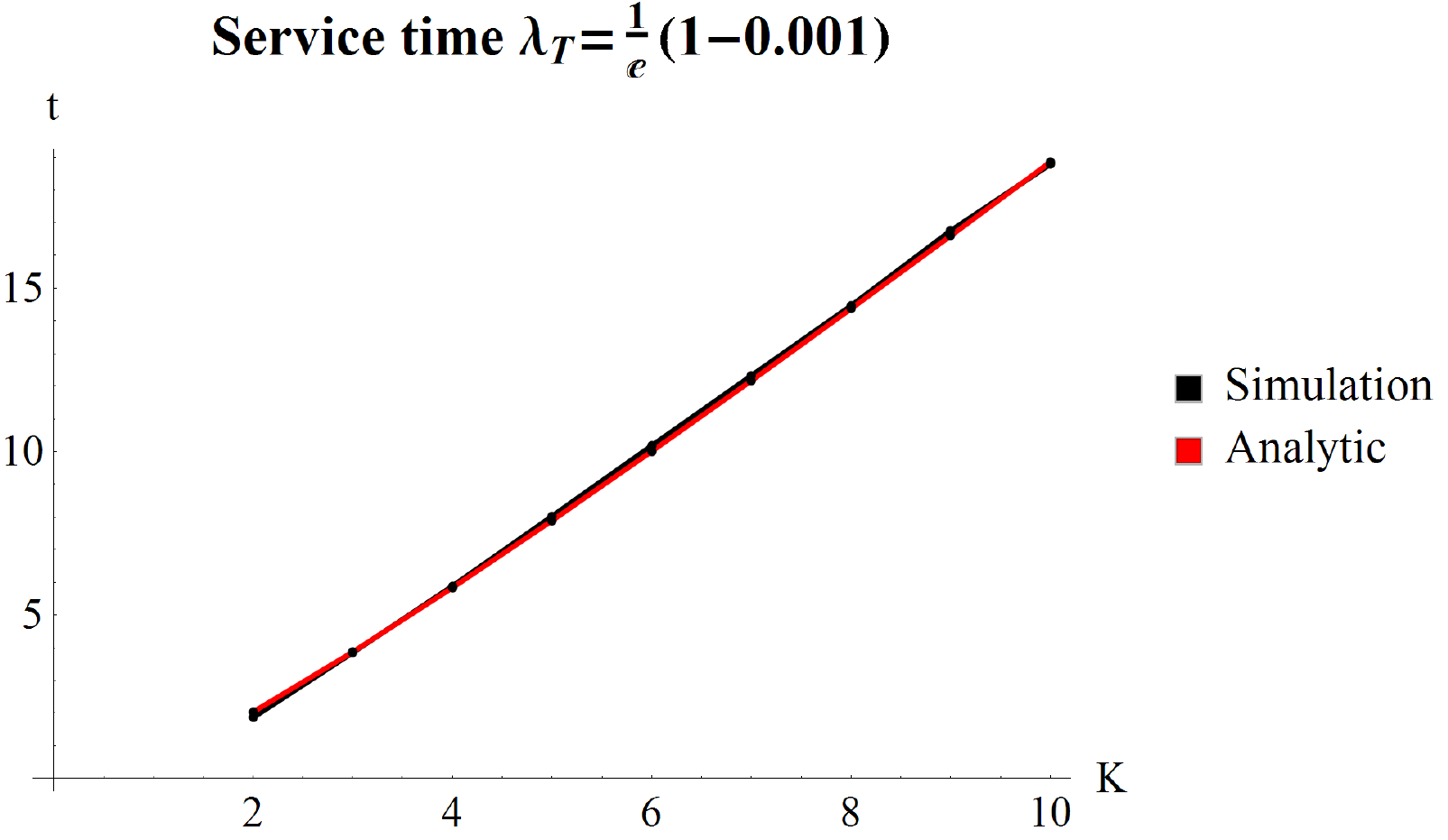}
                \caption{}
                \label{fig-ServiceTime_2-10_0366}
        \end{subfigure}%
        ~
        \begin{subfigure}[b]{0.5\textwidth}
                \includegraphics[width=\textwidth]{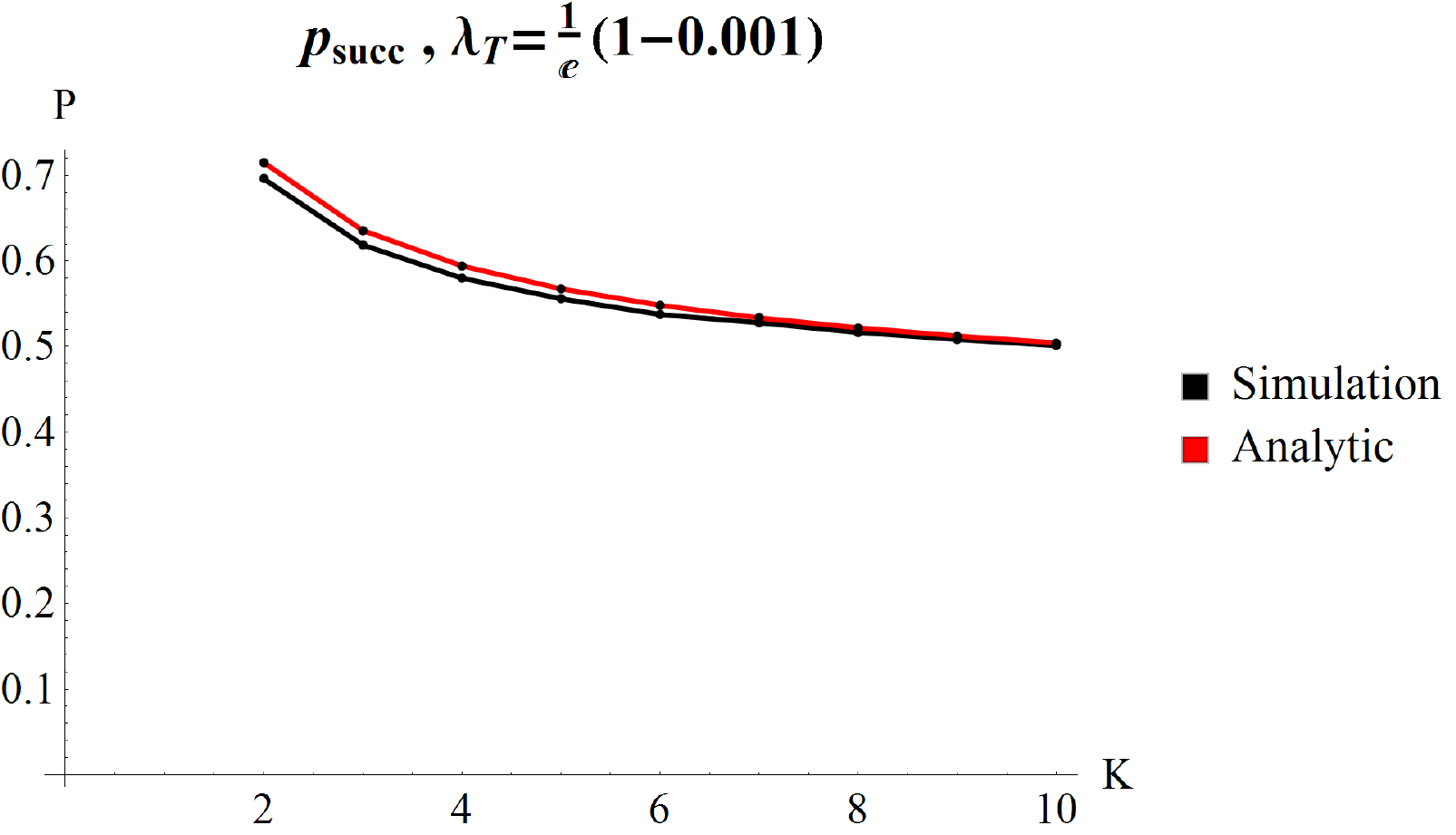}
                \caption{}
                \label{fig-SuccessProbability_2-10_0366}
        \end{subfigure}
        \caption[System performance time independent model]{Anaclitic results and simulation for the system performance as a function of the number of users. Total arrival rate is $\lambda_T=\frac{1}{e}(1-0.001)$. (a) The mean queue size. (b) The time in line. (c) The service time. (d) The probability for success. The red line describes the anaclitic expressions \eqref{equ-mean queue size model 1}, \eqref{equ-time in line approximation model 1}, \eqref{equ-service time approximation model 1}, \eqref{equ-success probability model 1} respectively.}
        \label{fig-QueuingPerformance_2-10_0366}
    \end{figure}

    \subsection{Performance analysis and comparisons}\label{Queueing performance analysis}
    We present numerical results under "delayed first transmission" scheme. We divide the delay to its different components, i.e., the time in queue, service time, average queue size, system throughput and, in addition, the success probability $p_{succ}$ given that a user has a package to send and exceeded the threshold. Due to the fact that there is no exact solution for the slotted aloha system, a comparison was made only with simulation results for this system. As explained earlier, the number of states in the system-status Markov chain grows exponentially with the number of users, therefore, the simulation and the numerical results were compared for $2$ to $10$ users. The performance parameters were also checked as a function of the arrival rates.

    The simulation and analytic calculations were preformed under the settings of a symmetric system. The arrival process and the threshold exceedance process were approximated by a Bernoulli process, due to the discrete nature of the system. Figure \ref{fig-QueuingPerformance_2-10_0366} depicts the performance metrics as a function of the number of users, where the total arrival rate is close to maximum value which still allows the system to be stable. The results depict very good match between the simulations and the analytic results given in \eqref{equ-service time approximation model 1}, \eqref{equ-time in line approximation model 1}, \eqref{equ-mean queue size model 1} and \eqref{equ-success probability model 1}.

     Very important result is the fact that the probability for successful transmission, $p_{succ}=1-p_{coll}$, behaves exactly as predicted by the analytical model (Figure \ref{fig-SuccessProbability_2-10_0366}), which means that it could be regarded as a \emph{constant average probability regardless of the strong interdependence between the queues}. This discovery is the foundation for the next two subsections, which, in contrast to the model considered thus far, are able to capture the system's behavior for large number of users.

    \begin{figure}[H]
        \centering
        \begin{subfigure}[b]{0.5\textwidth}
                \centering
                \includegraphics[width=\textwidth]{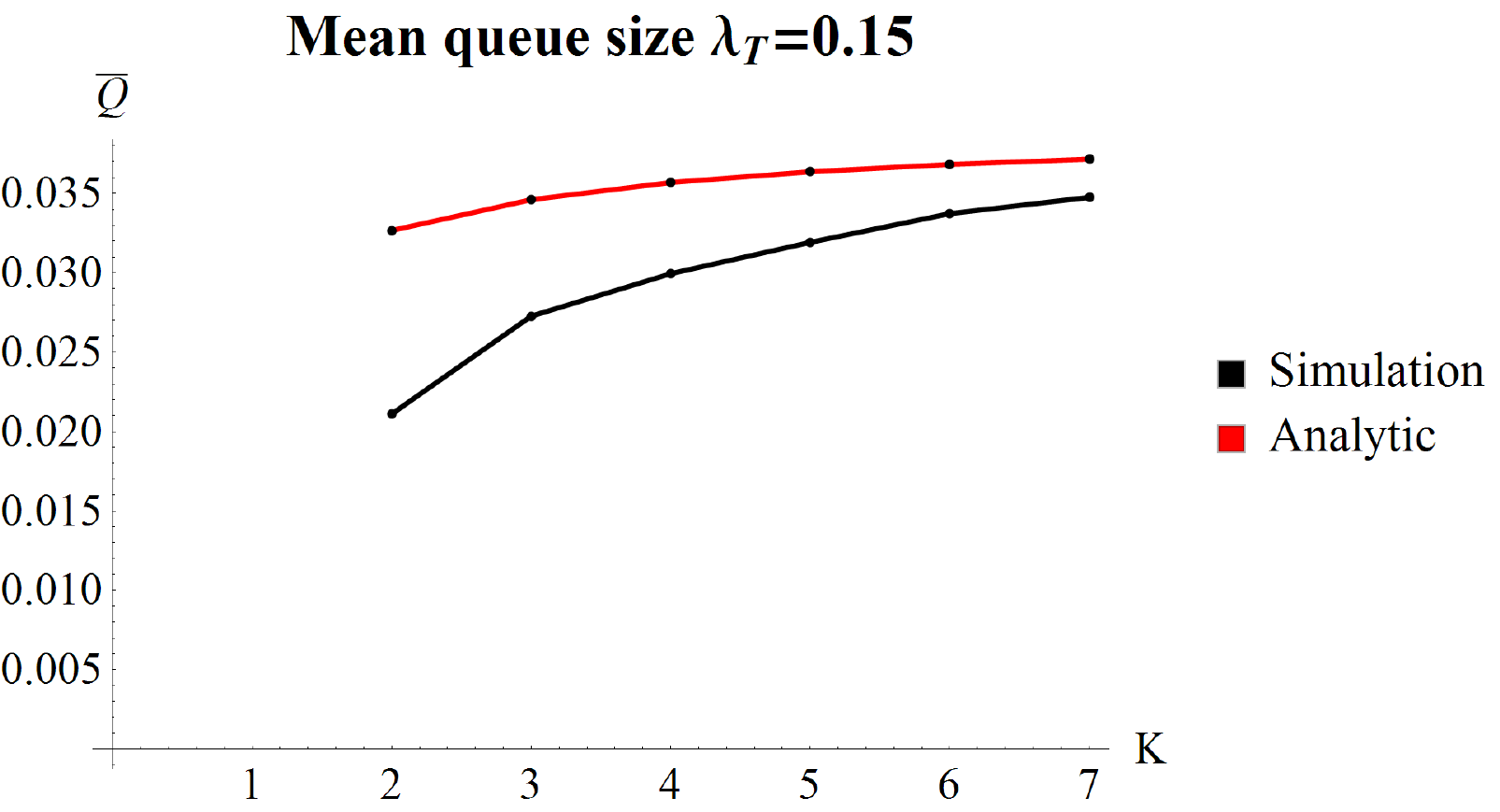}
                \caption{}
                \label{fig-MeanQueueSize_2-7_015}
        \end{subfigure}%
        ~
        \begin{subfigure}[b]{0.5\textwidth}
                \includegraphics[width=\textwidth]{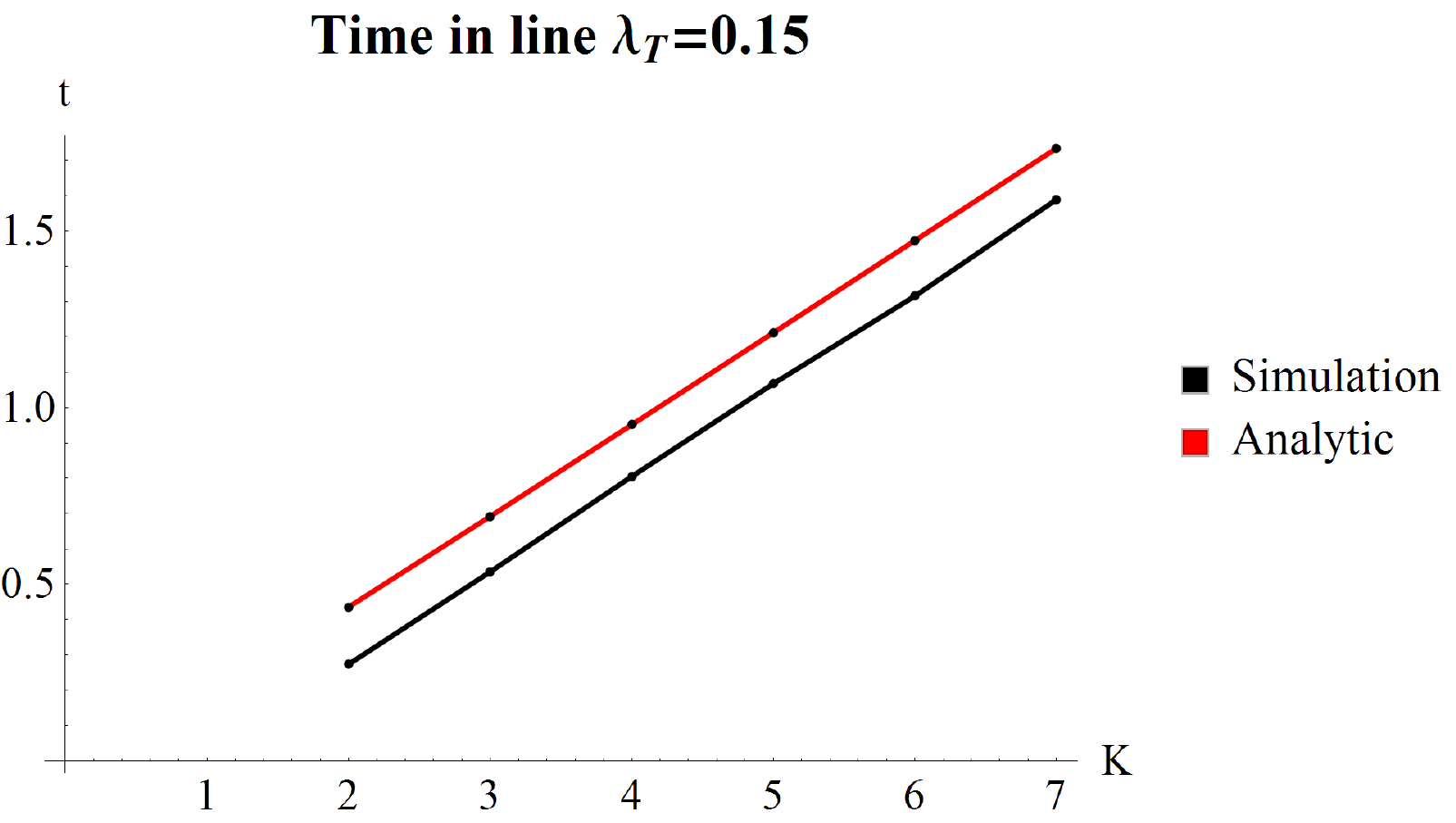}
                \caption{}
                \label{fig-TimeInLine_2-7_015}
        \end{subfigure}
        \hfil

        \begin{subfigure}[b]{0.5\textwidth}
                \includegraphics[width=\textwidth]{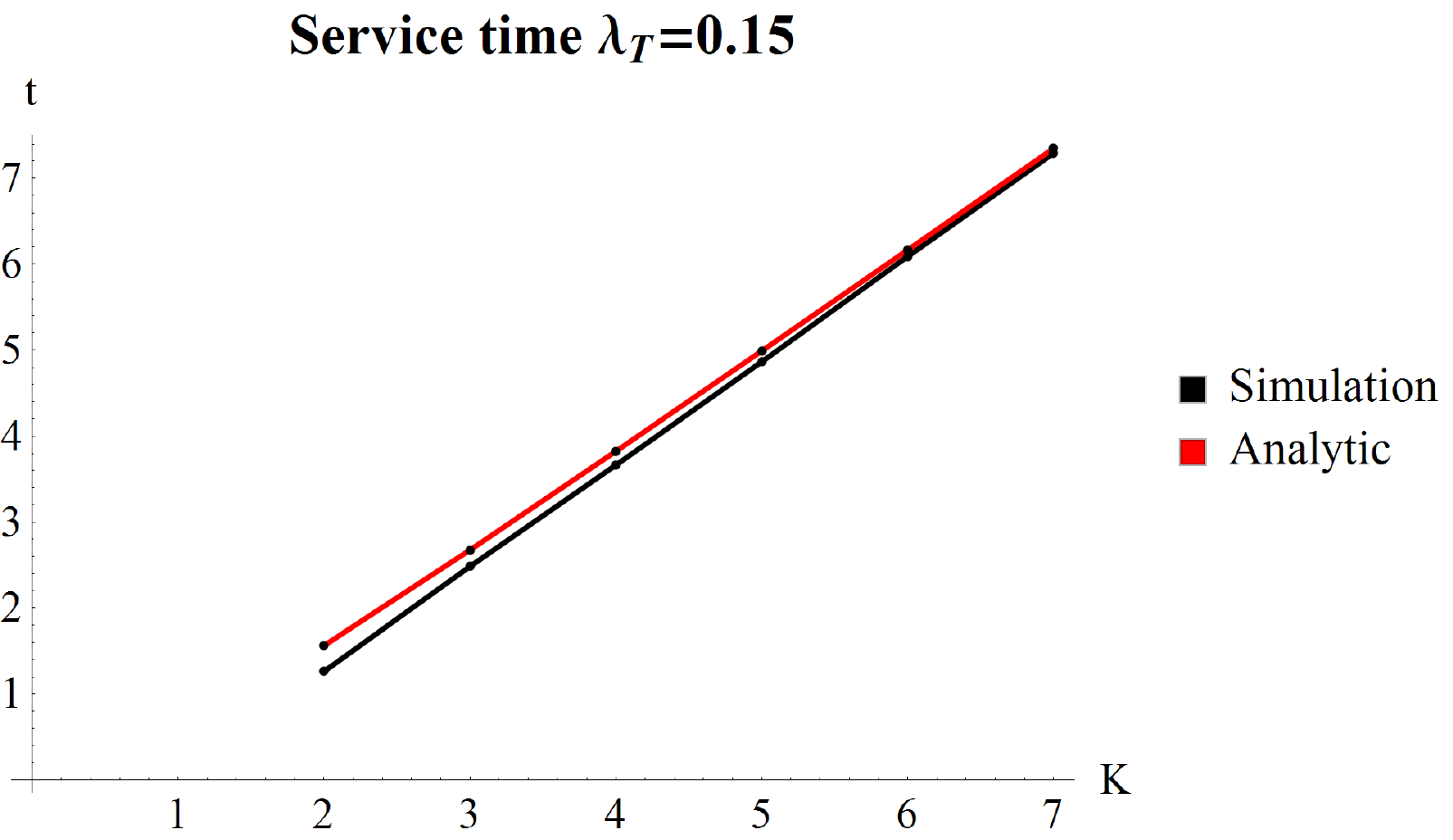}
                \caption{}
                \label{fig-ServiceTime_2-7_015}
        \end{subfigure}%
        ~
        \begin{subfigure}[b]{0.5\textwidth}
                \includegraphics[width=\textwidth]{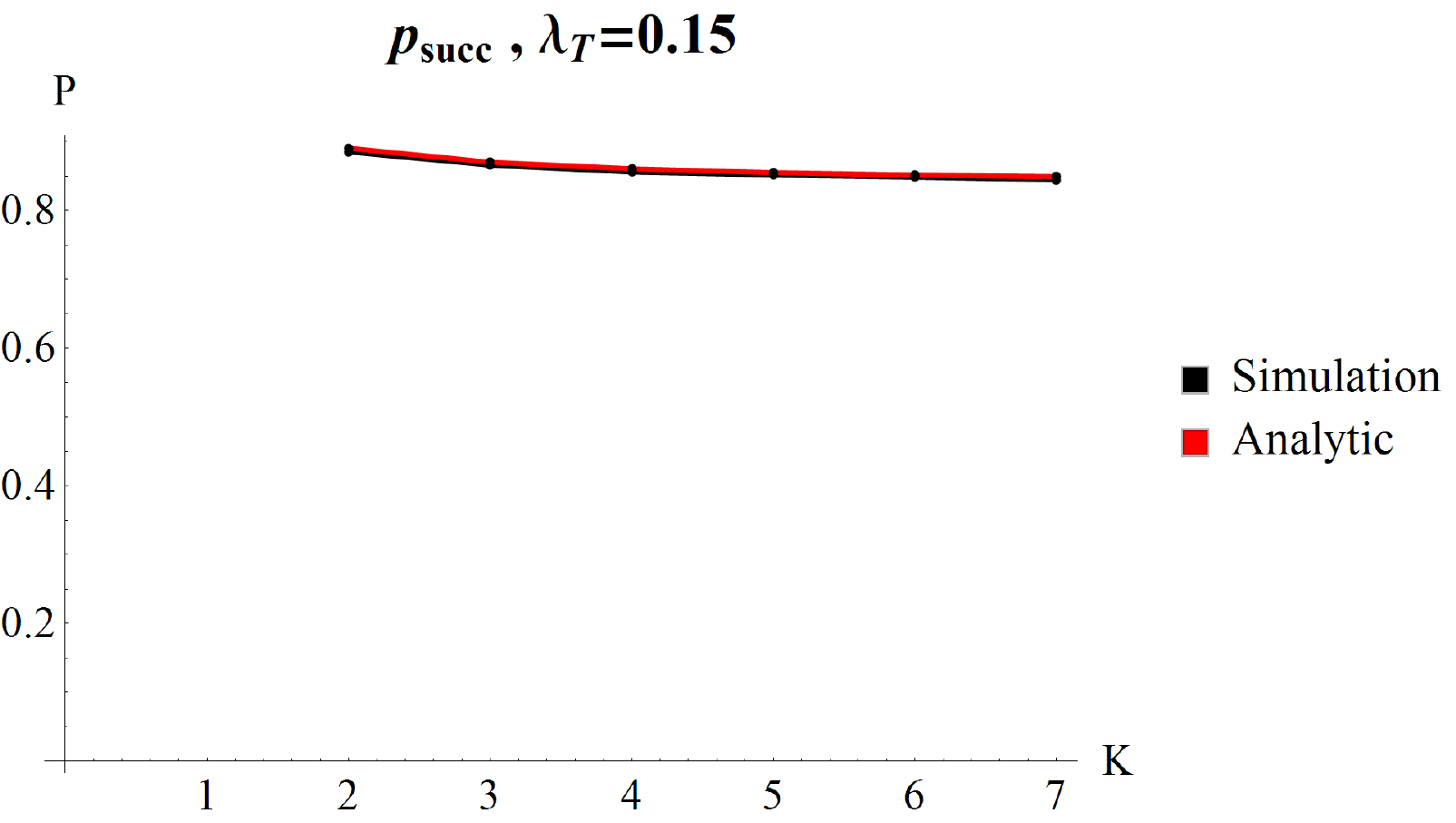}
                \caption{}
                \label{fig-SuccessProbability_2-7_015}
        \end{subfigure}
        \caption[System performance time independent model]{Anaclitic results and simulation for the system performance as a function of the number of users. Total arrival rate is $\lambda_T=0.15$. (a) The mean queue size. (b) The time in line. (c) The service time. (d) The probability for success. The red line describes the anaclitic expressions \eqref{equ-mean queue size model 1}, \eqref{equ-time in line approximation model 1}, \eqref{equ-service time approximation model 1}, \eqref{equ-success probability model 1} respectively. }
        \label{fig-QueuingPerformance_2-7_015}
    \end{figure}

    Figure \ref{fig-QueuingPerformance_K=7} and \ref{fig-QueuingPerformance_K=2} depicts the performance as a function of the total arrival rate, which help us deduce a few interesting observations. First, for small and moderate values of the arrival rate one can see very good agreement between the analytic results and the simulation while near the maximum value, the instability region, the values start to diverge, especially in Figures \ref{fig-MeanQueueSize_K=7} and \ref{fig-TimeInLine_K=7}. These results are consistent with the results of \cite{ephremides1987delay}, which used small values of $\lambda_i's$ as well. The second observation is that as the number of users grows, one can see even stronger compliance with the analytical results. This suggests that this approximate model may be suitable for a large population. Unfortunately, as mentioned, the calculation of the system steady state is intractable due to the exponential growth of the number of states with the number of users. In the following section, we give a simplified model, which utilizes the finding thus far and is able to approximate the performance for a large population as well.

    \begin{figure}[H]
        \centering
        \begin{subfigure}[b]{0.5\textwidth}
                \centering
                \includegraphics[width=\textwidth]{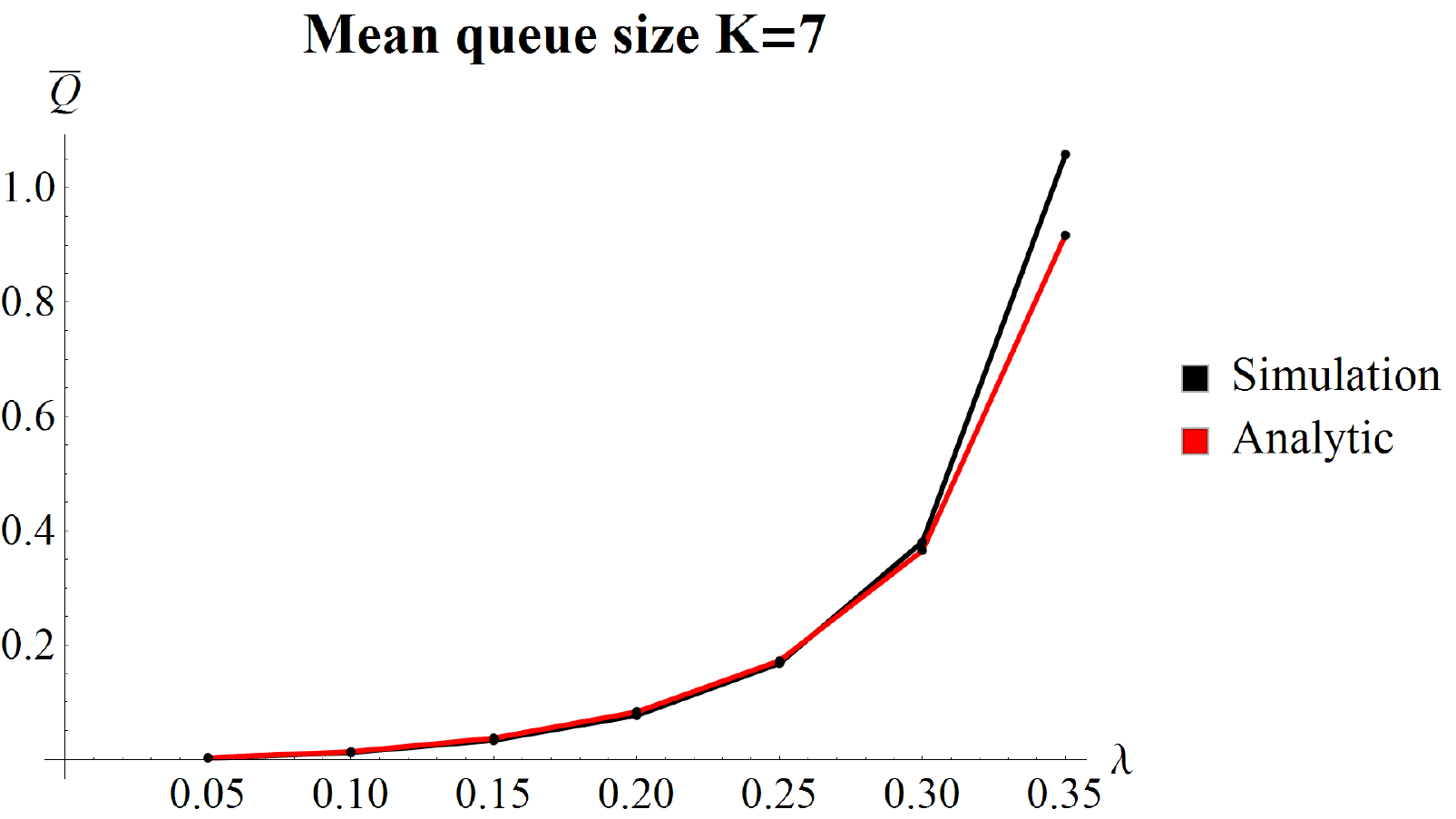}
                \caption{}
                \label{fig-MeanQueueSize_K=7}
        \end{subfigure}%
        ~
        \begin{subfigure}[b]{0.5\textwidth}
                \includegraphics[width=\textwidth]{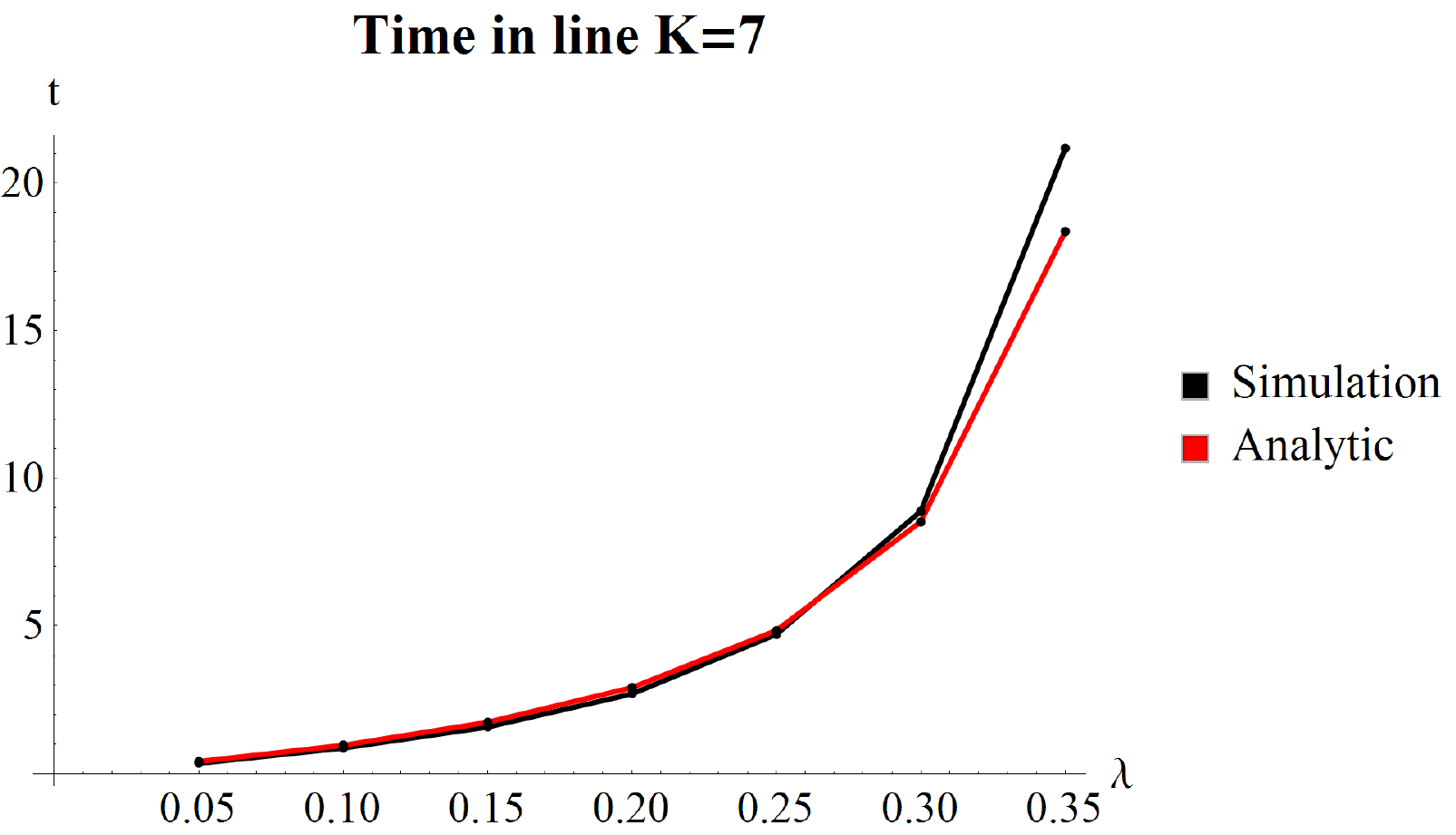}
                \caption{}
                \label{fig-TimeInLine_K=7}
        \end{subfigure}
        \hfil

        \begin{subfigure}[b]{0.5\textwidth}
                \includegraphics[width=\textwidth]{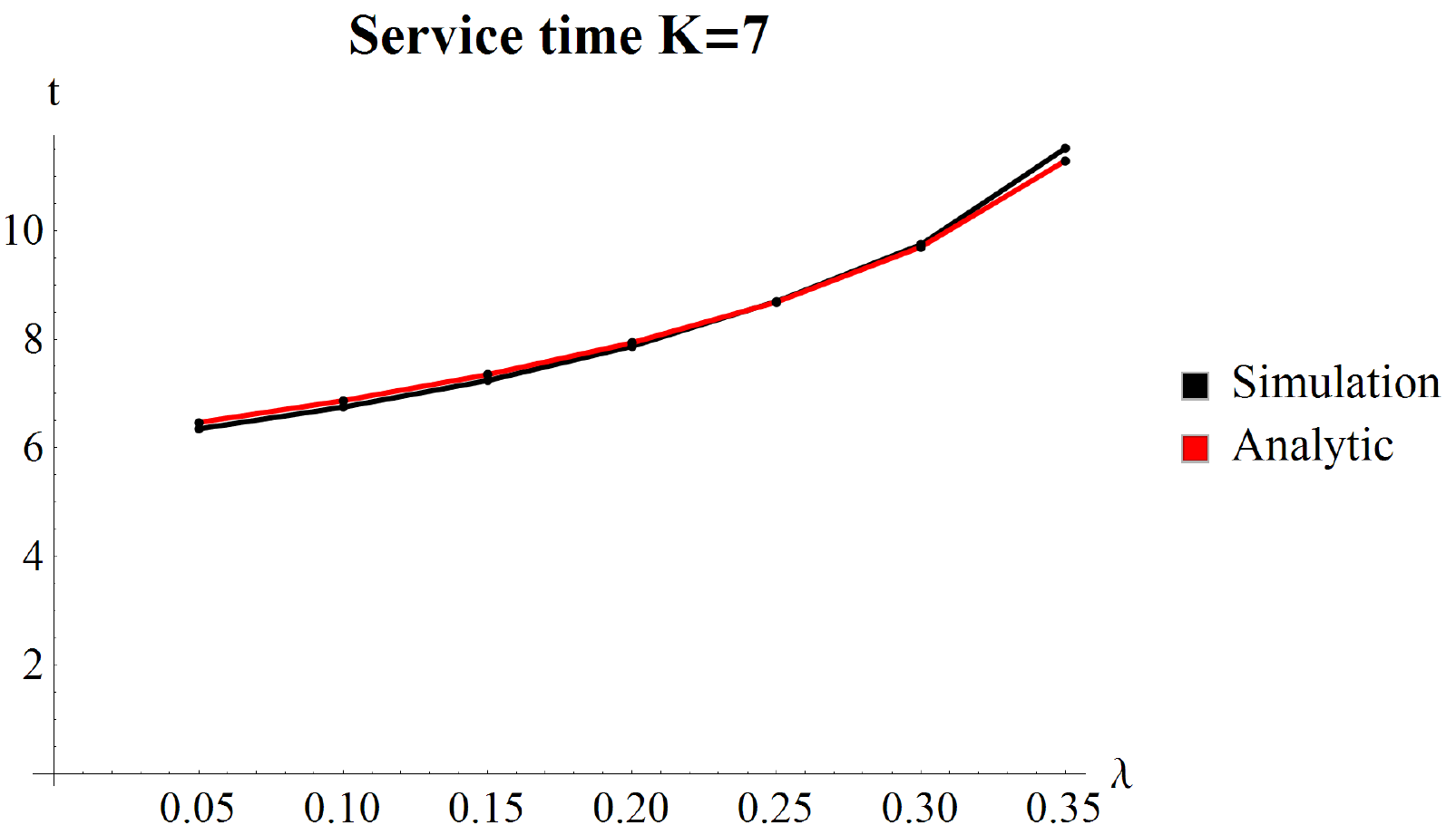}
                \caption{}
                \label{fig-ServiceTime_K=7}
        \end{subfigure}%
        ~
        \begin{subfigure}[b]{0.5\textwidth}
                \includegraphics[width=\textwidth]{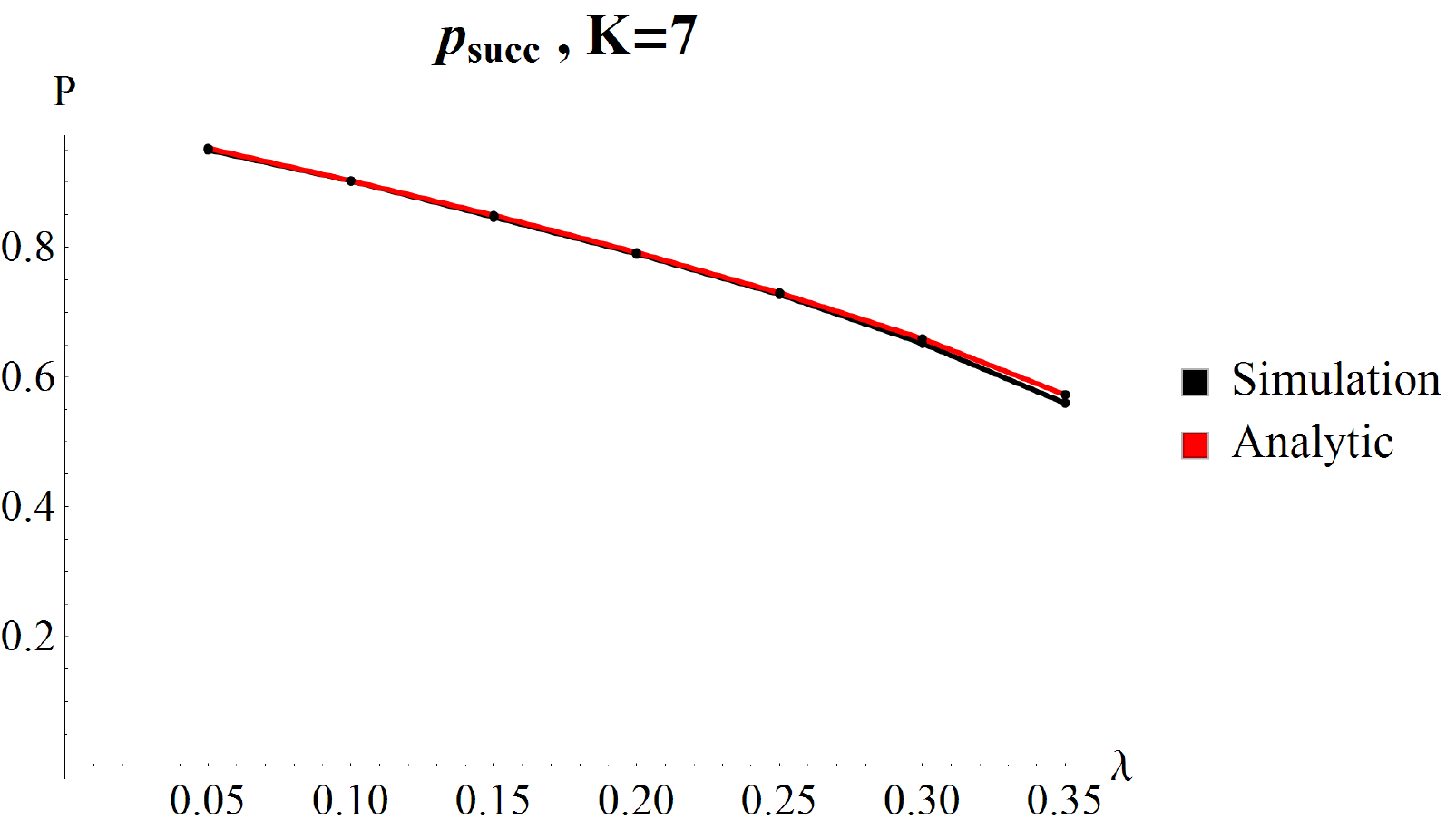}
                \caption{}
                \label{fig-SuccessProbability_K=7}
        \end{subfigure}
        \caption[System performance time independent model]{Anaclitic results and simulation for the system performance as a function of the total arrival rate. The number of users $K=7$. (a) The mean queue size. (b) The time in line. (c) The service time. (d) The probability for success. The red line describes the anaclitic expressions \eqref{equ-mean queue size model 1}, \eqref{equ-time in line approximation model 1}, \eqref{equ-service time approximation model 1}, \eqref{equ-success probability model 1} respectively.}
        \label{fig-QueuingPerformance_K=7}
    \end{figure}

    \begin{figure}[h]
        \centering
        \begin{subfigure}[b]{0.5\textwidth}
                \centering
                \includegraphics[width=\textwidth]{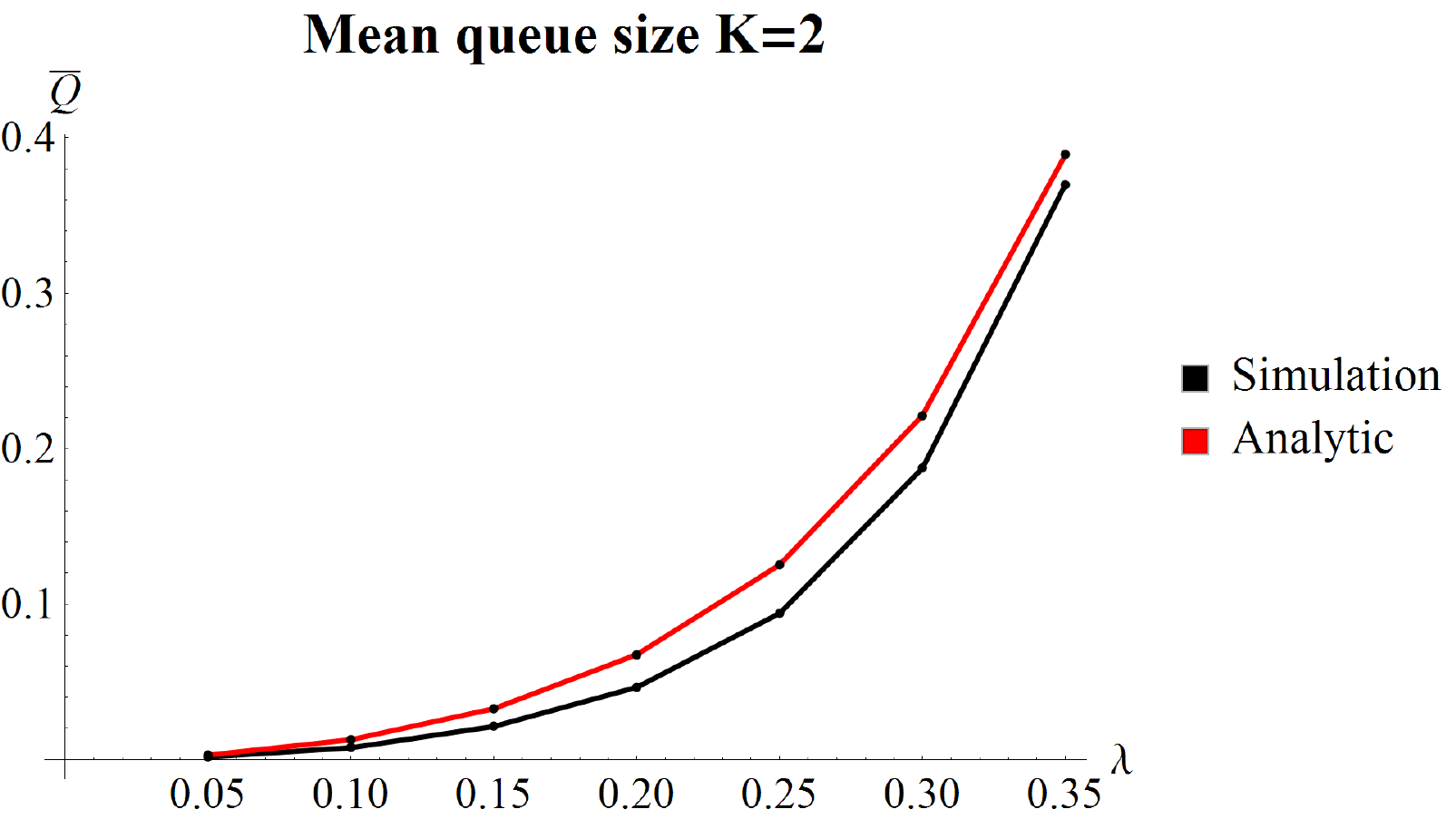}
                \caption{}
                \label{fig-MeanQueueSize_K=2}
        \end{subfigure}%
        ~
        \begin{subfigure}[b]{0.5\textwidth}
                \includegraphics[width=\textwidth]{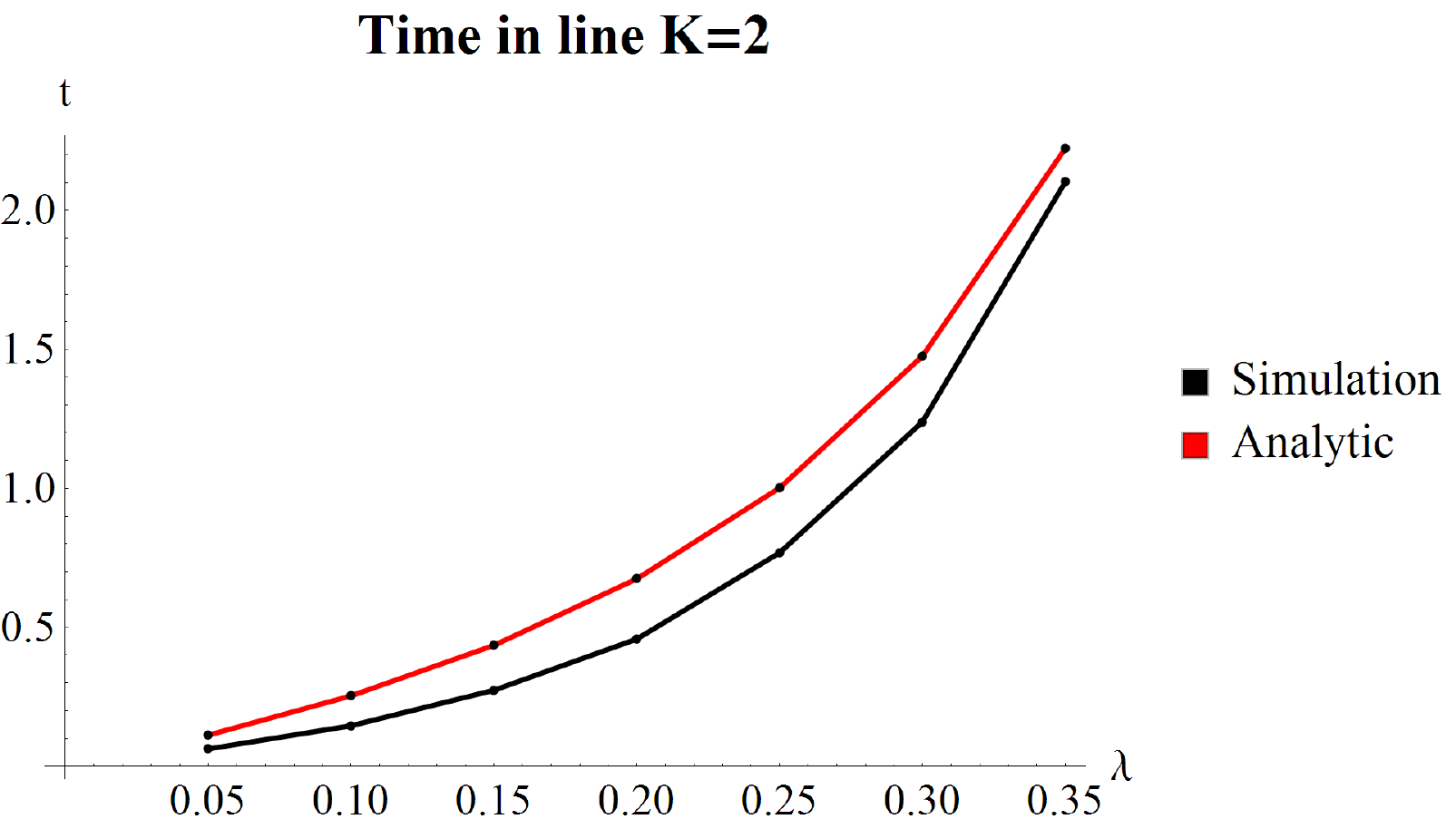}
                \caption{}
                \label{fig-TimeInLine_K=2}
        \end{subfigure}
        \hfil

        \begin{subfigure}[b]{0.5\textwidth}
                \includegraphics[width=\textwidth]{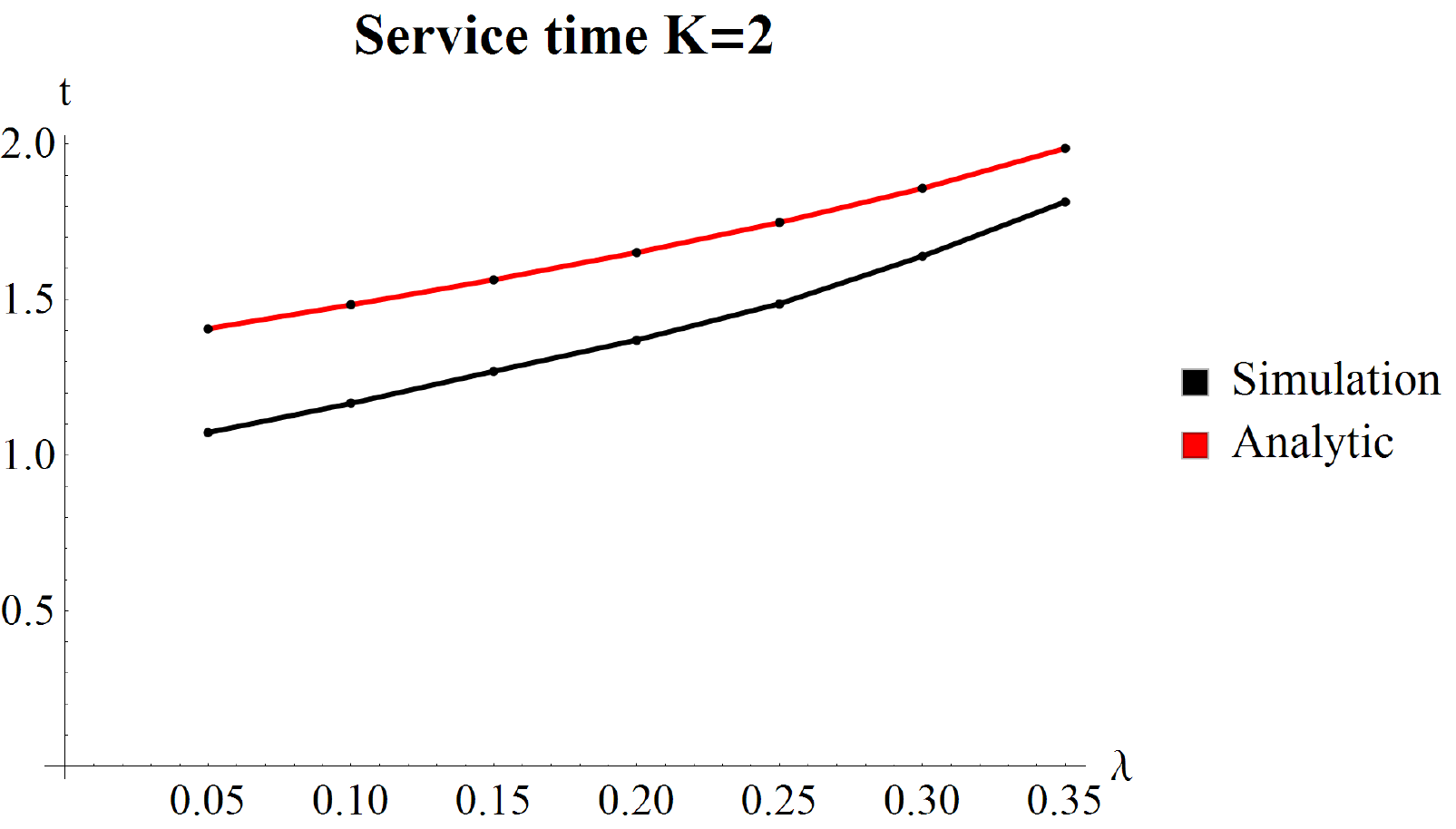}
                \caption{}
                \label{fig-ServiceTime_K=2}
        \end{subfigure}%
        ~
        \begin{subfigure}[b]{0.5\textwidth}
                \includegraphics[width=\textwidth]{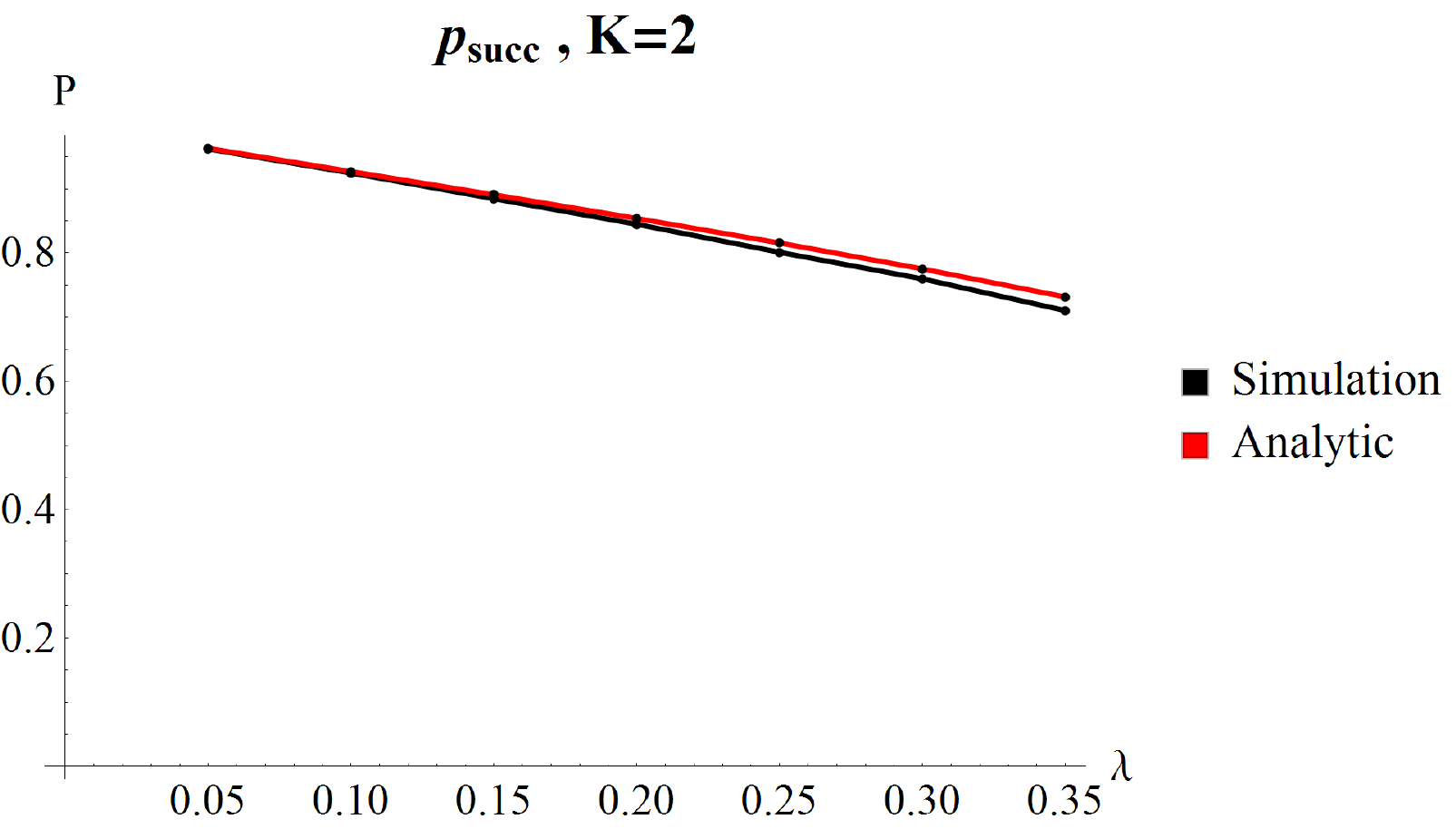}
                \caption{}
                \label{fig-SuccessProbability_K=2}
        \end{subfigure}
        \caption[System performance time independent model]{Anaclitic results and simulation for the system performance as a function of the total arrival rate. The number of users $K=2$. (a) The mean queue size. (b) The time in line. (c) The service time. (d) The probability for success. The red line describes the anaclitic expressions \eqref{equ-mean queue size model 1}, \eqref{equ-time in line approximation model 1}, \eqref{equ-service time approximation model 1}, \eqref{equ-success probability model 1} respectively.}
        \label{fig-QueuingPerformance_K=2}
    \end{figure}

    \section[Queueing Approximate model \Rmnum{2} ]{Approximation by Constant Collision Probability}\label{Approximate model 2}
    From the previous section it is clear that for large population, the performance metrics of the system are hard to calculate. In this subsection we will describe the service time using a different approach, which assumes that the\emph{ probability of collision is constant} and equals the average collision probability. The legitimacy comes from the results given in the previous section. Another work which used a constant collision probability as the key approximation is \cite{bianchi2000performance}. There, a similar model was considered, with $n$ users competing on a shared medium, however the users were \emph{considered backlogged} and only throughput performance was investigated. Nevertheless the results of the approximation in \cite{bianchi2000performance} show good agreement with the simulations, which strengthens its applicability. In \cite{sidi1983two}, independent $M/M/1$ queues were assumed, with a given distribution on the number of backlogged users. This results in an approximation with a varying collision probability. Herein, we use a different relaxation, assuming a \emph{constant collision probability}. This allows us to easily decouple the queues. We start by investigating the exceedance process, by giving the terms for convergence to the Poisson process in the independent case of the channel, which connects directly to the service time. Then we turn to the service time analysis

    \subsection{Threshold exceedance process }

    The point process approximation as described in the preliminaries gives us convergence to a non-homogeneous Poisson process $N$ of the points exceeding $u_n$. Here we shall use the notation given in \cite{EVT:Springer1983} for the Poisson properties of exceedance.\\
    Let $C_1,C_2,...,C_n$, a sequence of \textit{i.i.d.} random variables with distribution $F$, be the channel capacity draws of a specific user. From \cite[Theorem 2.1.1]{EVT:Springer1983}, if $u_n$ satisfies $n(1-F(u_n))\rightarrow \uptau$, then for $k=0,1,2,...,$
    \begin{equation*}
      P(N_n\leq k)\rightarrow e^{-\uptau}\sum_{s=0}^{k}\frac{\uptau^{s}}{s!}
    \end{equation*}
    where $N_n$ is the number of exceedances of a level $u_n$ by $\{C_n\}$. It is important to note that convergence in distribution happens if we change the "time scale" by a factor of $n$ as defined in the definition of $N_n$ in \eqref{equ-sequence of point processes N_n}. Meaning that we have a point process in the interval $(0,1]$ and the exceedance of such points has a limiting Poisson distribution. This restriction may be problematic in terms of the data arrival rate which will also need to be changed in a factor of $n$ for the analysis. However, under certain conditions, there may be convergence of the exceeding points for the entire positive line in place of the unit interval. In \cite[Theorem $5.2.1 (\rmnum{2})$]{EVT:Springer1983}, it is stated that if for each $\uptau >0$, exists a sequence $\{u_n(\uptau)\}$ satisfying $n(1-F(u_n(\uptau)))=nP(C_1>u_n(\uptau))\rightarrow \uptau$ as $n \rightarrow \infty$, and that $D(u_n(\uptau)), D'(u_n(\uptau))$ hold for all $\uptau > 0$, then for any fixed $\uptau$, $N_n$ converges in distribution to a Poisson process $N$ on $(0,\infty)$ with parameter $\uptau$. Clearly the dependence conditions hold for this case, since the sequence $\{C_n\}$ is \textit{i.i.d.} so we only need to show that the first condition holds for each $\uptau$.
    \begin{figure}[h]
        \begin{minipage}{0.6\textwidth}
            \centering
            \includegraphics[width=0.75\textwidth]{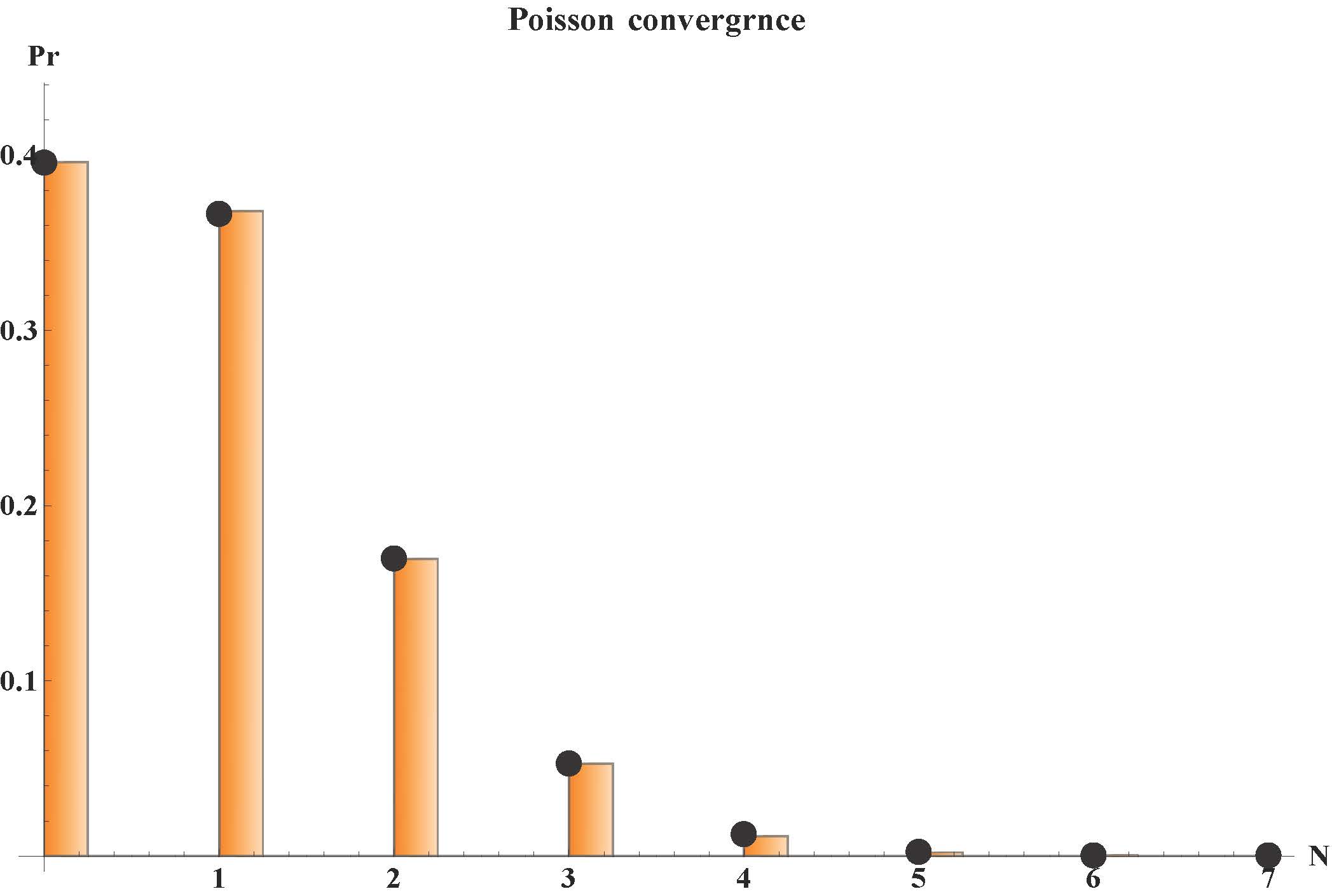}
            \caption[Poisson convergence of exceeding points]{Simulation for 10000 user's capacities which follows an \textit{i.i.d.} Gaussian distribution, showing the behaviour of exceedance which converge to a Poisson distribution (the dots) with parameter $\uptau$.}
            \label{fig-PoissonConvergence}
        \end{minipage} ~
        \begin{minipage}{0.2\textwidth}
        \captionsetup{justification=centering}
        \begin{tabular}{l|c c}
          \hline
           & Simulation & Poisson \\
           &            & distribution \\
          \hline
          Pr(N=0) & 0.3959 & 0.3961 \\
          Pr(N=1) & 0.3682 & 0.3668 \\
          Pr(N=2) & 0.1696 & 0.1698 \\
          Pr(N=3) & 0.0527 & 0.0524 \\
          Pr(N=4) & 0.0112 & 0.0121 \\
          Pr(N=5) & 0.002  & 0.0022 \\
          Pr(N=6) & 0.0004 & 0.0003 \\
          Pr(N=7) & 0.0    & 0.00004 \\
          \hline
        \end{tabular}

        \caption[Poisson convergence of exceeding points values table]{Values table of figure \ref{fig-PoissonConvergence}.}
        \end{minipage}
    \end{figure}

    \begin{lemma}\label{lem-condition for poisson convergence on all positive line}
        Assume F is the Gaussian distribution and let $a_n$ and $b_n$ be given according to \eqref{equ-parameter a_n} and \eqref{equ-parameter b_n} respectively.  Fix any  $\uptau>0$ and set $u_n(\uptau)= \frac{\log{1/\uptau}}{a_n}+b_n$ Then,
        \begin{equation*}
          \lim_{n \to \infty} n(1-F(u_n(\uptau))) = \uptau.
        \end{equation*}

    \end{lemma}
    \begin{proof}
        In a similar way for the derivation of the normalizing constant in appendix \ref{Appendix A}, and based on the derivation in \cite[Theorem 1.5.3]{EVT:Springer1983}.  Let us find $u_n(\uptau)$ which satisfies the equivalence condition for the convergence of the expression $n(1-F(u_n(\uptau)))$
        \begin{equation*}
          \begin{aligned}
                    n(1-F(u_n(\uptau))) &\rightarrow \uptau \qquad \text{as  } n\rightarrow \infty \\
                    \frac{nf(u_n(\uptau))}{u_n(\uptau)}&\rightarrow \uptau \qquad  \text{as  } n\rightarrow \infty
          \end{aligned}
        \end{equation*}
        where the second line is true due to the Gaussian relation $1-\Phi(u)\sim \frac{\phi(u)}{u}$ as $u\rightarrow \infty$, which in our case $u_n(\uptau)$ grows with $n$. So
        \begin{equation*}
          \begin{aligned}
                    \frac{1}{\sqrt{2\pi}}e^{-\frac{u_n^2(\uptau)}{2}} &\rightarrow \frac{\uptau \ u_n(\uptau)}{n} \qquad  \text{as  } n\rightarrow \infty\\
                    -\log\sqrt{2\pi}-\frac{u_n^2(\uptau)}{2}&\rightarrow \log \uptau + \log(u_n(\uptau)) - \log n \qquad  \text{as  } n\rightarrow \infty
          \end{aligned}
        \end{equation*}
        we know that $\log(u_n(\uptau))=\frac{1}{2}(\log 2 +\log{\log{n}})+o(1) $, hence
        \begin{equation*}
          \begin{aligned}
                    &\frac{u_n^2(\uptau)}{2}= \log{ \frac{1}{\uptau}}-\frac{1}{2}\log{4\pi}-\frac{1}{2}\log{\log{n}} + \log n +o(1)\\
                    &u_n^2(\uptau)=2\log n\left( 1+ \frac{\log{ \frac{1}{\uptau}}-\frac{1}{2}\log{4\pi}-\frac{1}{2}\log{\log{n}}}{\log n} +o\left(\frac{1}{\log n} \right) \right)\\
                    &u_n(\uptau)=\sqrt{2\log n}\left( 1+ \frac{\log{ \frac{1}{\uptau}}-\frac{1}{2}\log{4\pi}-\frac{1}{2}\log{\log{n}}}{2\log n} +o\left(\frac{1}{\log n} \right) \right)\\
                    &u_n(\uptau)=\frac{\log{1/\uptau}}{a_n}+b_n
          \end{aligned}
        \end{equation*}
        where the penultimate line is due to Taylor expansion.
    \end{proof}
    After proving that all conditions hold and we have a convergence to a Poisson process on the real line of the exceeding points, we can conclude that a user tries to transmit at a rate of $\uptau$ assuming he has packages to send. We will note this Poisson process of threshold exceedance as $N_{exc}$. Note that the time between each user exceedance is distributed exponentially with parameter $\uptau$. The above leads us to a discussion on the duration of a slot.

    At the limit of large number of users, the capacity seen by a user who exceeds the threshold is high (in fact, it scales like $O(\sigma\sqrt{2\log K}+\mu)$, see \eqref{equ-CapacityExpression}). Hence, for any finite size of a data packet, transmission time tends to zero as the number of users increases. Furthermore, in this case, as convergence to Poisson process $N_{exc}$ exists, events duration also goes to zero. This motivates us to suggest the asymptotic model, where transmission time is negligible but still exists, hence under this assumptions we will refer the slots as mini slots to emphasize this fact.

    \begin{remark}[Zero collisions] \label{rem-zero collisions}
        The users are independent, and each of them is also characterized with the same $N_{exc}$ process. As known, one of the properties of a Poisson process is that no two events can occur in the same time. Therefore, once a user exceeds $u_n$, assuming it's transmission time goes to zero, the probability that other users will exceed $u_n$ goes to zero as well. While considering this scenario, we can observe an approximation for the behavior of the user's queue with Poison arrival process with rate $\lambda$ and exponential service process with parameter $\uptau$, as the familiar $M/M/1$ queue. This setting gives as an upper bound for the service time and other QoS properties. Still this scenario is treated only as a guide line for the analysis, since we are assuming that the services are considered to be synchronized, and hence collisions may occur and result with decline in the rate of service.
    \end{remark}

    \subsection{Service time analysis}
    The service time, which is the time from the moment a package becomes first in queue, until it is successfully transmitted, depends on the rate which a user exceeds the threshold and the probability for success (the probability which collision did not occur), given that the user has a package to transmit. The time between threshold exceedances is exponentially distributed with parameter $\uptau$, due to the Poisson properties of the exceedance process (PPA). Therefore, the remaining parameter for the service time analysis is the probability for success. In fact, this probability is the main key for understanding this complex system. Due to the strong interdependence between the queues, this probability is constantly fluctuating and by quantifying it using average collision probability, one can approximate the service time.

    We start with the assumption that the users are backlogged and give the appropriate conclusions regarding the success probability. Then we turn to our model with memoryless arrival process. From now on, we will refer to the probability of collision which is the complementary probability of success, i.e. $p_{coll}=1-p_{succ}$.\\

    We define a collision scenario, if at least two users exceed the threshold $u_n$ and both has packages to transmit. The analysis is done by considering a specific user $i$ at the time which he exceeds $u_n$ and assuming his queue, $Q_i$, is not empty. Given that, we note $p_{coll}$ as the probability for collision,
    \begin{equation}\label{equ-probability for collision - General}
      p_{colll}=P( \bigcup_{j=1..K,j \neq i} C_j(n)>u_n,Q_j>0 \mid C_i(n)>u_n,Q_i>0)
    \end{equation}
    where $Q_j$ is the number of packages in the queue of user $j$. In general the probability $p_{coll}$ depends on the states of the other users queues. We will analyse $p_{coll}$ of a specific user, starting with the more simple analysis, resulting with an upper bound, assuming all users always have packages to transmit.

    \subsubsection{Backlogged system}
    In the case which all users are backlogged, the probability for collision is constat and independent among the users. The state of the users' queues does not change over time, since their queue is always full, and the only thing that impacts the probability for collision is whether a user exceeds $u_n$ or not. Then, by observing the $N_{exc}$ process, a certain point indicates a successful transmission with probability $p_{succ}$ and a collision with probability $p_{coll}=1-p_{succ}$. Therefore a random selection is made from $N_{exc}$ resulting with a Poisson process of successful transmissions with rate $\uptau \cdot(1- p_{coll})$, which we will note as $N_{tr}$. The proof is quite obvious since the probability that an arrival occurs from the original process in the interval $dt$ is $\uptau dt$, which is independent of the arrivals outside the interval, where after the random selection, the probability for an arrival in the interval $dt$ is $(1-p_{coll}) \cdot \uptau dt$ (independent of the arrivals outside the interval). Thus we conclude that the number of successful transmissions is a Poisson random variable with parameter $(1-p_{coll}) \cdot \uptau$. Nevertheless a well appointed proof can be found in appendix \ref{Appendix C}.\\

    Consider the following problem under the framework of the general problem. We have one user with arrival process with rate $\lambda$ and $N_{tr}^{\text{backlogged}}$ process with rate $\uptau \cdot (1-p_{coll}^{\text{backlogged}})$, where all other users are backlogged. The definition of the problem in this manner will help us obtain an upper bound on $p_{coll}$ by understanding the fact that the condition of a user in the original problem is better, in a way that users may have empty queues and although exceedance occur it will not result in a collision. This setting isolates for us the other users queues' states, from the analysis of the probability for collision.\\

    In a given time the probability that a user exceeds the level $u_n$ is $1-F(u_n)$, so in order to have a collision, at least one user must exceed $u_n$ among all users, hence we get
    \begin{equation}\label{equ-probability of successful transmission - Backlogged system}
      \begin{aligned}
            &p_{coll}^{\text{backlogged}}=P( \bigcup_{j=1..K,j \neq i} C_j(n)>u_n \mid C_i(n)>u_n,Q_i>0)\\
            &=1-P(C_1(n)<u_n,...,C_{K-1}(n)<u_n)\\
            &=1-(F(u_n))^{K-1}=1-\left( 1-\frac{1}{K} \right)^{K-1} \\
            &\overset{K \rightarrow \infty}{\rightarrow}1- e^{-1}
        \end{aligned}
    \end{equation}
    Again the value $u_n$ is set such that only one user on average exceeds it in each slot, hence $1-F(u_n)=1/K$.\\
    This analysis is based on the fact that all the users have packages to send in any given time, so we can conclude that $p_{coll} \leq p_{coll}^{\text{backlogged}}$. Therefore a lower bound for the service process $N_{tr}$ rate is $\uptau \cdot(1- p_{coll}^{\text{backlogged}})$.

    \subsubsection{Memoryless arrival}

    The assumption of constant and independent collision probability relaxes the analysis dramatically. We can now assume that each user's queue behaves like an $M/M/1$ queue, with a Poisson arrival process with rate $\lambda$ and exponential service time with parameter $(1-p_{coll}) \cdot \uptau$, due to random selection with probability $1-p_{coll}$ from the exceedance process. By doing so, we can use the known results for the probability of an empty $M/M/1$ queue, which is
    \begin{equation}\label{equ-the probability for empty queue M/M/1}
      P(Q=0)=1-\rho=1-\frac{\lambda}{(1-p_{coll}) \cdot \uptau},
    \end{equation}
    where $Q$ is the number of packages in the queue. We thus have the following lemma.

    \begin{lemma}\label{lem-Transmission success probability satisfies the equation - Memoryless arrival}
       \textit{The users queue can be modeled as an $M/M/1$ queue with an arrival rate $\lambda$ and service rate $(1-p_{coll}) \cdot \uptau$, where $\uptau$ is the rate of the exceedance process, $N_{exc}$, and the probability for collision $p_{coll}$ satisfies the equation
       \begin{equation}\label{equ-Transmission success probability satisfies the equation - Memoryless arrival}
         p_{coll}=1-e^{-\frac{\lambda}{(1-p_{coll}) \cdot \uptau}}.
       \end{equation}}
    \end{lemma}
    \begin{proof}
     As random selection (with probability $(1-p_{coll})$) is preformed on the exceedance Poisson process, this results in a Poisson process for successful transmission, which leads to an exponential service time with parameter $(1-p_{coll}) \cdot \uptau$. Let us examine the probability $p_{coll}$.\\
     The probability that a specific user exceeds $u_n$, and its queue is not empty is
     \begin{equation*}
       P(C_i>u_n,Q_i>0)=P(C_i>u_n)P(Q_i>0)=\frac{1}{K}\cdot\frac{\lambda}{\uptau\cdot (1-p_{coll})}.
     \end{equation*}
     Hence, given that one user exceeds $u_n$ and its queue is not empty, $p_{coll}$ is the probability that among all the other users which exceeded $u_n$, at least one has a package to send. Thus
        \begin{equation*}
          \begin{aligned}
          p_{coll}&=\sum_{i=1}^{K-1}  \binom {K-1} {i} \left(\frac{1}{K} \frac{\lambda}{\uptau\cdot (1-p_{coll})} \right)^{i}\left(1-\frac{1}{K} \frac{\lambda}{\uptau\cdot (1-p_{coll})} \right)^{K-1-i} \\
          &= 1-\left( 1-\frac{1}{K} \frac{\lambda}{\uptau\cdot (1-p_{coll})} \right)^{K-1} \\
          &\overset{K \rightarrow \infty}{\rightarrow} 1-e^{-\frac{\lambda}{\uptau\cdot (1-p_{coll})}}.
           \end{aligned}
        \end{equation*}
    \end{proof}

    This simple result helps understanding the behaviour of the system's service time. Of course, equation \eqref{equ-Transmission success probability satisfies the equation - Memoryless arrival} is an implicit equation and a numerical method is needed in order to find the value of $p_{coll}$. In Figure \ref{fig-SuccessProbability_2-10_0366_independent}, one can see that the numerical calculation for the success probability (the blue line), as given in \eqref{equ-Transmission success probability satisfies the equation - Memoryless arrival}, coincides with the results of the simulation \emph{and the approximate model from the previous section}. It is important to emphasise that for this comparison, the number of users is relatively small and therefore we used the equation in its explicit form, meaning without taking $K$ to infinity. The good agreement of the results implies that for the case which $K \rightarrow \infty$,  Lemma \ref{lem-Transmission success probability satisfies the equation - Memoryless arrival} is able to describe the behavior of the users' queues, hence the performance metrics of the system as a system of $M/M/1$ queues. Simulations for large number of users show excellent agreement as well, as can be seen in Figure \ref{fig-Service_time_model2}, where a comparison of the average service time was made between an $M/M/1$ queue with service rate $p_{succ}\uptau$, which was calculated according to equation \eqref{equ-Transmission success probability satisfies the equation - Memoryless arrival}, with simulation of the system, which of course includes dependent queues and variable $p_{succ}$

    \begin{figure}[h]
        \centering
        \includegraphics[width=4in]{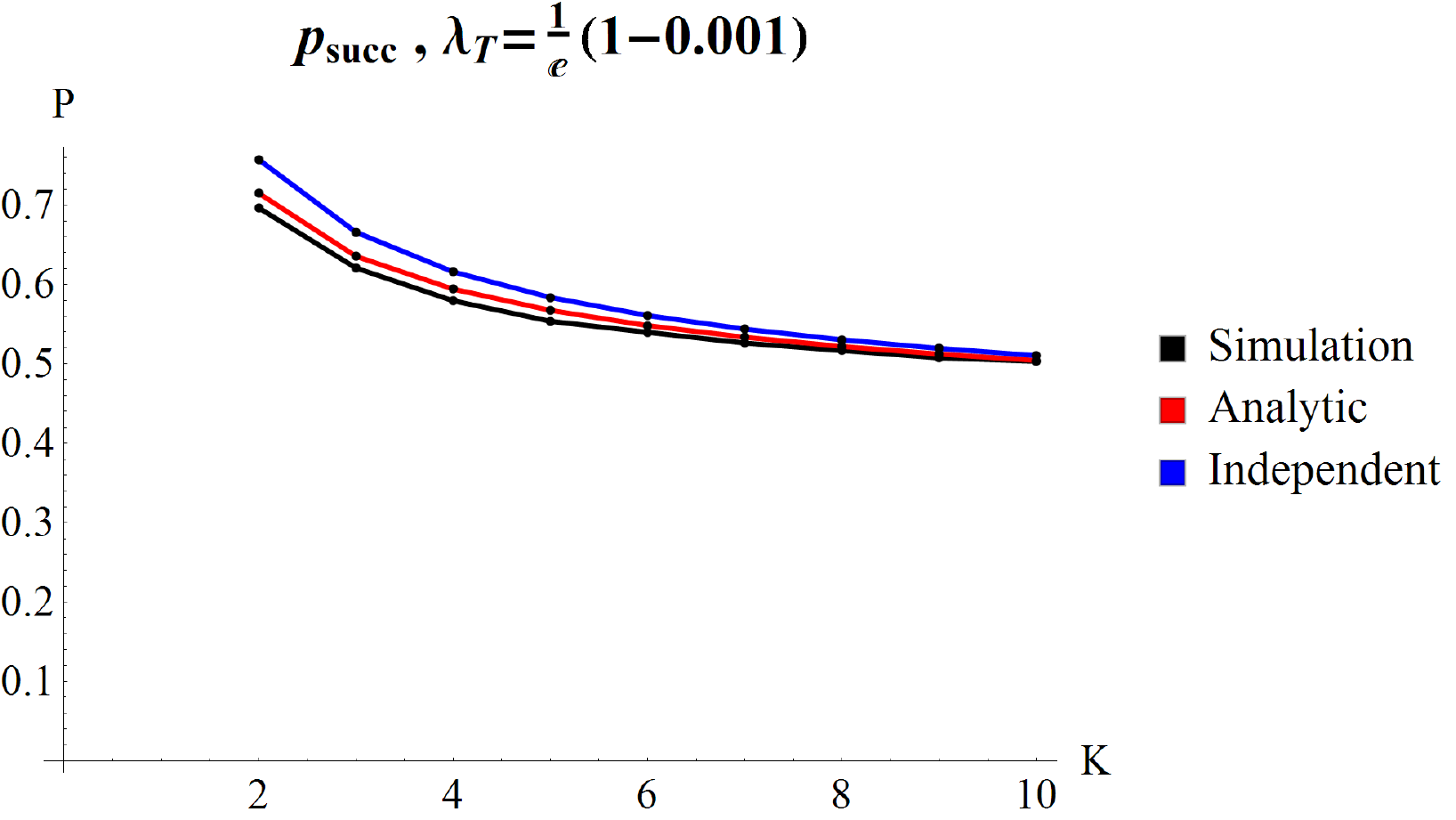}
        \caption[Comparison of success probability]{Comparison of the success probability between the analytic derivations of the approximate model \Rmnum{1}, simulation results, and the approximation of $K$ \emph{independent} M/M/1 queues given in Lemma \ref{lem-Transmission success probability satisfies the equation - Memoryless arrival}.}
        \label{fig-SuccessProbability_2-10_0366_independent}
    \end{figure}
    \begin{figure}[h]
        \centering
        \includegraphics[width=4in]{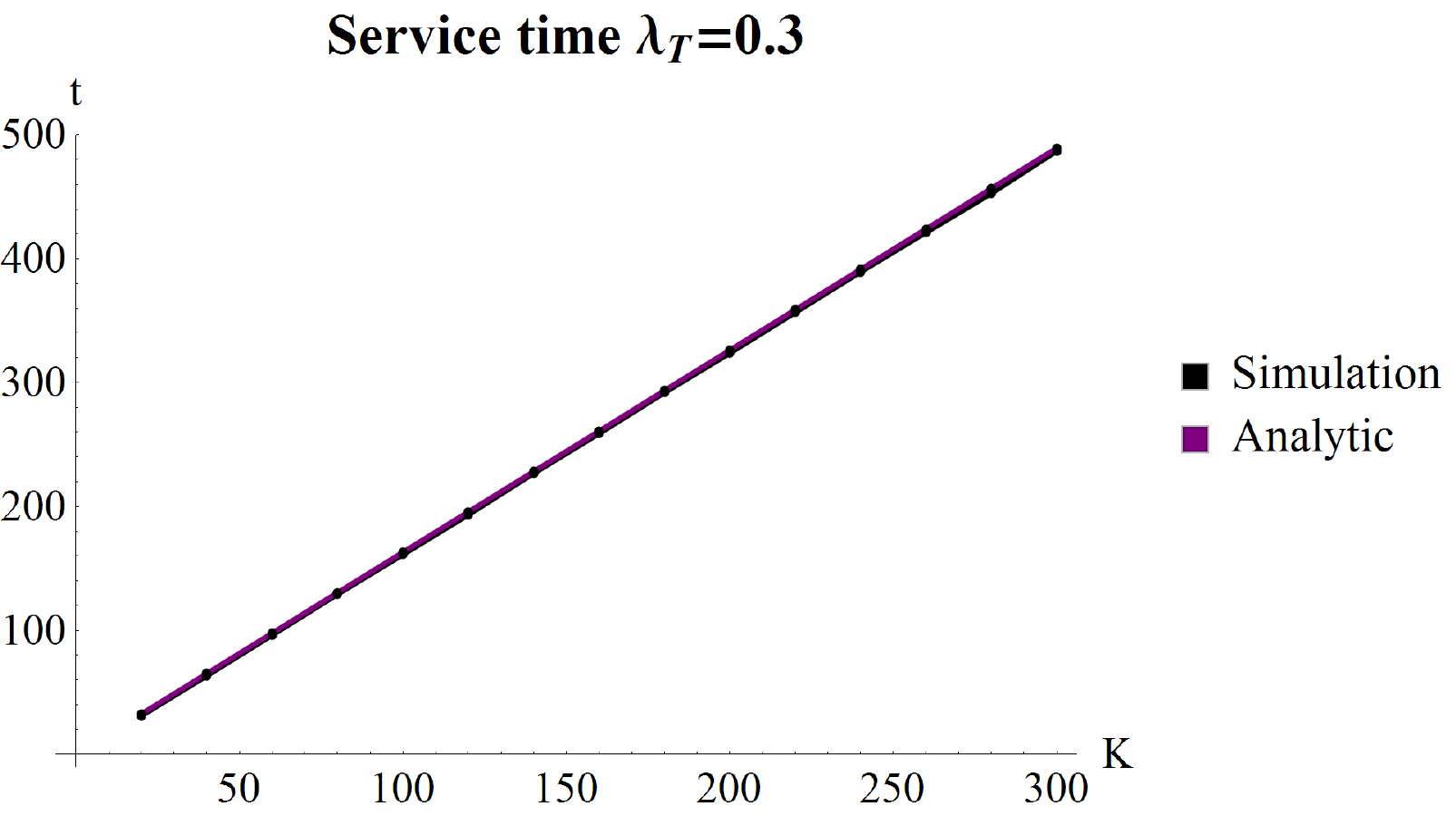}
        \caption[Comparison of service time]{The service time of an M/M/1 queue with service rate $p_{succ}\uptau$, compared to simulation results of system with $K$ \emph{interdependent} queues.}
            \label{fig-Service_time_model2}
        \label{fig-Service_time_model2}
    \end{figure}

    \section[Queueing Approximate model \Rmnum{3}]{Approximation by Constant Collision Probability - Time Dependent Channel}\label{Approximate model 3}
    In chapter \ref{Capacity under time dependent channel} we analyzed the channel capacity, where each user experiences a time dependent process for the channel gain. There we assumed that all users are backlogged, therefore queuing aspects were not taken into account. We showed that under our distributed scheme the exceedance process of a specific user, above some threshold $u_K$, is exponentially distributed with some rate $\uptau$. This rate of exceedance was derived from the stationary distribution of the gain process, and there is no real separation between users that are in good state or in bad state in means of different rates for the corresponding states. The difference is embodied in the stationary distribution. thus $\uptau$ is the average rate of exceedance for such a system, which a user spends some of the time in the good state and some in the bad state, according to the transition probabilities.\\

    Then, in the beginning of chapter \ref{Delay and QoS under the distributed algorithm} we turned to the analysis of the queues in such distributed multiple-access system, when we assumed a memoryless arrival process. In section \ref{Approximate model 1} we gave an elaboration for a known approximate model which gives us performance parameters for such system, and due to its bad scalability for a large population we gave in section \ref{Approximate model 2} another approach to approximate the users' queues in a large population system. We show that the collision probability in each slot may be treated as constant, where this is called decoupling approximation, which works nicely when the number of users is large. Both sections' analysis was preformed while assuming that the users' capacities are \textit{i.i.d.} random variables and therefore the threshold value was set accordingly.\\
    It's an interesting question, considering the stationary good-bad capacity sequence as described earlier in \eqref{equ-UserCapacityProcess}, is it enough to set only the threshold value for the stationary case, as given in \ref{equ-Estimated threshold}, with the relevant normalizing constants, as given in Theorem \ref{thm-Capacity distribution in a time dependent channel convergence to a Gumbel}, to capture the service time behaviour as described in Lemma \ref{lem-Transmission success probability satisfies the equation - Memoryless arrival}.
    We will not try to answer this question now, but rather propose another model which combines the queuing system with the users' good-bad channel process.\\

    Assume now that each user experiences a good or bad channel, according to the time dependent model in Figure \ref{fig-GoodBadchannel}. Assume further that the transitions between the states are exponentially distributed with rates $\alpha$ and $\beta$. Depending on the state a user exists in, it exceeds the threshold with a different rate. Namely, if a user is in a good state, it would exceed the threshold more often than when being in a bad state. Hence, the user's queue is modeled as follows. A Poisson arrival process with rate $\lambda$, and an exceedance process which is assumed to be exponentially distributed according to the user's state. If a user is in a good state, the exceedance process is exponentially distributed with rate $\mu_g$, and if it is in a bad state, the exceedance process is exponentially distributed with rate $\mu_b$.


    Since we consider multiple user system, where each user competes for the channel distributively, collisions may occur. Since we consider this probability as constant, we have a similar conclusion, that is, random selection is preformed from the exceedance process, with probability $p_{succ}$, regardless the state which the user exists in. This results in an exponentially distributed service process with rates $\mu_g\cdot p_{succ}$ or $\mu_b\cdot p_{succ}$, depending on the good-bad state respectively. The individual queue model for this case is presented in Figure \ref{fig-Time Dependent Queue}.

    \begin{figure}[h]
        \centering
        \includegraphics[width=2.5in]{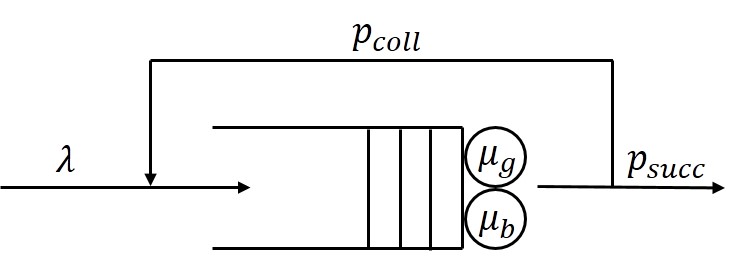}
        \caption[Time dependent queue model]{Queue model with varying service rates according to a time dependent channel.}
        \label{fig-Time Dependent Queue}
    \end{figure}

    This time-dependent queue model can be represented as a continuous Markov process on the set of states $\{\pi^i_m\}$ for $i\in\{b,g\}$, which indicates the good or bad state, and $m=0,1,2...$ the number of packages in the queue. This two dimensional Markov chain is presented in Figure \ref{fig-Time Dependent Queue Markov chain}. To ease notation, we denote $\mu'_g=\mu_g\cdot p_{succ}$ and $\mu'_b=\mu_b\cdot p_{succ}$.
    \begin{figure}[h]
        \centering
        \includegraphics[width=3.5in]{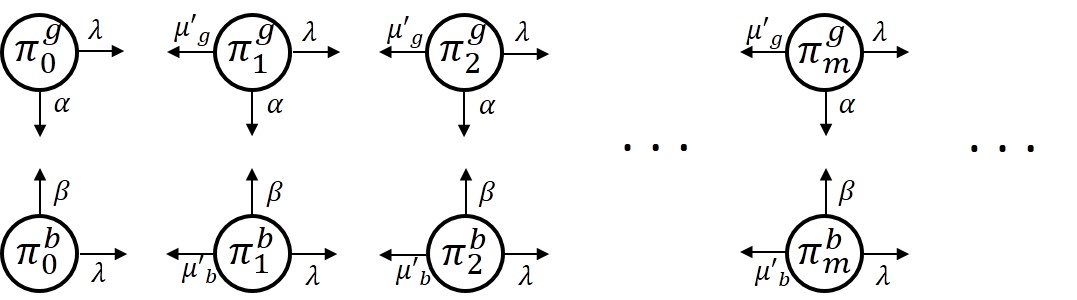}
        \caption[Time dependent queue Markov chain]{Queue model for time dependent user}
        \label{fig-Time Dependent Queue Markov chain}
    \end{figure}

    In \cite{yechiali1971queuing}, the authors investigated a problem with heterogeneous arrivals and service to an $M/M/1$ queue, which was solved by analysing a similar Markov chain to the one presented in Figure \ref{fig-Time Dependent Queue Markov chain}. We make use of the solution in \cite{yechiali1971queuing}, with the right adjustment to our model. The system is solved while using generating-function techniques, resulting with a solution of the steady state probabilities, $\{\pi^i_m\}$, as a function of the transition rate parameters and the root of a third degree polynomial $g(z)$ (only one exist under the assumptions). Here, we will show the relevant fitted steady state results of the chain.
    First, let us define
    \begin{equation*}
            \hat{\mu}=\pi_g \mu'_g+\pi_b\mu'_b
    \end{equation*}
    where
    \begin{equation*}
    \begin{aligned}
      \pi_g&=\frac{\beta}{\alpha+\beta}\\
      \pi_b&=\frac{\alpha}{\alpha+\beta}
      \end{aligned}
    \end{equation*}
    $\hat{\mu}$ is the average service rate. It is important to emphasize that in order to maintain stability of the system, we require that $\hat{\mu}>\lambda$. We define the partial generating functions of the system as
    \begin{equation*}
      G_i(z)=\sum^{\infty}_{m=0}\pi^i_m z^m \quad \quad \mid z \mid \leq 1, \ i=g,b.
    \end{equation*}
    which eventually equals to
    \begin{equation*}
      \begin{aligned}
            G_g(z)&=\left( \beta(\hat{\mu}-\lambda)z+\pi^g_0\mu'_g(1-z)(\lambda z-\mu'_b) \right)/g(z) \\
            G_b(z)&=\left( \alpha(\hat{\mu}-\lambda)z+\pi^b_0\mu'_b(1-z)(\lambda z-\mu'_g) \right)/g(z),
        \end{aligned}
    \end{equation*}
    where $g(z)$ is
    \begin{equation}\label{equ-third degree polynomial}
        \begin{aligned}
            g(z)=&\lambda^2z^3-(\alpha\lambda+\beta\lambda+\lambda^2+\lambda\mu'_b+\lambda\mu'_g)z^2  \\
            &+(\alpha\mu'_b+\beta\mu'_g+\mu'_g\mu'_b+\lambda\mu'_b+\lambda\mu'_g)z-\mu'_g\mu'_b
        \end{aligned}
    \end{equation}
    The steady state probabilities for an empty queue, depending on the system state, are
    \begin{equation}
        \begin{aligned}
            \pi^g_0&=\frac{\beta(\hat{\mu}-\lambda)z_0}{\mu'_g(1-z_0)(\mu'_b-\lambda z_0)} \\
            \pi^b_0&=\frac{\alpha(\hat{\mu}-\lambda)z_0}{\mu'_b(1-z_0)(\mu'_g-\lambda z_0)}
        \end{aligned}
    \end{equation}
    where $z_0$ is the root of the polynomial $g(z)$.
    The remaining steady state probabilities are as follows
    \begin{equation}
        \begin{aligned}
            \pi^g_m&= \pi^g_{m-1}\frac{\lambda}{\mu'_g}+\sum^{m-1}_{j=0}\pi^g_j\frac{\alpha}{\mu'_g}- \sum^{m-1}_{j=0}\pi^b_j\frac{\beta}{\mu'_g} \quad \quad m>0 \\
            \pi^b_m&= \pi^b_{m-1}\frac{\lambda}{\mu'_b}+\sum^{m-1}_{j=0}\pi^b_j\frac{\beta}{\mu'_b} - \sum^{m-1}_{j=0}\pi^g_j\frac{\alpha}{\mu'_b} \quad \quad m>0
        \end{aligned}
    \end{equation}
    From the first derivative of the partial generating functions, we can attain the expected queue size, which includes the head of line package, unlike the analysis in subsection \ref{Approximate model 1}
    \begin{equation*}
      \overline{Q}=\frac{\lambda}{\hat{\mu}-\lambda}+\frac{\mu'_g(\mu'_b-\lambda)\pi^g_0+
      \mu'_b(\mu'_g-\lambda)\pi^b_0-(\mu'_g-\lambda)(\mu'_b-\lambda)}{(\alpha+\beta)(\hat{\mu}-\lambda)}
    \end{equation*}
    Therefore, using Little's theorem we can attain the average waiting time in the queue:
    \begin{equation*}
      W=\frac{\overline{Q}}{\lambda}.
    \end{equation*}
    For the completeness of the above results, we need to calculate $p_{succ}$, which until now is a fixed and unknown parameter. Hence, in the same manner as the previous subsection, we define the probability that a specific user exceeds the threshold and its queue is not empty, to be
    \begin{equation}\label{equ-probability for attempt transmission}
    \begin{aligned}
      P_t &\triangleq Pr(\text{transmission attempt})\\
      &=Pr(\text{exceedance occur \& queue is not empty})\\
      &=(G_g(1)-\pi^g_0)(1-e^{-\mu_g})+(G_b(1)-\pi^b_0)(1-e^{-\mu_b}),
    \end{aligned}
    \end{equation}
    where the first term consists of the probability to be in a good state with a package to transmit, times the probability which an exceedance occurs while existing in the good state; The second term is similar, referring only to the bad state.
    The probability for success, given one user is about to transmit, can be obtained by a similar calculation as in Lemma \ref{lem-Transmission success probability satisfies the equation - Memoryless arrival}, and we have
    \begin{equation}\label{equ-probability for sucss-third model}
      p_{succ}=(1-P_t)^{K-1}.
    \end{equation}

     Now, in order to attain this probability, the above equation and the third degree polynomial \eqref{equ-third degree polynomial} must be solved simultaneously, since the root $z_0$ and $p_{succ}$ are coupled. As $K \rightarrow \infty$, the probability $p_{succ}$ remains constant and does not go to zero since the probability for a transmission attempt, $P_t$, scales like $O(1/K)$. This is true since in the first place, the threshold value was set such that on average, the probability for exceedance equals $1/K$. Thus, $(1-e^{-\mu_g})= O(c_g/K)$ and $(1-e^{-\mu_b})= O(c_b/K)$, where $c_g,c_b$ are some constants which satisfy $c_g>c_b$. Letting $K \rightarrow \infty$, we get a constant probability. In Figure \ref{fig-QueuingPerformance_time dependent_P_succ} one can see the probability for success $p_{succ}$, which coincides with the simulation results and therefore the good agreement of the system performance with the simulations, as presented in Figure \ref{fig-QueuingPerformance_time dependent}. The simulations were preformed while assuming equal transitions rates between the good and bad state, i.e. $\alpha=\beta=0.1$. Since we are considering a symmetric system, the arrival and service rate parameters, as mentioned in the figures, were considered as the total rates of the system and were divided equally among all the users.

    \begin{figure}[H]
        \centering
        \begin{subfigure}[b]{0.48\textwidth}
                \includegraphics[width=\textwidth]{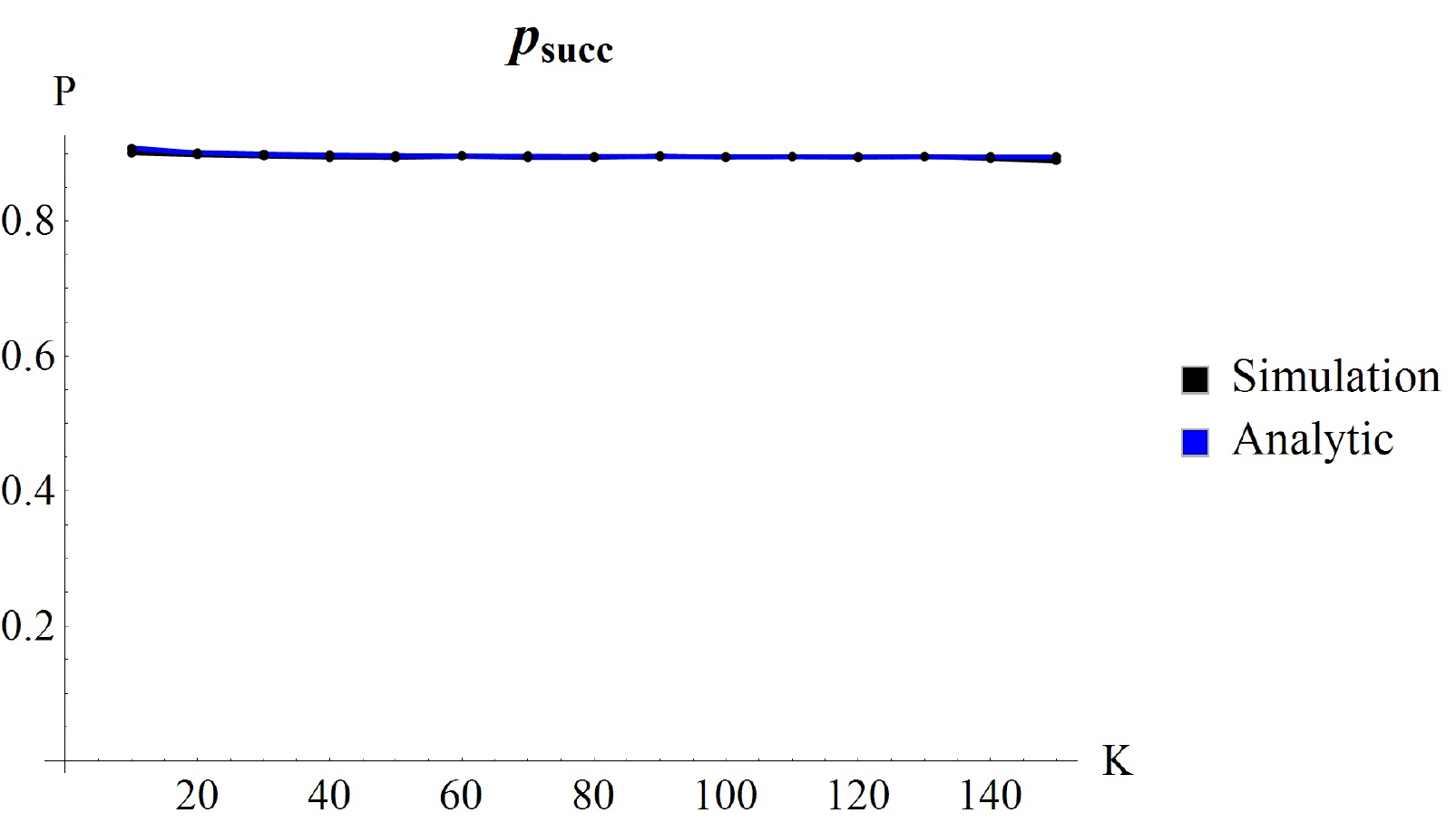}
                \caption{}
                \label{fig-P_succ_L=0.1}
        \end{subfigure}%
        ~
        \begin{subfigure}[b]{0.48\textwidth}
                \includegraphics[width=\textwidth]{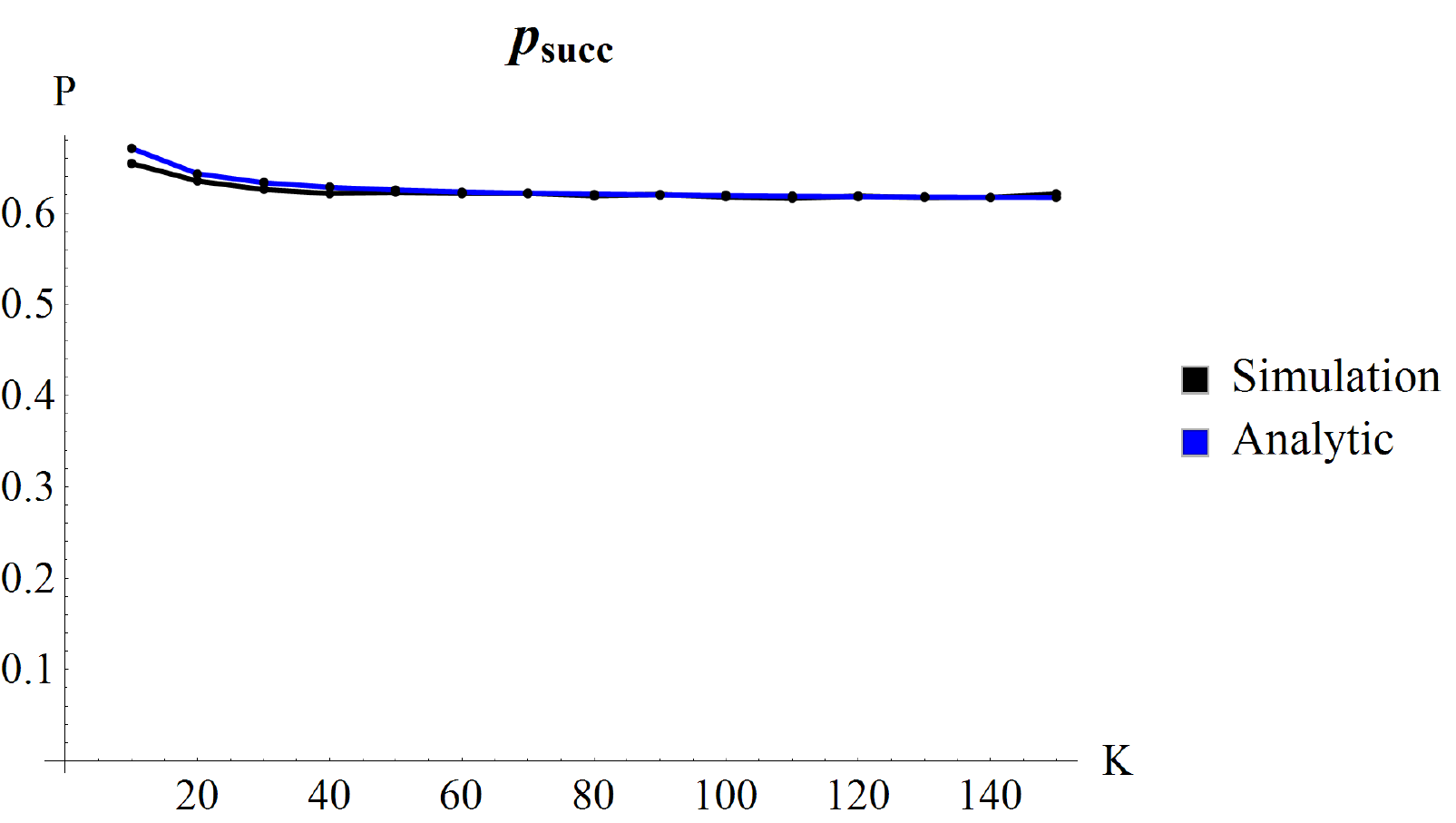}
                \caption{}
                \label{fig-P_succ_L=0.3}
        \end{subfigure}
        \caption[Success probability time dependent model]{Queuing system success probability results as a function of the number of users in a time dependent queue model. $\lambda_T=\{0.1,0.3\},\mu_g=0.7,\mu_b=0.5$}
        \label{fig-QueuingPerformance_time dependent_P_succ}
    \end{figure}

    \begin{figure}[H]
        \centering
        \begin{subfigure}[b]{0.48\textwidth}
                \centering
                \includegraphics[width=\textwidth]{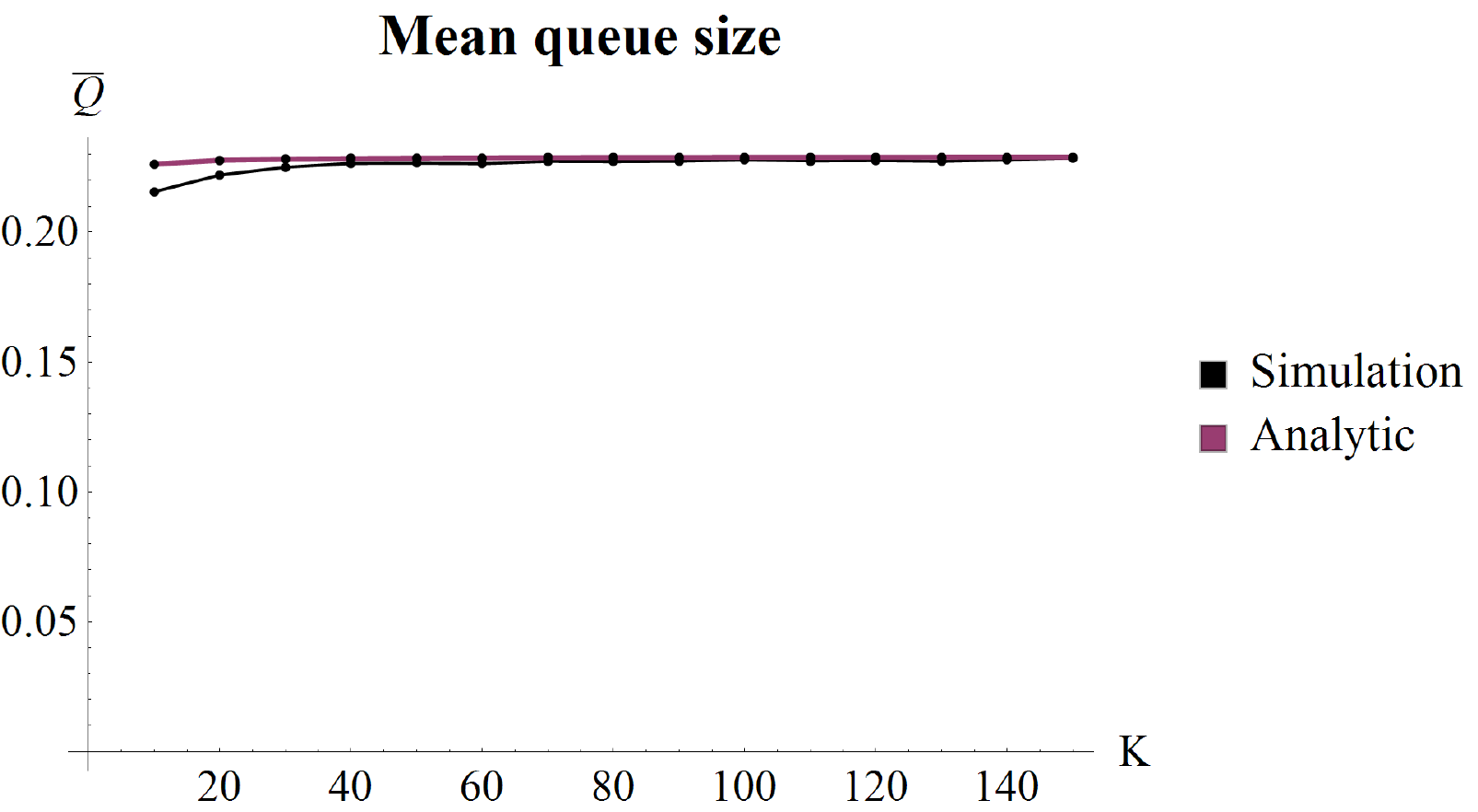}
                \caption{}
                \label{fig-MeanQueueSize_L=0.1}
        \end{subfigure}%
        ~
        \begin{subfigure}[b]{0.48\textwidth}
                \centering
                \includegraphics[width=\textwidth]{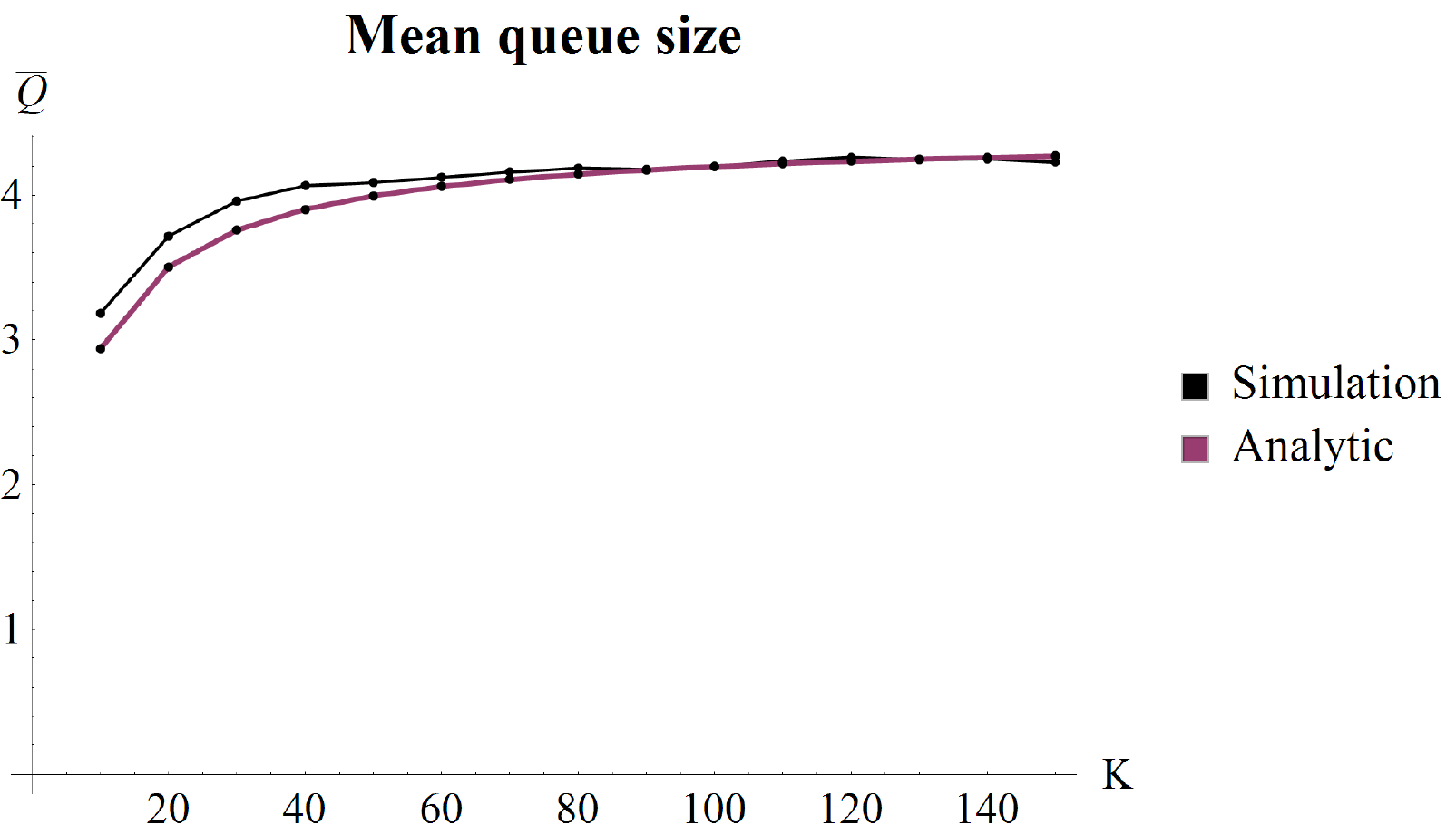}
                \caption{}
                \label{fig-MeanQueueSize_L=0.3}
        \end{subfigure}
        \hfil
        \begin{subfigure}[b]{0.48\textwidth}
                \includegraphics[width=\textwidth]{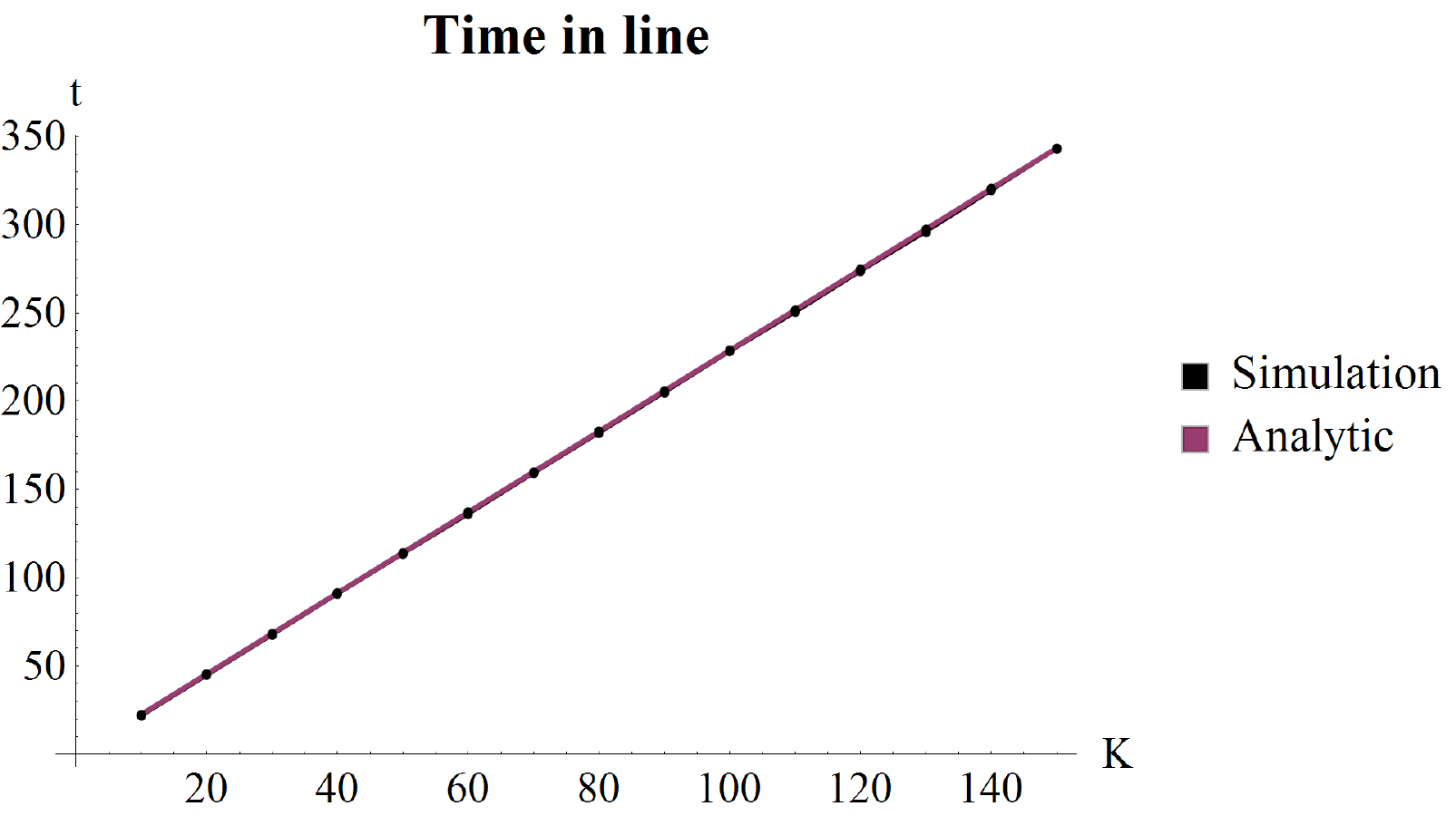}
                \caption{}
                \label{fig-TimeInLine_L=0.1}
        \end{subfigure}
        ~
        \begin{subfigure}[b]{0.48\textwidth}
                \includegraphics[width=\textwidth]{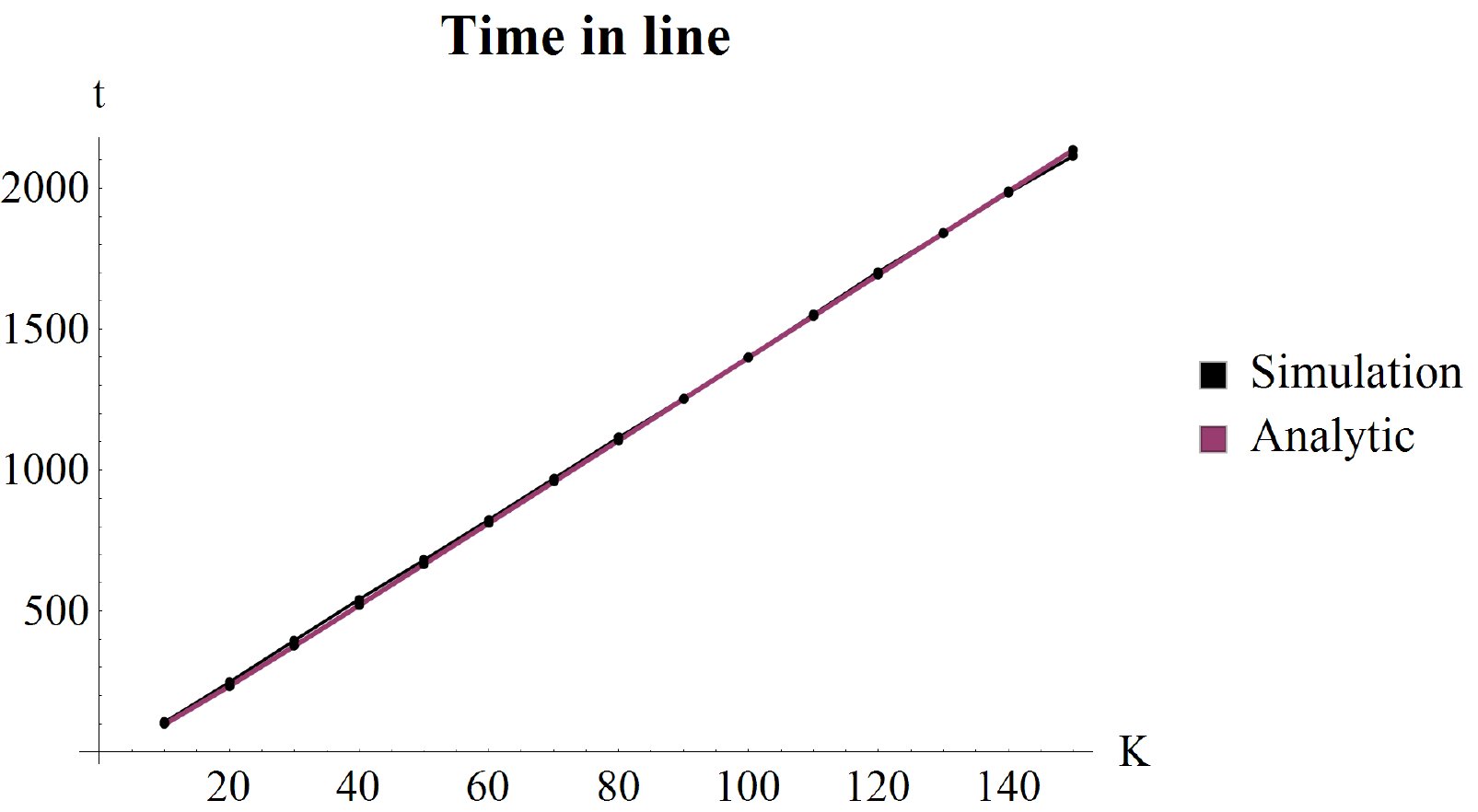}
                \caption{}
                \label{fig-TimeInLine_L=0.3}
        \end{subfigure}
        \caption[System performance time dependent model]{Queuing system performance results as a function of the number of users in a time dependent queue model. $\lambda_T=\{0.1,0.3\},\mu_g=0.7,\mu_b=0.5$}
        \label{fig-QueuingPerformance_time dependent}
    \end{figure}

\chapter{Conclusions}

    This work investigated the channel capacity and the performance of a multi-user MAC system in a time dependent environment under both centralized and distributed scheduling. Specifically, the expected channel capacity gain was derived in the case which the dependent capacity sequence was modeled as a stationary process characterized by the good bad channel Markov process. If the long range dependence condition $D(u_n)$ and the local dependence condition $D'(u_n)$ hold, the extreme laws are the same as if the sequence was \textit{i.i.d.} but with different normalizing constants. A distributed scheduling algorithm for this system was suggested as well, and we showed that both approaches have the same scaling laws, the distributed approach being smaller only by a factor of $e^{-1}$, so there was no loss of optimality due to the distributed algorithm.
    
    The performance of the system was derived while considering queueing theory aspects. In fact this precise characterization is a very difficult mission, which up until today was not solved. Therefore we presented approximation models to describe its behaviour. First we address to the \textit{i.i.d.} case where the users do not experience a time varying channel. For that case we elaborated an existing model and showed results for our paradigm. In addition we gave another approach, which relates the queues as independent, concerning the probability of collision in the random access mechanism, and enables us to consider them each as a much more simple queue. Lastly we tied the capacity part of this work with the queuing part when we suggested a queue model, which is time dependent and modeled by our good-bad channel model. We showed good agreement between the analytic model and the simulation results.
    
\appendix
    \chapter{Convergence to Extreme Value Distribution}\label{Appendix A}
        \section{Convergence proof}
        \begin{proof}
           In the following proof we give compliance for the convergence conditions and the derivation of $a_K$ and $b_K$.
           The stationary distribution  $F(t)=pF_g(t)+qF_b(t)$ as shown earlier has a negative derivative $f'$ from $x_0=\max\{\mu_g,\mu_b\}$ till $\infty$. So we only need to show that \eqref{equ-Sufficient type 1 condition} holds,
           \begin{equation*}
                \begin{aligned}
                    &\lim_{t \rightarrow \infty} \frac{f'(t)\left(1-F(t)\right)}{f^2(t)} =
                    \lim_{t \rightarrow \infty} \frac{\left(pf_g'(t)+qf_b'(t)\right)\left(1-\left(pF_g(t)+qF_b(t)\right)\right)}{\left(pf_g(t)+qf_b(t)\right)^2}\\
                    &=\lim_{t \rightarrow \infty} \frac{\frac{1}{2}\left(pf_g'(t)+qf_b'(t)\right)\left(pErfc\left(\frac{t-\mu_g}{\sqrt{2}\sigma_g}\right) +qErfc\left(\frac{t-\mu_b}{\sqrt{2}\sigma_b}\right)\right)}{\left(pf_g(t)+qf_b(t)\right)^2}\\
                \end{aligned}
               \end{equation*}
               In \cite{abramowitz2012handbook}, 7.1.13, we can find upper and lower bounds for the complementary error function,
               \begin{equation*}
                \frac{2}{\sqrt{\pi}}\frac{e^{-t^2}}{t+\sqrt{t^2+2}}< Erfc(t)\leq  \frac{2}{\sqrt{\pi}}\frac{e^{-t^2}}{t+\sqrt{t^2+\frac{4}{\pi}}}.
           \end{equation*}
           where these inequalities are true for $t>0$ which fits our case for $t \rightarrow \infty$. Using these bounds we will show with the sandwich rule that the limit above converge to $-1$. Let us consider first the lower bound of the complementary error function,
           \begin{equation*}
                 \begin{aligned}
                     &\lim_{t \rightarrow \infty} \frac{\frac{1}{2}\left(pf_g'(t)+qf_b'(t)\right)
                     \left( p\frac{2}{\sqrt{\pi}}\frac{e^{-\frac{(t-\mu_g)^2}{2\sigma_g^2}}}{\frac{t-\mu_g}{\sqrt{2}\sigma_g}+\sqrt{\frac{(t-\mu_g)^2}{2\sigma_g^2}+2}} + q\frac{2}{\sqrt{\pi}}\frac{e^{-\frac{(t-\mu_b)^2}{2\sigma_b^2}}}{\frac{t-\mu_b}{\sqrt{2}\sigma_b}+\sqrt{\frac{(t-\mu_b)^2}{2\sigma_b^2}+2}}\right)}
                     {\left(pf_g(t)+qf_b(t)\right)^2}\\
                     &=  \lim_{t \rightarrow \infty}\frac{1}{\sqrt{\pi}} \frac{\left(pf_g'(t)+qf_b'(t)\right)
                     \left( p\frac{\sqrt{2\pi}\sigma_g f_g(t)}{\frac{t-\mu_g}{\sqrt{2}\sigma_g}+\sqrt{\frac{(t-\mu_g)^2}{2\sigma_g^2}+2}} +
                            q\frac{\sqrt{2\pi}\sigma_b f_b(t)}{\frac{t-\mu_b}{\sqrt{2}\sigma_b}+\sqrt{\frac{(t-\mu_b)^2}{2\sigma_b^2}+2}}   \right)}
                     {\left(pf_g(t)+qf_b(t)\right)^2}\\
                     &=  \lim_{t \rightarrow \infty}-\sqrt{2} \frac{\left(pf_g(t)\frac{t-\mu_g}{\sigma_g^2}+qf_b(t)\frac{t-\mu_b}{\sigma_b^2}\right)
                     \left(p \frac{\sigma_g f_g(t)}{\frac{t-\mu_g}{\sqrt{2}\sigma_g}+\sqrt{\frac{(t-\mu_g)^2}{2\sigma_g^2}+2}} +
                           q \frac{\sigma_b f_b(t)}{\frac{t-\mu_b}{\sqrt{2}\sigma_b}+\sqrt{\frac{(t-\mu_b)^2}{2\sigma_b^2}+2}}   \right)}
                     {\left(pf_g(t)+qf_b(t)\right)^2}\\
                 \end{aligned}
           \end{equation*}
           The Limit above can be break to four different limits,
           \begin{equation*}
                 \begin{aligned}
                     &\lim_{t \rightarrow \infty}-\sqrt{2}p^2 \frac{f_g(t)\frac{t-\mu_g}{\sigma_g^2}
                      \frac{\sigma_g f_g(t)}{\frac{t-\mu_g}{\sqrt{2}\sigma_g}+\sqrt{\frac{(t-\mu_g)^2}{2\sigma_g^2}+2}}}
                     {\left(pf_g(t)+qf_b(t)\right)^2}+
                     \lim_{t \rightarrow \infty}-\sqrt{2}q^2 \frac{f_b(t)\frac{t-\mu_b}{\sigma_b^2}
                      \frac{\sigma_b f_b(t)}{\frac{t-\mu_b}{\sqrt{2}\sigma_b}+\sqrt{\frac{(t-\mu_b)^2}{2\sigma_b^2}+2}}}
                     {\left(pf_g(t)+qf_b(t)\right)^2}+\\
                     &\lim_{t \rightarrow \infty}-\sqrt{2}pq \frac{f_g(t)\frac{t-\mu_g}{\sigma_g^2}
                      \frac{\sigma_b f_b(t)}{\frac{t-\mu_b}{\sqrt{2}\sigma_b}+\sqrt{\frac{(t-\mu_b)^2}{2\sigma_b^2}+2}}}
                     {\left(pf_g(t)+qf_b(t)\right)^2}+
                     \lim_{t \rightarrow \infty}-\sqrt{2}pq \frac{f_b(t)\frac{t-\mu_b}{\sigma_b^2}
                      \frac{\sigma_g f_g(t)}{\frac{t-\mu_g}{\sqrt{2}\sigma_g}+\sqrt{\frac{(t-\mu_g)^2}{2\sigma_g^2}+2}}}
                     {\left(pf_g(t)+qf_b(t)\right)^2}\\
                 \end{aligned}
           \end{equation*}
           please notice that the first and the second limits are similar with the exception of their indexes, and so does the third and the fourth limit. We start with the first limit calculation,
           \begin{equation*}
                 \begin{aligned}
                     &\lim_{t \rightarrow \infty}-\sqrt{2}p^2 \ \frac{f_g(t) \ \frac{t-\mu_g}{\sigma_g^2} \
                      \frac{\sigma_g f_g(t)}{\frac{t-\mu_g}{\sqrt{2}\sigma_g}+\sqrt{\frac{(t-\mu_g)^2}{2\sigma_g^2}+2}}}
                     {\left(pf_g(t)+qf_b(t)\right)^2}
                     =\lim_{t \rightarrow \infty}-\sqrt{2} \ \frac{p^2f_g^2(t) \ (t-\mu_g)\sigma_g}
                     {\left(pf_g(t)+qf_b(t)\right)^2 \ \sigma_g^2 \ \frac{t-\mu_g}{\sqrt{2}\sigma_g}\left(1+\sqrt{1+\frac{4\sigma_g^2}{(t-\mu_g)^2}} \right) }\\
                     &=\lim_{t \rightarrow \infty} -2 \ \frac{p^2f_g^2(t)}
                     {\left(pf_g(t)+qf_b(t)\right)^2 \ \left(1+\sqrt{1+\frac{4\sigma_g^2}{(t-\mu_g)^2}} \right)}
                     =\lim_{t \rightarrow \infty} -2 \ \frac{p^2f_g^2(t)}
                     {\left(pf_g(t)+qf_b(t)\right)^2} \cdot
                     \lim_{t \rightarrow \infty} \frac{1}
                     {\left(1+\sqrt{1+\frac{4\sigma_g^2}{(t-\mu_g)^2}} \right)}\\
                     &=\lim_{t \rightarrow \infty} -2 \ \frac{p^2f_g^2(t)}
                     {\left(pf_g(t)+qf_b(t)\right)^2 } \cdot \frac{1}{2}
                     =\lim_{t \rightarrow \infty} - \ \frac{p^2f_g^2(t)}
                     {\left(pf_g(t)+qf_b(t)\right)^2}
                     =-\left(\lim_{t \rightarrow \infty} \frac{pf_g(t)}
                     {\left(pf_g(t)+qf_b(t)\right)}\right)^2\\
                     &\overset{(a)}{=}-\left(\lim_{t \rightarrow \infty} \frac{1}
                     {1+\frac{qf_b(t)}{pf_g(t)}}\right)^2
                     =-\left( \frac{1}
                     {1+\lim_{t \rightarrow \infty}\frac{qf_b(t)}{pf_g(t)}}\right)^2 \ =
                \end{aligned}
           \end{equation*}
           \begin{equation*}
                \begin{aligned}
                     &\Rightarrow \lim_{t \rightarrow \infty}\frac{qf_b(t)}{pf_g(t)}
                     =\frac{\sigma_g}{\sigma_b} \frac{q}{p} \lim_{t \rightarrow \infty} e^{-\frac{(t-\mu_b)^2}{2\sigma_b^2}+\frac{(t-\mu_g)^2}{2\sigma_g^2}}
                     = \frac{\sigma_g}{\sigma_b}\frac{q}{p} e^{\lim_{t \rightarrow \infty} -\frac{(t-\mu_b)^2}{2\sigma_b^2}+\frac{(t-\mu_g)^2}{2\sigma_g^2}}\\
                     &\Rightarrow \lim_{t \rightarrow \infty}\frac{t^2(\sigma_b^2-\sigma_g^2)+t(2\mu_b\sigma_g^2-2\mu_g\sigma_b^2)+C}{2\sigma_b^2\sigma_g^2}
                     = \left\{
                                \begin{array}{l l}
                                   \infty  & \quad \sigma_g^2 < \sigma_b^2\\
                                   -\infty  & \quad \sigma_g^2 \geq \sigma_b^2 \ \text{assuming} \ \mu_g>\mu_b
                                \end{array} \right.
                \end{aligned}
           \end{equation*}
           So,
           \begin{equation*}
                 \lim_{t \rightarrow \infty}\frac{qf_b(t)}{pf_g(t)} =
                        \left\{
                            \begin{array}{l l}
                               \infty  & \quad \sigma_g^2 < \sigma_b^2\\
                               0  & \quad \sigma_g^2 \geq \sigma_b^2 \ \text{assuming} \ \mu_g>\mu_b
                        \end{array} \right.
           \end{equation*}
           Note that in (a) we assume that $p \neq 0$. This assumption implies that the situation which all the users are in the bad group is not taken in consideration. For that case all the users have the same channel and the analysis is known and not in our interest. In the same manner we assume also that $q \neq 0$ for the opposite situation.
           Hence, the first limit result is
           \begin{equation*}
                 \lim_{t \rightarrow \infty}-\sqrt{2}p^2 \ \frac{f_g(t) \ \frac{t-\mu_g}{\sigma_g^2} \
                  \frac{\sigma_g f_g(t)}{\frac{t-\mu_g}{\sqrt{2}\sigma_g}+\sqrt{\frac{(t-\mu_g)^2}{2\sigma_g^2}+2}}}
                 {\left(pf_g(t)+qf_b(t)\right)^2}=
                        \left\{
                            \begin{array}{l l}
                               0  & \quad \sigma_g^2 < \sigma_b^2\\
                               -1  & \quad \sigma_g^2 \geq \sigma_b^2 \ \ \text{assuming} \ \mu_g>\mu_b
                        \end{array} \right.
           \end{equation*}
           As mentioned earlier the first and the second limits different only in their indexes, therefore the result for the second limit is
           \begin{equation*}
                 \lim_{t \rightarrow \infty}-\sqrt{2}q^2 \frac{f_b(t)\frac{t-\mu_b}{\sigma_b^2}
                  \frac{\sigma_b f_b(t)}{\frac{t-\mu_b}{\sqrt{2}\sigma_b}+\sqrt{\frac{(t-\mu_b)^2}{2\sigma_b^2}+2}}}
                 {\left(pf_g(t)+qf_b(t)\right)^2}=
                        \left\{
                            \begin{array}{l l}
                               -1  & \quad \sigma_g^2 < \sigma_b^2\\
                               0  & \quad \sigma_g^2 \geq \sigma_b^2 \ \ \text{assuming} \ \mu_g>\mu_b
                        \end{array} \right.
           \end{equation*}
           We turn now for the third limit calculation
           \begin{equation*}
                \begin{aligned}
                     &\lim_{t \rightarrow \infty}-\sqrt{2}pq \frac{f_g(t)\frac{t-\mu_b}{\sigma_g^2}
                      \frac{\sigma_b f_b(t)}{\frac{t-\mu_b}{\sqrt{2}\sigma_b}+\sqrt{\frac{(t-\mu_b)^2}{2\sigma_b^2}+2}}}
                     {\left(pf_g(t)+qf_b(t)\right)^2}
                     =\lim_{t \rightarrow \infty}-\sqrt{2} \ \frac{pqf_g(t) \ (t-\mu_g)\sigma_b \ f_b(t) \sqrt{2} \sigma_b}
                     {\left(pf_g(t)+qf_b(t)\right)^2 \ \sigma_g^2 \ (t-\mu_b) \ \left(1+\sqrt{1+\frac{4\sigma_b^2}{(t-\mu_b)^2}} \right) }\\
                     &=-2\frac{\sigma_b^2}{\sigma_g^2}\lim_{t \rightarrow \infty} \ \frac{pf_g(t)}{\left(pf_g(t)+qf_b(t)\right)} \cdot
                     \lim_{t \rightarrow \infty}\ \frac{qf_b(t)}{\left(pf_g(t)+qf_b(t)\right)} \cdot
                     \lim_{t \rightarrow \infty}\ \frac{(t-\mu_g)}{(t-\mu_b)}\cdot
                     \lim_{t \rightarrow \infty}\ \frac{1}
                     {\left(1+\sqrt{1+\frac{4\sigma_g^2}{(t-\mu_g)^2}} \right) }\\
                     &=-\frac{\sigma_b^2}{\sigma_g^2}\lim_{t \rightarrow \infty}\ \frac{pf_g(t)}{\left(pf_g(t)+qf_b(t)\right)} \cdot
                     \lim_{t \rightarrow \infty}\ \frac{qf_b(t)}{\left(pf_g(t)+qf_b(t)\right)}\\
                     &=-\frac{\sigma_b^2}{\sigma_g^2}
                     \left\{
                        \begin{array}{l l}
                             0  & \quad \sigma_g^2 < \sigma_b^2\\
                             1  & \quad \sigma_g^2 \geq \sigma_b^2 \ \ \text{assuming} \ \mu_g>\mu_b
                        \end{array} \right.
                     \cdot\left\{
                        \begin{array}{l l}
                             1  & \quad \sigma_g^2 < \sigma_b^2\\
                             0  & \quad \sigma_g^2 \geq \sigma_b^2 \ \ \text{assuming} \ \mu_g>\mu_b\
                        \end{array} \right.
                     =0
                 \end{aligned}
           \end{equation*}
           The fourth limit also share the same result, and if we add all the parts we can see that the limit converge to $-1$. If we consider the upper bound of the complementary error function we will find also that the limit convergence to $-1$ since the analytical development is the same and the limit
           \begin{equation*}
                 \lim_{t \rightarrow \infty} \frac{1}
                 {\left(1+\sqrt{1+\frac{4\sigma_g^2}{(t-\mu_g)^2}} \right)}
                 =\lim_{t \rightarrow \infty} \frac{1}
                 {\left(1+\sqrt{1+\frac{8\sigma_g^2}{\pi(t-\mu_g)^2}} \right)}
                 =\frac{1}{2}
           \end{equation*}
           Therefore we can conclude that  condition \eqref{equ-Sufficient type 1 condition} holds for our stationary distribution.
       \end{proof}

       We now show that the second condition \eqref{equ-Necessary and sufficient type 1 condition} also holds for the stationary distribution. Let us examine the expression:
       \begin{equation*}
            1-F(t)=p(1-F_g(t))+q(1-F_b(t))
       \end{equation*}
       Since $F_i(t)$ is a Gaussian distribution we shall use the asymptotic relation:
       \begin{equation}\label{equ-Asymptotic relation of Gaussian distribution}
            1-\Phi(t) \sim \frac{\phi(t)}{t} \quad \text{as } t \rightarrow \infty
       \end{equation}
       \begin{equation*}
            \begin{aligned}
                1-F(t)&=p\left(\frac{\sigma_g}{t-\mu_g}\phi(\frac{t-\mu_g}{\sigma_g}) \right)+q\left(\frac{\sigma_b}{t-\mu_b}\phi(\frac{t-\mu_b}{\sigma_b}) \right)      \\
                      &=\frac{1}{\sqrt{2\pi}}\left(\frac{p}{t-\mu_g}e^{-\frac{(t-\mu_g)^2}{2\sigma_g^2}} +\frac{q}{t-\mu_b}e^{-\frac{(t-\mu_b)^2}{2\sigma_b^2}} \right)  \\
                      &\overset{(a)}{=}\frac{1}{\sqrt{2\pi}}\frac{p}{t-\mu_g}e^{-\frac{(t-\mu_g)^2}{2\sigma_g^2}}(1+o(1)) \quad \text{as } t \rightarrow \infty
            \end{aligned}
       \end{equation*}
       where (a) is true since
       \begin{equation*}
       \lim_{t \rightarrow \infty} \frac{\frac{\frac{q}{t-\mu_b}e^{-\frac{(t-\mu_b)^2}{2\sigma_b^2}}}{\frac{p}{t-\mu_g}e^{-\frac{(t-\mu_g)^2}{2\sigma_g^2}}}}{1}=
       \lim_{t \rightarrow \infty} \frac{e^{-\frac{(t-\mu_b)^2}{2\sigma_b^2}}}{e^{-\frac{(t-\mu_g)^2}{2\sigma_g^2}}} = 0
       \end{equation*}
       assuming $\sigma_g>\sigma_b$.
       So taking in consideration condition \eqref{equ-Necessary and sufficient type 1 condition}:
       \begin{equation*}
            \begin{aligned}
                &\frac{1-F(t+xg(t))}{1-F(t)}=\frac{\frac{1}{\sqrt{2\pi}}\frac{p}{t+xg(t)-\mu_g}e^{-\frac{(t+xg(t)-\mu_g)^2}{2\sigma_g^2}}(1+o(1))}
                                             {\frac{1}{\sqrt{2\pi}}\frac{p}{t-\mu_g}e^{-\frac{(t-\mu_g)^2}{2\sigma_g^2}}(1+o(1))}                     \\
                &=\frac{t-\mu_g}{t+xg(t)-\mu_g} e^{\frac{-(t+xg(t)-\mu_g)^2+(t-\mu_g)^2}{2\sigma_g^2}} (1+o(1))            \\
                &=\frac{1}{1+\frac{xg(t)}{t-\mu_g}}e^{-\frac{g(t)x(t-\mu_g)}{\sigma_g^2}}e^{-\frac{g^2(t)x^2}{2\sigma_g^2}}(1+o(1))=\\
            \end{aligned}
       \end{equation*}
       By choosing $g(t)=\frac{\sigma_g^2}{t-\mu_g}$ as the strictly positive function for $t \rightarrow \infty$ we get
       \begin{equation*}
                =\frac{1}{1+\frac{x\sigma_g^2}{(t-\mu_g)^2}}e^{-x}e^{-\frac{\sigma_g^2 \ x^2}{2(t-\mu_g)^2}}(1+o(1))=e^{-x} \quad \quad \text{as } t \rightarrow \infty
       \end{equation*}
       That conclude that the distribution function $F(x)$ belongs to the domain of attraction of Type \Rmnum{1}. Similar analysis can be found in \cite{mladenovic1999extreme}, where some examples for convergence of sequences of independent random variables with the same mixed distribution is investigated.\\
       \section{Normalizing constants}
       Derivation of $a_K$ and $b_K$:\\
       From Theorem \ref{thm-EVT convergence different formulation } we know that $u_K$ is a sequence of real numbers such that $K(1-F(u_K))\rightarrow \uptau$ as $K\rightarrow \infty$, therefore in our case:
       \begin{equation*}
         1-pF_g(u_K)-qF_b(u_K)\rightarrow \frac{1}{K}e^{-x}, \quad K\rightarrow \infty
       \end{equation*}
       where $\uptau=e^{-x}$. The same way as the previous proof using \eqref{equ-Asymptotic relation of Gaussian distribution} we obtain that
       \begin{equation*}
         \left(\frac{p\sigma_g}{u_K-\mu_g}\phi(\frac{u_K-\mu_g}{\sigma_g}) \right)+\left(\frac{q\sigma_b}{u_K-\mu_b}\phi(\frac{u_K-\mu_b}{\sigma_b}) \right)\rightarrow \frac{1}{K}e^{-x}
       \end{equation*}
       \begin{equation*}
         \frac{1}{\sqrt{2\pi}}\frac{p}{u_K-\mu_g}e^{-\frac{(u_K-\mu_g)^2}{2\sigma_g^2}}(1+o(1)) \rightarrow \frac{1}{K}e^{-x}
       \end{equation*}
       the last step is true since $u_K \rightarrow \infty$ as $K \rightarrow \infty$, similar to the pervious proof.
       \begin{equation}\label{equ-Proof for a_n b_n (1)}
                 -\frac{1}{2}\log2\pi+\log p-\log{(u_K-\mu_g)}-\frac{(u_K-\mu_g)^2}{2\sigma_g^2}+\log K+x+o(1) \rightarrow 0
       \end{equation}
       It follows at once that $(t-\mu_g)^2 / 2\log K \rightarrow 1$, and hence
       \begin{equation*}
         \log{(u_K-\mu_g)}=\frac{1}{2}(\log 2 +\log{\log{K}})+o(1)
       \end{equation*}
       Putting this in \eqref{equ-Proof for a_n b_n (1)} ,we obtain
       \begin{equation*}
                 \frac{(u_K-\mu_g)^2}{2\sigma_g^2}=-\frac{1}{2}\log2\pi+\log p-\frac{1}{2}(\log 2 +\log{\log{K}})+\log K +x+o(1)
       \end{equation*}
       or
       \begin{equation*}
                 \frac{(u_K-\mu_g)^2}{\sigma_g^2}=
                  2\log K\left(1+\frac{x-\frac{1}{2}\log{\frac{4\pi}{p^2}}-\frac{1}{2}\log{\log{K}}}{\log K}+o\left(\frac{1}{\log K}\right)\right)
       \end{equation*}
       and hence
       \begin{equation*}
                 \frac{(u_K-\mu_g)}{\sigma_g}=
                  \sqrt{2\log K}\left(1+\frac{x-\frac{1}{2}\log{\frac{4\pi}{p^2}}-\frac{1}{2}\log{\log{K}}}{2\log K}+o\left(\frac{1}{\log K}\right)\right)
       \end{equation*}
       so by using expansion we have,
       \begin{equation*}
                 u_K=\sigma_g\sqrt{2\log K}
                  \left(1+\frac{x-\frac{1}{2}\log{\frac{4\pi}{p^2}}-\frac{1}{2}\log{\log{K}}}{2\log K}+o\left(\frac{1}{\log K}\right)\right)+\mu_g
       \end{equation*}
       since we know that $u_K=x/a_K+b_K$ we conclude that
       \begin{equation*}
           \begin{aligned}
                 &a_K=\frac{\sqrt{2\log{K}}}{\sigma_g}\\
                 &b_K=\sigma_g\left((2\log K)^{1/2}-\frac{\log{\log K}+\log{\frac{4\pi}{p^2}}}{2(2\log K)^{1/2}}\right)+\mu_g
           \end{aligned}
       \end{equation*}

\chapter{Capacity - Direct expression analysis}\label{Appendix B}

    \section{Capacity - analyzing change in group sizes}

        The system can be modeled in a different way, still taking in consideration our channel model, we can analyse the channel capacity by addressing separately to the different groups which are differ by their channel state and evaluate the maximum for each group separately. Therefore the system is modeled as a first order time-homogeneous Markov chain with $K+1$ states, each state represent The number of users which are in G state or in B state. In other words, The state $S_i(n)$ represent the system state in discrete time $n$, where $i$ users are in Bad channel state and $K-i$ users are in Good channel state.\\

        The Markov chain is irreducible and aperiodic chain. We can notice that it is possible to go with positive probability from any state of the Markov chain to any other state in a finite number of steps since the chain is a complete graph. And it is easy to see that since each state has positive probability to remain at the same state all the states are aperiodic.\\
        we will give now some properties for the system behavior which will be used later.

        \subsection{Transition probability matrix}

        The transition probability matrix $P$ is given:
        \begin{equation}\label{equ-probabilityMatrix}
        P_{ij}=
                \begin{cases}
                    \sum\limits_{n=0}^{{\min\{i,K-i,j,K - j\}}} {K-i \choose j-i+n}{i \choose n} \alpha^{j-i+n}(1-\alpha)^{K-j-n}\beta^{n}(1-\beta)^{i-n}    & \text{if } i \leq j \\
                    \sum\limits_{n=0}^{\min\{i,K-i,j,K - j\}} {K-i \choose n}{i \choose i-j+n} \alpha^{n}(1-\alpha)^{K-i-n}\beta^{i-j+n}(1-\beta)^{j-n}    & \text{if } i > j
                \end{cases}
                \text{for } 0 \leq i,j \leq K
        \end{equation}
        where $\alpha$ and $\beta$ represent the probability that user changes his channel state from G to B or the opposite as shown in Figure \ref{fig-GoodBadchannel}.\\

        The entries in the transition matrix are binomials probabilities, the probability to get from state $i$ to state $j$ depends on the number of users which in a given time $n$ changed their channel state. The probability $P_{ij}$ sums up all the probabilities of the cases that eventually lead to having $j$ users in bad state and $K-j$ users in good state. Each case describes the number of user transition from good to bad state with success probability $\alpha$ and the number of user transition from bad to good state with success probability $\beta$ in a given time $n$.\\
        \begin{proposition}
            The number of transition cases possible, which begins in state $i$ and ends in state $j$ in a given time $n$, are ${\min\{i,K-i,j,K - j\}}+1$.
        \end{proposition}
        \begin{proof}
            The initial state $i$ represents that there are $i$ users in the bad group and $K-i$ users in the good group and the target state $j$ represent that there are $j$ users in the bad group and $K-j$ users in the good group. For example, the number of possible transitions, for $K=6, i=2, j=3$, is 3; the first one is the transition of one user from the good group to the bad group and none from the bad group to the good group, the second is the transition of two users from the good group to the bad group and one user from the bad group to the good group and the third one is the transition of three users from the good group to the bad group and two users from the bad group to the good group. Another example, if $i$ or $j$ are $0$ or $K$ we would have only one possible transition.\\
            The expression $\min\{i,K-i\}$ defines the maximum transitions possible from the group regardless the target state $j$, there can't be more transition cases than the size of the smallest group. The expression $\min\{j,K-j\}$ defines the maximum transitions possible in order to obtain the new sizes of the groups, since again we cannot count more transition cases than the size of the smallest group. Finally the minimal value between the two should be chosen, so it would hold for both conditions. Hence, we get : $\min\{\min\{i,K-i\},\min\{j,K-j\}\} = \min\{i,K-i,j,K - j\}$
        \end{proof}

        Let us consider the assumption that the transition probabilities in the good-bad channel are symmetric ($\alpha$ equals $\beta$).\\
        The transition matrix $P_{ij}$ under the condition $\alpha=\beta$, \eqref{equ-probabilityMatrix}, can be written as follow:
        \begin{equation*}
        P_{ij}=
                \begin{cases}
                    \sum\limits_{n=0}^{\min\{i,K-i,j,K - j\}} {K-i \choose j-i+n}{i \choose n} \alpha^{j-i+2n}(1-\alpha)^{K-j+i-2n}    & \text{if } i \leq j \\
                    \sum\limits_{n=0}^{\min\{i,K-i,j,K - j\}} {K-i \choose n}{i \choose i-j+n} \alpha^{i-j+2n}(1-\alpha)^{K-i+j-2n}    & \text{if } i > j
                \end{cases}
                \text{for } 0 \leq i,j \leq K
        \end{equation*}
        Looking on the structure of the equation above we can notice that $P_{ij}=P_{(K-i)(K-j)}$ assuming $i\leq j$ and without loss of generality. In both entries we get the probability:
        \begin{equation*}
          \sum\limits_{n=0}^{\min\{i,K-i,j,K - j\}} {K-i \choose j-i+n}{i \choose n} \alpha^{j-i+2n}(1-\alpha)^{K-j+i-2n}
        \end{equation*}
        This applies also to the case where $i>j$.\\
        So the matrix has a structure as shown below:
        \begin{equation*}
        P_{i,j} =
            \begin{pmatrix}
                    a_{0,0} & a_{0,1}   & \cdots & a_{0,K-1} & a_{0,K} \\
                    a_{1,0} & a_{1,1}   & \cdots & a_{1,K-1} & a_{1,K} \\
                    \vdots  & \vdots    & \ddots & \vdots    & \vdots \\
                    a_{1,K} & a_{1,K-1} & \cdots & a_{1,1}   & a_{1,0} \\
                    a_{0,K} & a_{0,K-1} & \cdots & a_{0,1}   & a_{0,0}
            \end{pmatrix}
        \end{equation*}

        \begin{proposition}
            The transition matrix P is a Centrosymmetric (Cross-Symmetric) Matrix.
        \end{proposition}
        \begin{proof}
             An ($n\times n$) matrix $A = [ A_{i,j} ]$ is centrosymmetric when its entries satisfy $ A_{i,j} = A_{n-i+1,n-j+1}$ for $1 \leq i,j \leq n$ (\cite{Weaver1985Centrosymmetric}). In our case, $0 \leq i,j \leq n$ and the matrix is ($(n+1)\times (n+1)$), so we would get that the entries satisfy $ A_{i,j} = A_{n-i,n-j}$ for $0 \leq i,j \leq n$. This consist with the structure of the transition matrix explained above.
        \end{proof}

        Let us consider the ($n\times n$) exchange matrix $J$, which satisfies $I^2=J$ where $I$ is the identity matrix, and it entries defined such that:
        \begin{equation*}
        J_{i,j} =
          \begin{cases}
            1    & \text{if } j =n-i \\
            0    & \text{if } j \neq n-i
          \end{cases}
          \text{for } 0 \leq i,j \leq n
        \end{equation*}
        \begin{proposition}[\cite{Weaver1985Centrosymmetric}]
            An ($n \times n$) matrix A is centrosymmetric if and only if $JA = AJ$ .
            \label{proposition-centrosymmetric}
        \end{proposition}

    \subsection{Stationary distribution}

        If a Markov chain is a time-homogeneous, so that the process is described by a single, time-independent matrix $P_{ij}$ then the vector $\overline{\pi}$ is called a stationary distribution and it satisfies the global balance condition : $\overline{\pi}=\overline{\pi}P$.\\
        In the case that the finite state Markov chain is irreducible and aperiodic, then the stationary distribution exist and it is unique \cite{serfozo2009basics}.\\
        The characteristics mentioned above exist in our system model. We will investigate our stationary distribution $\overline{\pi}$.

        \begin{theorem}\label{thm-Staionary distribution is symmetric}
            The stationary distribution $\overline{\pi}$ of the Markov chain with transition matrix $P$ is a symmetric distribution around its middle.
        \end{theorem}
        \begin{proof}
            Let $\overline{\pi}$ be the unique stationary distribution of the Markov chain which hold the global balance condition:
            \begin{equation*}
              \begin{aligned}
                \overline{\pi}   & =\overline{\pi}P \\
                \overline{\pi}J  & \overset{(a)}{=}\overline{\pi}PJ  \\
                \overline{\pi^*} & \overset{(b)}{=}\overline{\pi}PJ  \\
                \overline{\pi^*}   & \overset{(c)}{=}\overline{\pi}JP  \\
                \overline{\pi^*}   & \overset{(d)}{=}\overline{\pi^*}P \\
              \end{aligned}
            \end{equation*}
         Where (a) since the equation hold we can to multiply both sides with the matrix $J$ as defined earlier, (b) any vector $\overline{v}$ multiplied by the matrix $J$ turns to be the same vector $\overline{v}$ in its revers order, which mean $v^*_i=v_{|v|-i}$, (c) $P$ is a centrosymmetric matrix and according to proposition \ref{proposition-centrosymmetric}, $PJ=JP$ (d) same as (b).\\
         As a result we can see that the vector $\overline{\pi^*}$ satisfies the global balance equation therefore in order not to contradict the singularity of the stationary distribution the obvious conclusion is that $\overline{\pi}=\overline{\pi^*}$. From this we can conclude that the stationary distribution is a symmetric distribution.
        \end{proof}

        A random variable $\xi$ and its distribution function $F(x)$ are said to be unimodal with mode $M$ if $F(x)$ is convex on $(-\infty,M)$ and concave on $(M,\infty)$. The notion of unimodality has proved to be very fruitful, therefore its analogs have been introduced for discrete and multivariate distributions. In \cite{keilson1971some} some results for discrete unimodality can be found, we will consider some of these results for better understanding the stationary distribution $\overline{\pi}$. A distribution $\{p_n\}$ with $k\in \mathbb{Z}$ as its support will be said to be unimodal if there exists at least one integer $M$ such that:
        \begin{equation*}
          \begin{aligned}
            &p_n     \geq p_{n-1},\ \text{ for all } n \leq M,\\
            &p_{n+1} \leq p_{n},  \quad \text{ for all } n \geq M.
          \end{aligned}
        \end{equation*}
        The distribution may consist of a single point mass, or it may be bounded on one side, both sides, or be unbounded. A discrete distribution $\{h_n\}$ is strongly unimodal if the convolution of $\{h_n\}$ with any unimodal $\{p_n\}$ is unimodal. In addition, a necessary and sufficient condition that $\{h_n\}$ is strongly unimodal is that (\cite{keilson1971some}, \textit{Theorem 3}):
        \begin{equation*}
          {h_n}^2 \geq h_{n-1}h_{n+1} \quad \text{  for all }n
        \end{equation*}
        Another theorem presented is: The limit of a convergent sequence of strongly unimodal discrete distributions is strongly unimodal (\cite{keilson1971some}, \textit{Theorem 1}).\\
        Also, it has been proven that any binomial distribution is strongly unimodal since it is a convolution of a finite number of Bernoulli distributions which are strongly unimodal. This derives from the \emph{Lemma} given in \cite{keilson1971some}, that the convolution of two strongly unimodal discrete distributions is strongly unimodal.

        \begin{theorem}\label{thm-Staionary distribution is discrete strongly unimodal}
            The stationary distribution $\overline{\pi}$ of the Markov chain with transition matrix $P$ is a discrete strongly unimodal distribution.
        \end{theorem}
        \begin{proof}
            Since the stationary distribution is defined as the limit distribution in an ergodic Markov chain (\cite{serfozo2009basics}, \textit{Theorem 59}), according to (\cite{keilson1971some}, \textit{Theorem 1}), in order to prove that $\overline{\pi}$ is a discrete strongly unimodal distribution it is sufficient to prove that each row of the transition matrix $P$ is strongly unimodal.\\
            Let us consider the following situation to describe the system model. Given the system state is $i$ ($i$ users in the bad group and $K-i$ users in the good group), we define two binomial random variables, $X\sim B(K-i,\alpha)$ and $Y\sim B(i,\alpha)$ which represent the number of users that traversed from the good group to the bad group and from the bad group to the good group in a time unit, respectively. Therefore, the distribution for the next system state $J$, i.e. the row $i$ of the transition matrix $P$ is:
            \begin{equation*}
              Pr(J=j)=Pr(i+X-Y=j)=Pr(X-Y=j-i)
            \end{equation*}
            The next system state is ruled from the given state $i$ the amount of users that traversed to the bad group minus the users that traversed to the good group.
            The difference between the two binomial random variables $X$ and $Y$ gives the next state distribution, hence we only need to show that this distribution is strongly unimodal. In the case of sum of two independent binomial random variables it is very simple, since each of them has strongly unimodal distribution and the distribution of the sum is the convolution of the two strongly unimodal distributions which is strongly unimodal. Let us define a new random variable $W$ such that $W=-Y$. The distribution of $W$ is defined:
            \begin{equation*}
              Pr(W=-k)= \binom {i} {k}\alpha^{k}(1-\alpha)^{i-k} \quad \text{  for }k=0,1,2,...,i
            \end{equation*}
            This distribution is clearly strongly unimodal since the necessary and sufficient condition in (\cite{keilson1971some}, \textit{Theorem 3}) holds. For $Y$, from its binomial characteristics, the inequality $p_{k}^2 \geq p_{k+1}p_{k-1}$ holds. Since $p(Y=k)=p(W=-k)$, $p(Y=k+1)=p(W=-k-1)$ and $p(Y=k-1)=p(W=-k+1)$, the inequality $p_{-k}^2 \geq p_{-k+1}p_{-k-1}$ for $W$ holds as well.\\
            Therefore, we get that the distribution:
            \begin{equation*}
              Pr(J=j)=Pr(X-Y=j-i)=Pr(X+W=j-i)
            \end{equation*}
            is strongly unimodal, since $X$ and $W$ each has a strongly unimodal distribution, the convolution between them will also give a strongly unimodal distribution \cite{keilson1971some}.
        \end{proof}

        From Theorems \ref{thm-Staionary distribution is symmetric} and \ref{thm-Staionary distribution is discrete strongly unimodal}, we can conclude that the stationary distribution of the Markov chain is a symmetric unimodal function. See Figure \ref{fig-StationaryDistribution_a} for the case of symmetric transition probabilities.\\
        We will derive the capacity for this scenario given the current state of the system, while this state determines the amount of users experiencing the different channel distributions.

          \begin{figure}
            \centering
            \begin{subfigure}[b]{0.3\textwidth}
                    \centering
                    \includegraphics[width=\textwidth]{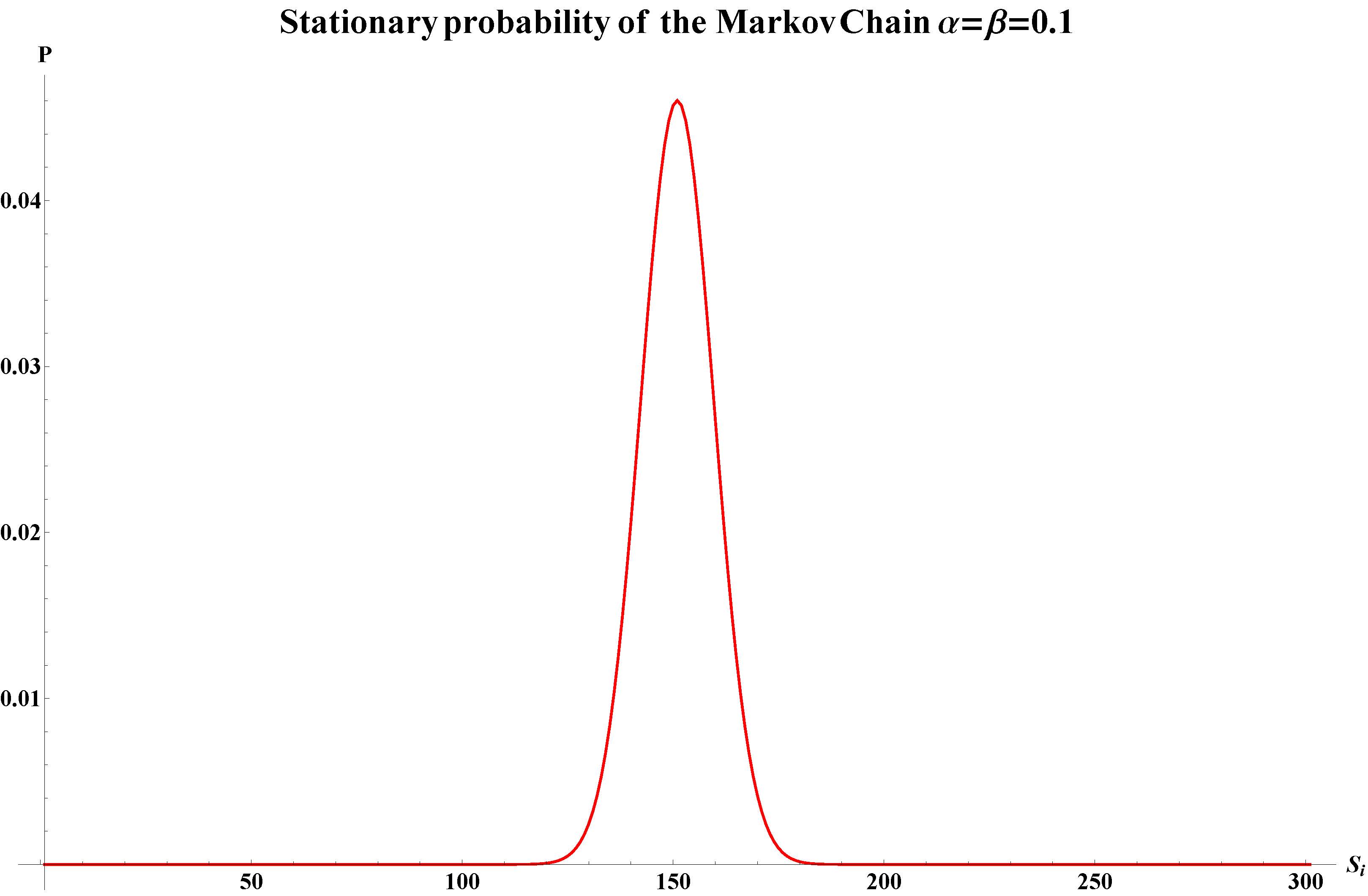}
                    \caption{}
                    \label{fig-StationaryDistribution_a}
            \end{subfigure}%
            ~ 
            \begin{subfigure}[b]{0.3\textwidth}
                    \includegraphics[width=\textwidth]{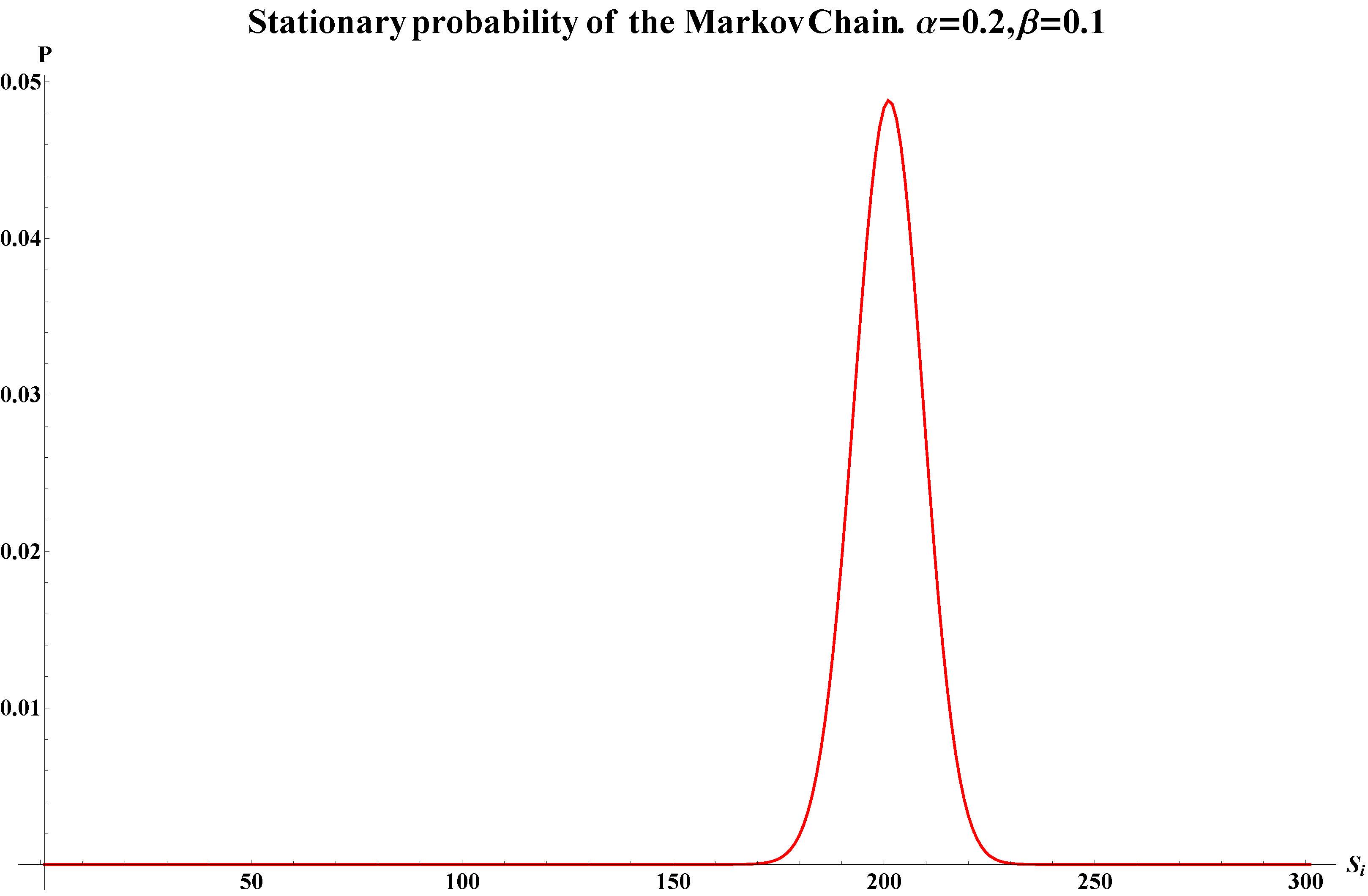}
                    \caption{}
                    \label{fig-StationaryDistribution_b}
            \end{subfigure}
            ~ 
            \begin{subfigure}[b]{0.3\textwidth}
                    \includegraphics[width=\textwidth]{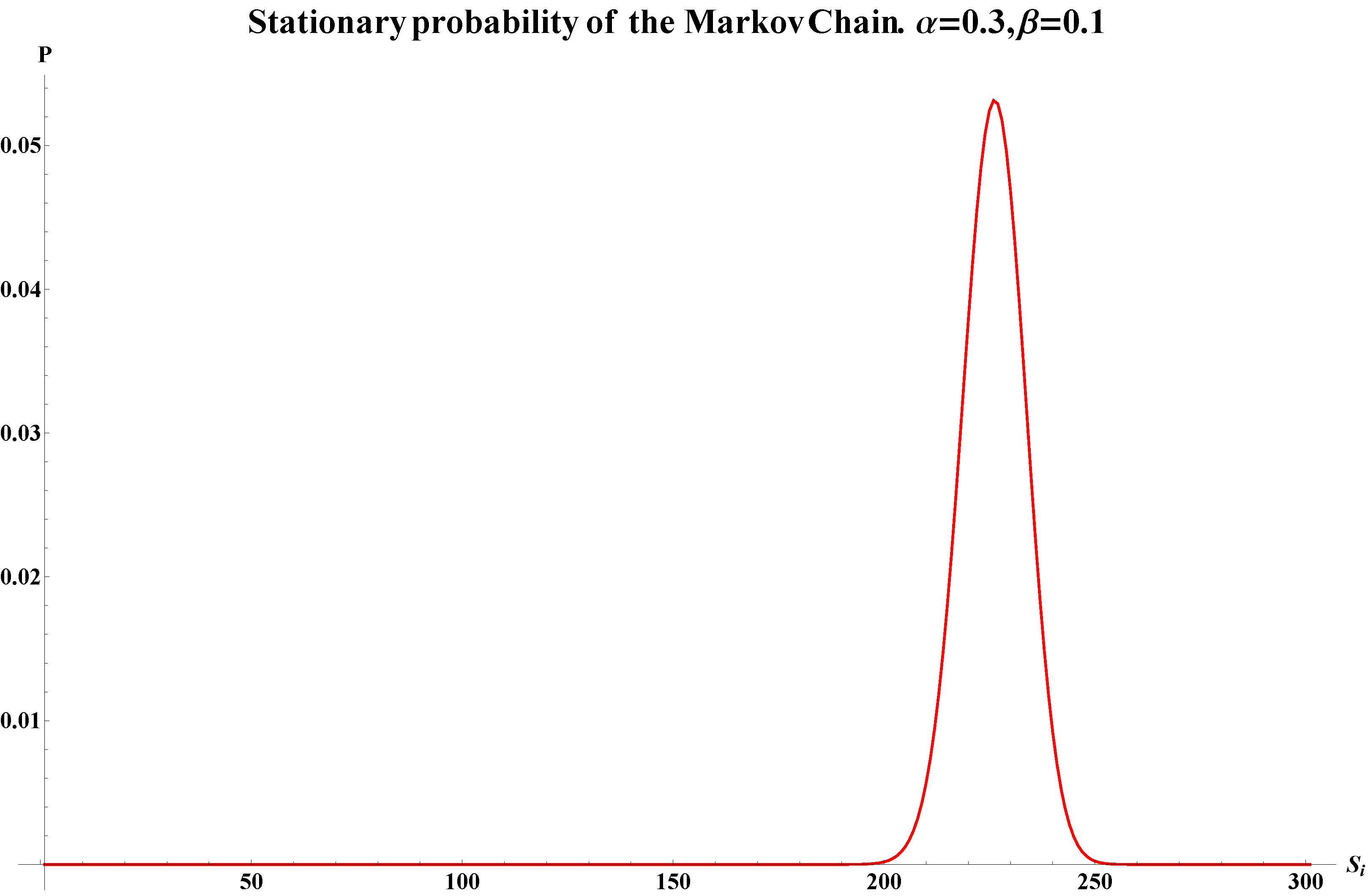}
                    \caption{}
                    \label{fig-StationaryDistribution}
            \end{subfigure}
            \caption[Stationary distribution for MC system model]{Stationary distribution of the Markov chain system model. Where the transition probabilities in (a)$\alpha=\beta=0.1$ (b)$\alpha=0.2\ \beta=0.1$ (c)$\alpha=0.3\ \beta=0.1$}\label{fig-StationaryDistribution_c}
        \end{figure}

        \subsection{Capacity analyze}

        We consider the problem of finding $\widetilde{M_K}$ as defined in \eqref{equ-CapacityDefinition} and the distribution of $\{C_i\}$ as defined \eqref{equ-UserCapacityProcess}.
        As mentioned earlier at first we will assume that the system is in a specific state $S_i(n)$ and therefore we know the amount of users which are in good and bad channel state. It is important to understand that given the system state $S_i(n)$, there are $i$ users in Bad channel state which we denote by $K_1$, and $K-i$ users in Good channel state which we denote by $K_2$, note that $K=K_1+K_2$.\\
        In order to find the expected capacity we shall derive the capacity distribution of $\widetilde{M_K}$:
        \begin{equation}\label{equ-CapacityDistribution}
          \begin{aligned}
            &P_r\big(\widetilde{M_K}<x\big)=\sum \limits_{i=0}^K P_r\big(\widetilde{M_K}<x\;\big|\; S_i(n)\big)P\big(S_i(n)\big)=\\
            &\sum \limits_{i=0}^K P_r\big(\max\big\{C_1(n),C_2(n),...,C_K(n)\big\}<x\;|\; S_i(n)\big)P\big(S_i(n)\big)
          \end{aligned}
        \end{equation}
        Given $S_i(n)$ we get two independent groups $A=\{a_1,...,a_{K_1}\}$ and $B=\{b_1,...,b_{K_2}\}$ where $a_j$ and $b_j$ represents the users' capacities according to the group each user belongs to. Since every state of the system need to be taken in consideration, we cannot use the EVT for all the group sizes. Let us define a parameter $\varphi$ for the minimal group size which we can use EVT to achieve the asymptotic distribution for the maximum capacity. Assuming $A$ and $B$ are big enough, the distribution of the maximum capacity in each group is the asymptotic distribution $G_j(x,K_i)$ as in equation \eqref{equ-EVTMaxDistribution}, where the parameters $a_K$ and $b_K$ will be a function of $K_1$ and $K_2$. And for the states where one group is too small for us to use EVT we will treat in a straightforward manner.\\
        The distribution given the system state $S_i(n)$:
        \begin{equation*}
          \begin{aligned}
            &P_r\big(\widetilde{M_K}<x\;\big|\; S_i(n)\big)=P_r\big(\max\big\{C_1(n),C_2(n),...,C_K(n)\big\}<x\;\big|\; S_i(n)\big)\\
            &=P_r(\max\big\{\max\{A\},\max\{B\}\big\}<x\;\big|\; S_i(n))\\
            &=P_r(\max\{A\}<x,\max\{B\}<x\;\big|\; S_i(n))\\
            &\overset{(a)}{=}P_r\big(\max\{A\}<x\;\big|\; S_i(n)\big)P_r\big(\max\{B\}<x\;\big|\; S_i(n)\big)\\
          \end{aligned}
        \end{equation*}
        where (a) is due to the independency between the groups. So \eqref{equ-CapacityDistribution} can be written:
        \begin{equation}\label{equ-CapacityDistribution-cont}
            \begin{aligned}
            P_r\big(\widetilde{M_K}<x\big) &=\sum \limits_{i=0}^K P_r\big(\max\{A\}<x\;\big|\; S_i(n)\big)P_r\big(\max\{B\}<x\;\big|\; S_i(n)\big)P(S_i(n))\\
                                           &=\sum \limits_{i=0}^\varphi F_B^{K_1}(x)G_2(x,K_2)+\sum \limits_{i=\varphi+1}^{K-\varphi} G_1(x,K_1)G_2(x,K_2)+\sum \limits_{i=K-\varphi+1}^K G_1(x,K_1)F_G^{K_2}(x)
            \end{aligned}
        \end{equation}

        Thus the expected channel capacity is:
        \begin{equation}\label{equ-ExpectedCapacity}
            \begin{aligned}
                E\Big[\widetilde{M_K}\Big] &=\sum \limits_{i=0}^K E\Big[\widetilde{M_K}\;\big|\; S_i(n)\Big]P(S_i(n))\\
                &=\sum \limits_{i=0}^\varphi P(S_i(n))\int \limits_{-\infty}^{\infty} x\frac{\partial}{\partial x}\Big[F_B^{K_1}(x)G_2(x,K_2)\Big]\, dx\\
                &+\sum \limits_{i=\varphi+1}^{K-\varphi} P(S_i(n))\int \limits_{-\infty}^{\infty} x\frac{\partial}{\partial x}\Big[G_1(x,K_1)G_2(x,K_2)\Big]\, dx\\
                &+\sum \limits_{i=K-\varphi+1}^K P(S_i(n))\int \limits_{-\infty}^{\infty} x\frac{\partial}{\partial x}\Big[G_1(x,K_1)F_G^{K_2}(x)\Big]\, dx\\
            \end{aligned}
        \end{equation}

         In Figure \ref{fig-CapacityDistributionTwoGroups} we can see the capacity PDF in the case $N_1=150$ and $N_2=150$ where each group experience different channel distribution with $\mu_1=1.25$,$\mu_2=0.25$ and $\sigma_1=\sigma_2=0.19$ respectively.
        \begin{figure}[h]
            \centering
            \includegraphics[width=0.6\textwidth]{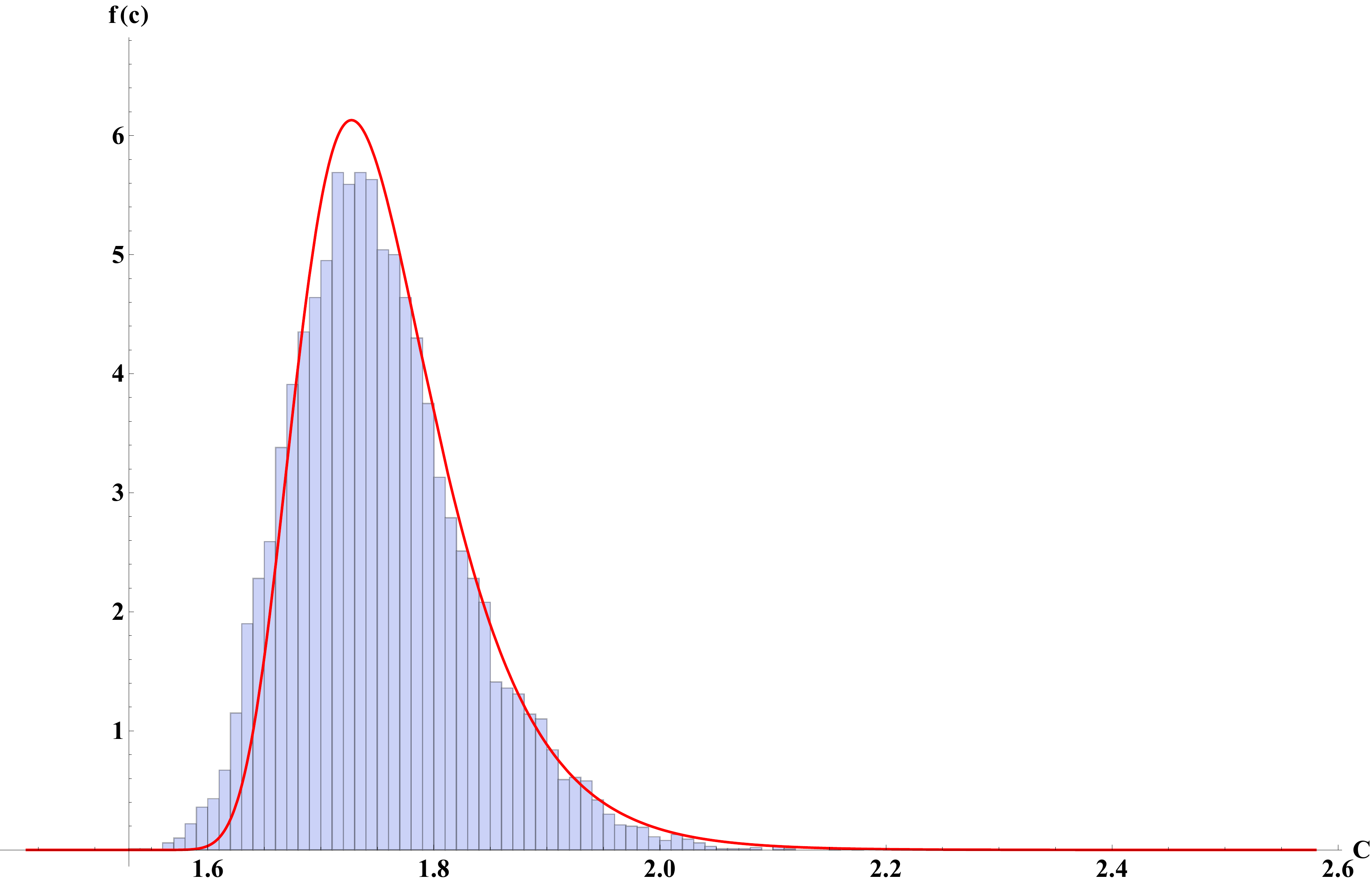}
            \caption[Capacity distribution - two groups]{Capacity Distribution with Two Groups distributed differently }
            \label{fig-CapacityDistributionTwoGroups}
        \end{figure}

        A more realistic simulation is in Figure \ref{fig-CapacityDistributionTwoGroupsVariable}, we can see the maximal MIMO capacity among 300 users, with time variable size for each group (bad or good), and time variable state for each user according to the good bad channel Markov model. Achieved by saving a state vector for the users and updating it following the transition probabilities for each drawn.
        \begin{figure}
            \centering
            \begin{subfigure}[b]{0.5\textwidth}
                    \centering
                    \includegraphics[width=\textwidth]{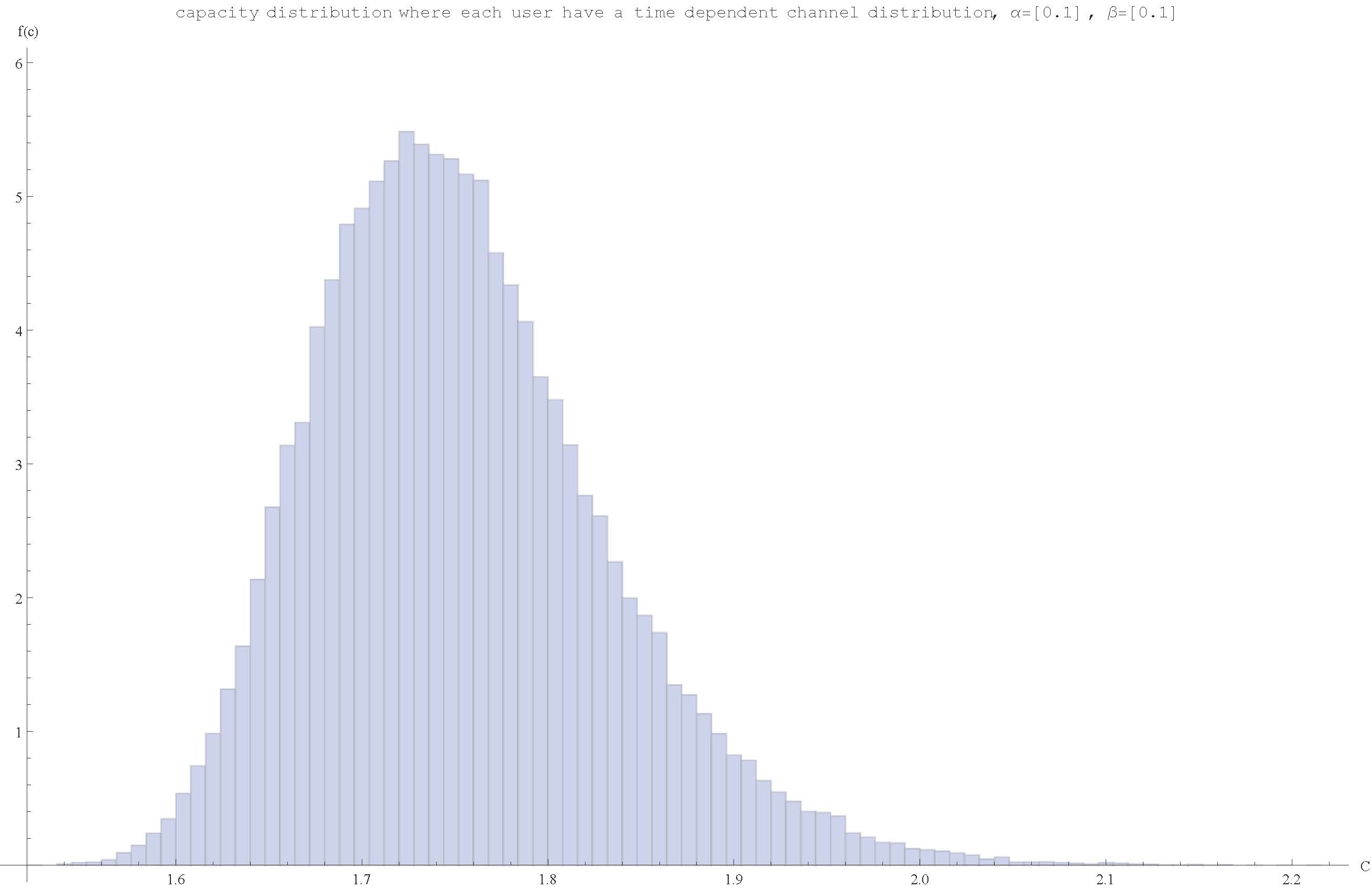}
                    \caption{}
            \end{subfigure}%
            ~ 
            \begin{subfigure}[b]{0.5\textwidth}
                    \includegraphics[width=\textwidth]{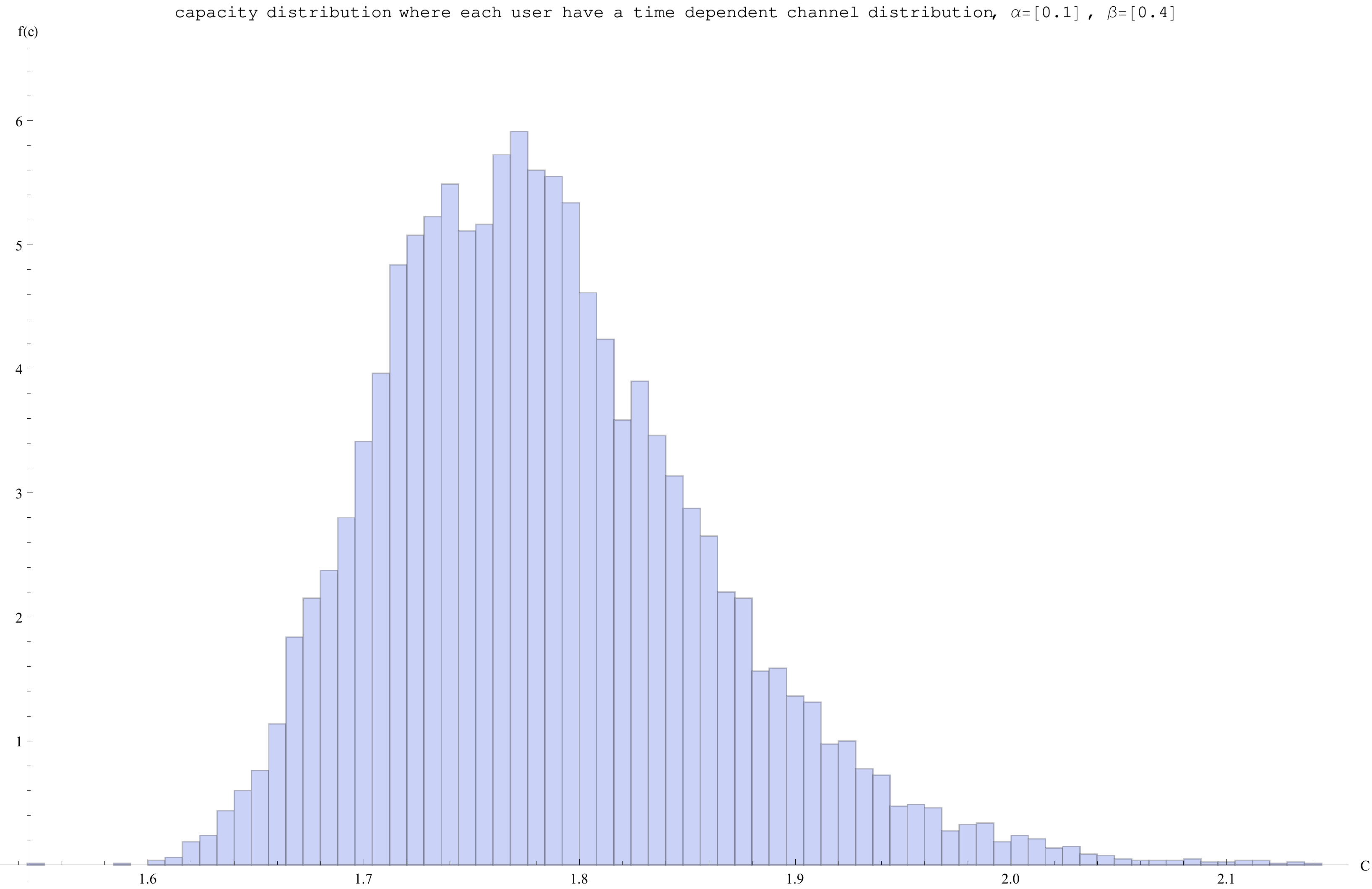}
                    \caption{}
            \end{subfigure}
            \caption[Capacity distribution - two varies groups]{Capacity distribution for variable size of Groups distributed differently. Where the transition probabilities in (a)$\alpha=\beta=0.1$ (b)$\alpha=0.1\ \beta=0.4$ }\label{fig-CapacityDistributionTwoGroupsVariable}
        \end{figure}

    \subsection{Capacity Bounds}
        Finding a closed solution for the expected channel capacity for the expression given in \eqref{equ-ExpectedCapacity} is very complicated, because each $C_i(n)$ is determined according to a Markov process. One can think of some upper and lower bounds for the expected  channel capacity in the analytical development in equation \eqref{equ-ExpectedCapacity}, the first naive lower bound would be to choose $i$ such that gives the maximum probability from the stationary distribution $\overline{\pi}$ of the system model (figure \ref{fig-StationaryDistribution}) in order to simplify the expression in \eqref{equ-ExpectedCapacity}, once we know the specific $i$ we know the parameters $K_1$ and $K_2$:
        \begin{equation}\label{equ-CapacityBound1}
            \begin{aligned}
                E\Big[\widetilde{M_K}\Big] &\overset{(a)}{=}\sum \limits_{i=0}^K E\Big[\widetilde{M_K}\;\big|\; \pi_i\Big]\; \pi_i\\
                                           &\overset{(b)}{\geq} \max_i\{\pi_i\}\; E\Big[\widetilde{M_K}\;\big|\; \pi_i\Big]
            \end{aligned}
        \end{equation}
        where (a) is from equation \eqref{equ-ExpectedCapacity} and (b) is to choose the maximum probability in $\overline{\pi}$, note that $P(S_i(n))=\pi_i$ .

        The index $i$ which satisfy the $\max_i\{\pi_i\}$ in the case where the transition probabilities in the good-bad channel are symmetric is $i=\frac{K}{2}$ as shown earlier in the analysis of the stationary distribution of the system.

        Since $i=\frac{K}{2}$  which mean that the groups are still big as $K \rightarrow \infty$, its fair to say that the EVT can be applied on the distribution of the two groups.\\
        Therefore, we can achieve a closed bound for this case by setting $i=\frac{K}{2}$ in equation \eqref{equ-CapacityBound1}:
        \begin{equation}\label{equ-CapacityBound1-closed}
            \begin{aligned}
                E\Big[\widetilde{M_K}\Big] &\geq \pi_{\frac{K}{2}}\; E\Big[\widetilde{M_K}\;\big|\; \pi_{\frac{K}{2}}\Big]\\
                                           &= {\Huge\pi}_{\frac{K}{2}} \int \limits_{-\infty}^{\infty} x\frac{\partial}{\partial x}\Big[G_1(x,\frac{K}{2})G_2(x,\frac{K}{2})\Big]\, dx\\
                                           &\overset{(a)}{=} {\Huge\pi}_{\frac{K}{2}} \int \limits_{-\infty}^{\infty} x\frac{\partial}{\partial x}\Big[\big(G_1(x,\frac{K}{2})\big)^2\Big]\, dx\\
                                           &= {\Huge\pi}_{\frac{K}{2}} \int \limits_{-\infty}^{\infty} x\frac{\partial}{\partial x}\Big[\Big(e^{-e^{-\frac{x-b_{\frac{K}{2}}}{a_{\frac{K}{2}}}}}\Big)^2\Big]\, dx\\
                                           &= 2{\Huge\pi}_{\frac{K}{2}} \int \limits_{-\infty}^{\infty} x e^{-e^{-\frac{x-b_{\frac{K}{2}}}{a_{\frac{K}{2}}}}} e^{-\frac{x-b_{\frac{K}{2}}}{a_{\frac{K}{2}}}} \frac{1}{a_{\frac{K}{2}}}\, dx\\
                                           &\overset{(b)}{=} 2{\Huge\pi}_{\frac{K}{2}} \int \limits_{0}^{\infty} \big(a_{\frac{K}{2}}\ln\big(\frac{1}{y}\big)+ b_{\frac{K}{2}} \big)e^{-y}\, dy\\
                                           &=2{\Huge\pi}_{\frac{K}{2}}\big(a_{\frac{K}{2}} \gamma + b_{\frac{K}{2}} \big)
            \end{aligned}
        \end{equation}
        where $a_{\frac{K}{2}}$ and $b_{\frac{K}{2}}$ as defined in \eqref{equ-parameter a_n} and \eqref{equ-parameter b_n}.\\
        The first lower bound given above can be improved. In this bound we only took in consideration the maximum probability from the sum of probabilities in equation \eqref{equ-CapacityDistribution-cont}. As mentioned we Know that the two groups sizes are $O(K/2)$, hence EVT can be used, and another bound can be taken in consideration which includes only the states that represent similar group sizes. Let us consider the situation that there is more good users than bad users, i.e. $\frac{K}{2}-\delta$ bad users and $\frac{K}{2}+\delta$ good users, where $\delta$ can be chosen as we please in order to maximize the offered bound. Now if we discard $2\delta$ users from the good group we would get a lower bound for the expected channel capacity, since the users who give the most contribution are the good users, so it is clear that if we reduce these users we reduce the expected channel capacity.
        \begin{equation}\label{equ-CapacityBound2}
            \begin{aligned}
                E\Big[\widetilde{M_K}\Big] &=\sum \limits_{i=0}^K E\Big[\widetilde{M_K}\;\big|\; \pi_i\Big]\; \pi_i\\
                                           &\geq  \sum \limits_{i=\frac{K}{2}-\delta}^{\frac{K}{2}} E\Big[\widetilde{M_K}\;\big|\; \pi_i\Big]\; \pi_i\\
                                           &=  \sum \limits_{i=\frac{K}{2}-\delta}^{\frac{K}{2}} E\left[G_1(x,K_1)G_2(x,K_2) \right] \pi_i\\
                                           &\geq  \sum \limits_{i=\frac{K}{2}-\delta}^{\frac{K}{2}} E\left[G_1(x,K_1)G_2(x,K_1) \right] \pi_i\\
                                           &=  \sum \limits_{i=\frac{K}{2}-\delta}^{\frac{K}{2}} \left(a_{K_1}\gamma+b_{K_1}\right) \pi_i\\
                                           &\geq  \left(a_{\frac{K}{2}-\delta}\gamma+b_{\frac{K}{2}-\delta}\right) \pi_{\frac{K}{2}-\delta}\delta\\
            \end{aligned}
        \end{equation}

\chapter{Random selection Proof}\label{Appendix C}

    Let $N$ be a Poisson process with intensity $\lambda$ and $N^\ast$ be a point process obtained from $N$ such that each event is removed or retained independently with probabilities $1-p$ and $p$ respectively. So for any Borel set $B$,
    \begin{equation*}
      \begin{aligned}
        P(N^\ast(B)=k)&=\sum_{i=k}^{\infty} P(N(B)=i)P(N^\ast(B)=k|N(B)=i)\\
                      &=\sum_{i=k}^{\infty} \frac{e^{-\lambda(B)}(\lambda(B))^i}{i!} \binom{i}{k}p^k(1-p)^{i-k}\\
                      &= \frac{e^{-p\lambda(B)}(p\lambda(B))^k}{k!}
      \end{aligned}
    \end{equation*}
    so that $N^\ast(B)$ is a Poisson random variable with mean $p\lambda(B)$.

\chapter{System-status chain transition probabilities}\label{Appendix D}

    Here we show the calculation of the transition probabilities of the system chain. We use the guide lines in \cite{ephremides1987delay} to simplify things and use also their notions but under our case which state that $r_i=s_i=q_i=p_i$. Due to this fact there are many more possible transition is the state space and therefore the following calculations would be more complicated. In the same manner the calculations are based on the change in the number of active users, thus the transitions can be classified into four types as follows:
    \begin{enumerate}
      \item \underline{The number of active users transits from 0 to 1:}\\
      Only a blocked user may become active, and only one, say $j$. All other blocked blocked must not exceed the threshold. In addition all idle users which may receive a package must not exceed the threshold and therefore, say $w$, becomes blocked. The latter will repeat it self during this calculation. The probability of such a transition is expressed by:
      \begin{equation*}
        P(\Delta A=1, \Delta B=-1+w, \Delta I=-w)=\prod^{K-n-w-1}_{k}\overline{\lambda}_k\prod^{w}_{l}\lambda_l\overline{p}_l
        \prod^{n+1}_{i}\overline{p}_i\frac{p_j}{\overline{p}_j}\left[(1-P_j(0|2))+\lambda_jP_j(0|2)\right]
      \end{equation*}
      where the product is on the sets of users by their transition and $n$ is the number of blocked users after the transition. Note that the user $j$ will be active if his queue is not empty after the successful transmission or a package arrived in the beginning of the slot. Also note that $\overline{p}$ and $\overline{\lambda}$ stand for $1-p$ and $1-\lambda$, respectively.

      \item \underline{The number of active users remains 1:}\\
      This case divides to two subcases, the same active user $j$ remains active or another blocked user $s$ becomes active. Nevertheless $w$ idle users may still become blocked.
      \begin{enumerate}
        \item the probability for the first case is:
        \begin{equation*}
            P(\Delta A=0, \Delta B=+w, \Delta I=-w,j\rightarrow j)= \prod^{K-n-w-1}_{k}\overline{\lambda}_k\prod^{w}_{l}\lambda_l\overline{p}_l
            \prod^{n+1}_{i}\overline{p}_i\frac{p_j}{\overline{p}_j}\left[P_j(1|1)+\lambda_j(1-P_j(1|1))\right]
      \end{equation*}
        \item the probability for the second case is:
        \begin{equation*}
            P(\Delta A=0, \Delta B=+w, \Delta I=-w,j\rightarrow s)= \\ \prod^{K-n-w-1}_{k}\overline{\lambda}_k\prod^{w}_{l}\lambda_l\overline{p}_l
            \prod^{n+1}_{i}\overline{p}_i\frac{p_s}{\overline{p}_s}\left[P_s(1|1)+\lambda_s(1-P_s(1|1))\right]
      \end{equation*}
      \end{enumerate}

      \item \underline{The number of active users transits from 1 to 0:}\\
      This case is subdivided to three subcases,
      \begin{enumerate}
        \item The first subcase is that the active user $j$ becomes idle:
        \begin{equation*}
            P(\Delta A=-1, \Delta B=+w, \Delta I=-w)= \prod^{K-n-w-1}_{k}\overline{\lambda}_k\prod^{w}_{l}\lambda_l\overline{p}_l
            \prod^{n+1}_{i}\overline{p}_i\frac{p_j}{\overline{p}_j}\left[\overline{\lambda}_j(1-P_j(1|1))\right]
        \end{equation*}
        \item The second subcase is that the active user $j$ becomes blocked and blocked user $s$ becomes idle:
        \begin{equation*}
            P(\Delta A=-1, \Delta B=+1-1+w, \Delta I=+1-w)= \prod^{K-n-w-1}_{k}\overline{\lambda}_k\prod^{w}_{l}\lambda_l\overline{p}_l
            \prod^{n+1}_{i}\overline{p}_i\frac{p_s}{\overline{p}_s}\left[\overline{\lambda}_s(1-P_s(1|1))\right]
        \end{equation*}
        \item The third subcase is that the active user $j$ becomes blocked:
        \begin{equation*}
        \begin{aligned}
            &P(\Delta A=-1, \Delta B=+1+w, \Delta I=-w)= \\
            &\prod^{K-n-w-1}_{k}\overline{\lambda}_k\prod^{w}_{l}\lambda_l\overline{p}_l\prod^{n+1}_{i}\overline{p}_i\\
            +&\prod^{K-n-w-1}_{k}\overline{\lambda}_k\prod^{w}_{l}\lambda_l\overline{p}_l \sum^{K-n-w-1}_{s}\frac{\lambda_sp_s}{\overline{\lambda}_s}  \prod^{n+1}_{i}\overline{p}_i\\
            +&\prod^{K-n-w-1}_{k}\overline{\lambda}_k\prod^{w}_{l}\lambda_l\left( 1-\prod^{w}_{s}\overline{p}_s \right)\left(1-\prod^{n}_{i}\overline{p}_i\right)\\
            +&\prod^{K-n-w-1}_{k}\overline{\lambda}_k\prod^{w}_{l}\lambda_l\overline{p}_l\left[ 1-\prod^{n}_{s}\overline{p}_s - \overline{p}_j\prod^{n}_{i}\overline{p}_i\sum^{n}_{q}\frac{p_q}{\overline{p}_q} \right]\\
            +&\prod^{K-n-w-1}_{k}\overline{\lambda}_k\prod^{n}_{i}\overline{p}_i\prod^{w}_{l}\lambda_l\left[ 1-\prod^{w}_{s}\overline{p}_s - \overline{p}_j\prod^{w}_{b}\overline{p}_b\sum^{w}_{q}\frac{p_q}{\overline{p}_q} \right]
        \end{aligned}
        \end{equation*}
        The first expression is for the case which non of the users exceeds the threshold. The second is for the case which one idle user managed to successfully transmits. The third case describes the case which one or more from the blocked and the idle groups tries to transmit, therefore what happens with user $j$ is meaningless. The forth expression describes the situation which user $j$ collided with one or more users from the blocked group, or collision happened between the blocked users and $j$ didn't exceed the threshold. The fifth expression is the same as the fourth concerning the group of idle users which receives a package.
      \end{enumerate}

      \item \underline{The number of active users remains 0 :}\\
      This case is subdivided to three subcases,
      \begin{enumerate}
        \item The first subcase is that all users maintain their status without change:
        \begin{equation*}
            P(\Delta B=0, \Delta I=0)= \left[ 1-\prod^{n}_{i}\overline{p}_i\sum^{n}_{q}\frac{p_q}{\overline{p}_q} \right]\prod^{K-n}_{k}\overline{\lambda}_k + \prod^{n}_{i}\overline{p}_i\prod^{K-n}_{k}\overline{\lambda}_k \sum^{K-n}_{q}\frac{\lambda_qp_q}{\overline{\lambda}_q}
        \end{equation*}
        where the first term is the probability that no packet arrives at the idle users while no blocked user, or at least two blocked users, transmit, and the second term is the probability that no blocked user exceeds the threshold while only one of the idle users successfully transmits.
        \item The second subcase is that one blocked user $j$ becomes idle:
        \begin{equation*}
            P(\Delta B=-1+w, \Delta I=+1-w)= \prod^{K-n-w-1}_{k}\overline{\lambda}_k\prod^{w}_{l}\lambda_l\overline{p}_l
            \prod^{n+1}_{i}\overline{p}_i\left[\overline{\lambda}_j\frac{p_j}{\overline{p}_j}P_j(0|2)\right]
        \end{equation*}
        \item The third subcase describes some situation for $w$ idle users becomes blocked:
        \begin{equation*}
        \begin{aligned}
            &P(\Delta B=+w, \Delta I=-w)= \\
            &\prod^{K-n-w}_{k}\overline{\lambda}_k\prod^{w}_{l}\lambda_l\overline{p}_l\left[ 1-\prod^{n}_{i}\overline{p}_i\sum^{n}_{q}\frac{p_q}{\overline{p}_q} \right]\\
            +&\prod^{K-n-w}_{k}\overline{\lambda}_k\prod^{w}_{l}\lambda_l\left[ 1-\prod^{w}_{s}\overline{p}_s -\prod^{w}_{b}\overline{p}_b\sum^{w}_{q}\frac{p_q}{\overline{p}_q} \right]\\
            +&\prod^{K-n-w}_{k}\overline{\lambda}_k\prod^{w}_{l}\lambda_l\overline{p}_l\sum^{w}_{q}\frac{p_q}{\overline{p}_q} \left[ 1-\prod^{n}_{i}\overline{p}_i\right]\\
            +&\prod^{K-n-w}_{k}\overline{\lambda}_k\sum^{K-n-w}_{s}\frac{\lambda_sp_s}{\overline{\lambda}_s} \prod^{w}_{l}\lambda_l\overline{p}_l\prod^{n}_{i}\overline{p}_i
        \end{aligned}
        \end{equation*}
        The first expression is for the case which non of the $w$ users exceeds the threshold and no blocked user, or at least two blocked users, transmit, the second is for the case which at least two users from $w$ exceeds the threshold, the third case describes collision between one of the users from $w$ with at least one from the blocked users and the forth expression describes the situation which one idle user succussed to transmit while all the other does't exceed the threshold.
      \end{enumerate}
    \end{enumerate}
    According to these transition probabilities we can calculate the steady state of the chain given the auxiliary quantities $p_i(1|1)$ and $p_i(0|2)$, the probability of exceedance and the users arrival rates.

\chapter{Average success probabilities}\label{Appendix E}

    Once the steady state of the system chain is known, the average success probabilities can be calculated. We do so by calculate the values of $P_B(i),P_A(i)$ and $P_I(i)$ in the same manner as calculated in \cite{ephremides1987delay} while taking in consideration that user must exceed the threshold in order to transmit.
    \begin{enumerate}
      \item For Blocked user $i$:
      \begin{equation*}
        \begin{aligned}
        P_B(i)&=Pr(\text{user $i$ success $\mid$ user $i$ is blocked})
              &=\frac{Pr(\text{user $i$ success and user $i$ is blocked})}{Pr(\text{user $i$ is blocked})}
        \end{aligned}
      \end{equation*}
      where
      \begin{equation*}
        Pr(\text{user $i$ success and user $i$ is blocked})=p_i\sum_{\substack{S_i=2 \\ j\neq i \\ S_j=0,1,2}}P(S_1,...,S_K)\prod_{j\neq i} (\lambda_j\overline{p}_j+\overline{\lambda}_j)^{\delta_{S_j=0}}(\overline{p}_j)^{\delta_{S_j=1}} (\overline{p}_j)^{\delta_{S_j=2}}
      \end{equation*}
      The exponent $\delta$ equals $1$ when it's condition holds.
      \begin{equation*}
        Pr(\text{user $i$ is blocked})=\sum_{S_i=2}P(S_1,..,S_i,..,S_K)
      \end{equation*}

      \item For Active user $i$:
      \begin{equation*}
        \begin{aligned}
        P_A(i)&=Pr(\text{user $i$ success $\mid$ user $i$ is active})
              &=\frac{Pr(\text{user $i$ success and user $i$ is active})}{Pr(\text{user $i$ is active})}
        \end{aligned}
      \end{equation*}
      where
      \begin{equation*}
        Pr(\text{user $i$ success and user $i$ is active})=p_i\sum_{\substack{S_i=1 \\ j\neq i \\ S_j=0,2}}P(S_1,...,S_K)\prod_{j\neq i} (\lambda_j\overline{p}_j+\overline{\lambda}_j)^{\delta_{S_j=0}}(\overline{p}_j)^{\delta_{S_j=2}}
      \end{equation*}
      the difference here is due to the fact that no more than one user may be active.
      \begin{equation*}
        Pr(\text{user $i$ is active})=\sum_{S_i=1}P(S_1,..,S_i,..,S_K)
      \end{equation*}

      \item For Idle user $i$:
      \begin{equation*}
        \begin{aligned}
        P_I(i)&=Pr(\text{user $i$ success $\mid$ user $i$ is idle})
              &=\frac{Pr(\text{user $i$ success and user $i$ is idle})}{Pr(\text{user $i$ is idle})}
        \end{aligned}
      \end{equation*}
      where
      \begin{equation*}
        Pr(\text{user i success and user $i$ is idle})=\lambda_ip_i\sum_{\substack{S_i=0 \\ j\neq i \\ S_j=0,1,2}}P(S_1,...,S_K)\prod_{j\neq i} (\lambda_j\overline{p}_j+\overline{\lambda}_j)^{\delta_{S_j=0}}(\overline{p}_j)^{\delta_{S_j=1}} (\overline{p}_j)^{\delta_{S_j=2}}
      \end{equation*}
      and
      \begin{equation*}
        Pr(\text{user $i$ is idle})=\sum_{S_i=0}P(S_1,..,S_i,..,S_K)
      \end{equation*}
    \end{enumerate}
    Using the above success probability it is not difficult to attain the expression of \eqref{equ-probability for blocked and empty}, \eqref{equ-probability for idle and empty},\eqref{equ-probability for active and not empty},\eqref{equ-probability to be blocked},\eqref{equ-probability to be unblocked} from the steady state of the queue-length Markov chain. After knowing the these values the boundary conditions can be calculated:
    \begin{equation*}
      \begin{aligned}
        P_i(1\mid 1)=&1-\frac{\pi(1,1)}{G^i_1(1)-\pi(1,0)} \\
        P_i(0\mid 2)=&\frac{\pi(0,0)}{G^i_0(1)}
      \end{aligned}
    \end{equation*}

\thispagestyle{empty}
\newpage \bibliographystyle{unsrt}
\bibliography{My_bib}

\end{document}